\documentclass[12pt]{article}

\usepackage[english]{babel}
\usepackage[utf8]{inputenc}
\usepackage[T1]{fontenc}
\usepackage[margin=25mm]{geometry}
\usepackage{graphicx}
\usepackage{indentfirst}
\usepackage{chngcntr}
\usepackage{amssymb, amsmath, amsthm, esint, mathrsfs, thmtools, etoolbox, trivfloat, enumitem}
\usepackage{fancybox, shadow, color}
\usepackage{longtable}
\usepackage[bookmarks=true,hidelinks]{hyperref}
\usepackage[open,openlevel=0]{bookmark}

\newtheorem{teo}{Theorem}[section]
\newtheorem{cor}{Corollary}[section]
\newtheorem{lema}{Lemma}[section]
\newtheorem{pro}{Proposition}[section]

\newtheorem{obs}{Remark}[section]

\newtheorem*{proof1}{Proof of Lemma~\ref{21.}}
\newtheorem*{proof2}{Proof of Lemma~\ref{28.}}
\numberwithin{equation}{section}

\BeforeBeginEnvironment{defi}{\vspace{1em}}
\AfterEndEnvironment{defi}{\vspace{0.05em}}
\BeforeBeginEnvironment{teo}{\vspace{1em}}
\AfterEndEnvironment{teo}{\vspace{0.05em}}
\BeforeBeginEnvironment{cor}{\vspace{1em}}
\AfterEndEnvironment{cor}{\vspace{0.05em}}
\BeforeBeginEnvironment{lema}{\vspace{1em}}
\AfterEndEnvironment{lema}{\vspace{0.05em}}
\BeforeBeginEnvironment{pro}{\vspace{1em}}
\AfterEndEnvironment{pro}{\vspace{0.05em}}
\BeforeBeginEnvironment{ex}{\vspace{1em}}
\AfterEndEnvironment{ex}{\vspace{0.05em}}
\BeforeBeginEnvironment{obs}{\vspace{1em}}
\AfterEndEnvironment{obs}{\vspace{0.05em}}

\definecolor{blue}{rgb}{0,0,1}
\definecolor{red}{rgb}{1,0,0}
\title{The Oracle Theorem for Matrix-Valued Jacobi Operators}
\author{Silas L. Carvalho, Douglas C. P. Freitas}
\date{}

\begin{document}
\maketitle

\begin{abstract}
This paper develops the matrix-valued analogue of the reflectionless and oracle framework for Jacobi operators. Starting from matrix-valued Weyl--Titchmarsh $m$-functions on the Siegel upper half-space, we study the distance-decreasing action of transfer matrices, matrix-valued harmonic measures and value distribution convergence. These ingredients are then used to establish the reflectionless property of the $\omega$-limit set and to prove an Oracle Theorem for matrix-valued Jacobi operators with absolutely continuous spectrum of full multiplicity. 
\end{abstract}

\tableofcontents
\newpage

\section{Introduction}\label{sec:introduction}

\subsection*{Context, motivation and objectives}

The spectral theory of discrete one-dimensional operators, particularly Schr\"odinger and Jacobi operators, has experienced profound advancements over the last decades [\ref{DamanikSurvey}, \ref{KillipSimon},  \ref{LastSimon}, \ref{remling}, \ref{Teschl}]. A central problem in this field consists in understanding the intricate relationship between the asymptotic, spatial behavior of a given potential and the nature of its associated spectrum.

From a mathematical standpoint, a celebrated milestone in this area is the work of C. Remling [\ref{remling}]. Remling established, in particular, a remarkable rigidity property for scalar half-line operators of the form $H^V=-\Delta+V$: if the operator has nonempty absolutely continuous spectrum, its right-limit potentials must be reflectionless on that spectrum. More specifically, if
\[
\omega(V)=\{W\in \mathcal{V}^C\mid \text{there exists a sequence } n_j\to\infty \text{ such that } d(S^{n_j}V,W)\to 0\}
\]
stands for the $\omega$-limit set of $V$, where $S$ denotes the shift map, i.e., $(S^kV)(n)=V(n+k)$, and $d$ stands for the metric
\[
d(V,W)=\sum_{n\in A\cap B}2^{-|n|}|V(n)-W(n)|,
\]
with $V,W \in\mathcal{V}^C$ (the space of bounded potentials $\Vert V\Vert_\infty:=\sup_{n\in\mathbb{Z}}|V(n)|\leqslant C$), then Remling showed that if $|\Sigma_{ac}|>0$, then 
\[
\omega(V)\subset\mathcal{R}(\Sigma_{ac}),
\]
where $\Sigma_{ac}$ is the essential support of the absolutely continuous part of the spectral measure of $H^V$ and
\[
\mathcal{R}(\Sigma_{ac})=\{W\in\bigcup_{C>0}\mathcal{V}^C\mid \text{$W$ is reflectionless on $\Sigma_{ac}$}\}, 
\]
that is, $m_+(t)=-\overline{m_-(t)}$ for a.e. $t\in \Sigma_{ac}$, where $m_\pm$ stand for the half-line Weyl-Titchmarsh $m$-functions of $H^W$.

A striking and highly counterintuitive consequence of this rigidity is the scalar Oracle Theorem. It states that the future values of a bounded potential can be deterministically predicted, up to an arbitrary degree of precision, based solely on a finite window of its past values, provided that the absolutely continuous spectrum is nonempty. More specifically, for each $\epsilon>0$ and each $C>0$, there exist $L\in\mathbb{N}$ and a smooth function
\[
\Delta:[-C,C]^{L+1}\to [-C,C],
\]
called \textit{an oracle}, such that the following holds: for each potential $V$ (half-line or not) such that $||V||_{\infty}\leqslant C$ and $\Sigma_{ac}(V)\supset A$ (where $A$ is a Borel set with positive Lebesgue measure), there exists $n_0 \in \mathbb{N}$ such that for each $n\geqslant n_0$,
\[
|V(n+1)-\Delta (V(n-L), V(n-L+1), \ldots , V(n))|<\epsilon.
\]
The primary motivation of this work is to extend this powerful rigidity theory from the scalar setting to the broader and structurally more complex class of matrix-valued Jacobi operators.

For the inner product in $\mathbb{C}^l$, we adopt the convention of linearity with respect to the first component,
\[
\langle\mathbf{u},\mathbf{v}\rangle_{\mathbb{C}^l}:=\sum_{k=1}^{l}u_k\overline{v}_k.
\]
The Hilbert space of square-summable $\mathbb{C}^l$-valued sequences is denoted by $\ell^2(\mathbb{Z},\mathbb{C}^l)$. A matrix-valued Jacobi operator is a symmetric operator on this space defined by the action
\begin{equation}\label{fabricio2.1}
    (J\mathbf{u})_n:=D_{n-1}\mathbf{u}_{n-1}+V_n\mathbf{u}_n+D_n\mathbf{u}_{n+1},
\end{equation}
where $(D_n)_{n\in\mathbb{Z}}$ and $(V_n)_{n\in\mathbb{Z}}$ are sequences of self-adjoint matrices in $M(l,\mathbb{R})$ such that, for each $n\in\mathbb{Z}$, $D_n$ is invertible. In a broader setting, one can consider $(V_n)$ and $(D_n)$ as sequences of Hermitian matrices (see [\ref{MJ}] for details), but here one needs to require more: each $D_n$ must be positive definite. Moreover, we require that the sequences $(V_n)$ and $(D_n)$ satisfy
\begin{equation}\label{CSP}
   0<\inf_{n\in\mathbb{Z}}s_l(D_n)\le \sup_{n\in\mathbb{Z}}s_1(D_n)<\infty\qquad\text{and}\qquad 0\le\sup_{n\in\mathbb{Z}}s_1(V_n)<\infty.
\end{equation}
Here $s_k(B)$ stands for the $k$-th singular value of the $l\times l$ real matrix $B$. The particular case where, for each $n\in\mathbb{Z}$, $D_n=I$, corresponds to the so-called matrix-valued Schr\"odinger operator.

Before addressing the rigidity phenomena, it is essential to contextualize how the spectral theory of these operators has been explored in the literature. As comprehensively survey by Marx and Jitomirskaya in [\ref{MJ}], the study of quasi-periodic Schr\"odinger and Jacobi operators has increasingly adopted a global perspective, heavily relying on the transfer matrix formalism and the behavior of Lyapunov exponents. In this matrix-valued setting, which naturally arises from physical properties like Aubry-Andre duality (see [\ref{MJ}] for details), the structural complexity needs a generalized approach to characterize the spectrum. Building upon this dynamical framework, recent advances by Oliveira and Carvalho [\ref{fabricioTese}, \ref{fabricioSilasA}, \ref{FabricioSilas1}] have successfully extended crucial spectral characterizations to ergodic matrix-valued Jacobi operators.



Our main goal is to construct a universal oracle for matrix-valued Jacobi operators. In order to achieve this, we prove that the fundamental connection between the absolutely continuous spectrum and reflectionless potentials is robust enough to survive in the matrix setting. Specifically, we show that the $\omega$-limit set of a bounded matrix potential consists entirely of reflectionless matrix potentials on its absolutely continuous spectrum of multiplicity $l$; see Theorem~\ref{27}. 

While preparing this manuscript, we became aware that a related Remling-type result had been obtained by Acharya; see~\cite{Acharya2}. Acharya's result is formulated in the setting of matrix-valued discrete Schrödinger operators. In the present paper, we prove the corresponding result for matrix-valued Jacobi operators, using somewhat different ideas. More precisely, we introduce matrix-valued analogues of the objects used by Remling in his proof, whereas Acharya's approach retains the scalar structures from Remling's original argument. We also establish an Oracle Theorem for matrix-valued Jacobi operators; see Theorem~\ref{98}. Finally, we discuss some consequences of this theorem, which are not addressed in~\cite{Acharya2}.




\subsection*{Addressed problems and results}


We define the space of potentials $\mathcal{V}^C$ by imposing strict bounds not only on the norms of the diagonal matrices $V_n$, but specifically on the singular values of the off-diagonal matrices $D_n$. By equipping this space with a weighted product metric $\rho_\infty$, defined for any two potential sequences $V = \{(D_n^V, V_n^V)\}_{n \in \mathbb{Z}}$ and $W = \{(D_n^W, V_n^W)\}_{n \in \mathbb{Z}}$ in $\mathcal{V}^C$ as
\[
\rho_\infty(V, W) = \sum_{n=-\infty}^{\infty} 2^{-|n|} \left( \rho(D_n^V, D_n^W) + \rho(V_n^V, V_n^W) \right),
\]
where $\rho$ represents the metric induced by the Frobenius norm, one guarantees that $\mathcal{V}^C$ is a compact metric space.

This topological compactness allows us to study the asymptotic behavior of potentials under the shift map $S$; recall that the $\omega$-limit set of a bounded potential $V$ is given by
\[
\omega(V) = \{W \in \mathcal{V}^C\mid \exists \ n_j \to \infty \text{ such that } \rho_\infty(S^{n_j}V, W) \to 0\}.
\]
This guarantees that $\omega(V)$ is non-empty and compact (see Proposition~\ref{34}). 

In order to properly investigate these operators and their asymptotic properties, we decompose the whole-line matrix-valued Jacobi operator $J$, defined by~\eqref{fabricio2.1}, into half-line restrictions, denoted by $J_+$ and $J_-$. This decoupling is paramount because the spectral analysis, and the reflectionless property itself, relies entirely on comparing the distinct asymptotic behaviors at $+\infty$ and $-\infty$. By analyzing the restrictions $J_{\pm}$ on $\ell^2(\mathbb{Z}_{\pm}, \mathbb{C}^l)$, one can construct the corresponding half-line Weyl-Titchmarsh $m$-functions, which act as the fundamental parameterizations of the resolvent operators for these independent restricted systems.

Consequently, the adaptation of the Titchmarsh-Weyl theory for matrix-valued Jacobi operators requires a deep investigation of the matrix-valued $m$-functions that arise from the Jacobi eigenvalue equation~\eqref{1.1}. In our generalized setting, the half-line $m$-functions associated with $J_\pm$, denoted as $M_\pm(n, z)$, evolve into matrices living in the Siegel upper half-space $\mathfrak{S}_l$, defined as the set of $l \times l$ symmetric matrices with a strictly positive definite imaginary part. In order to control the convergence of these functions, one considers $\mathfrak{S}_l$ endowed with the Finsler metric
\[
d_\infty(Z_1, Z_2) = \inf_{Z(t)} \int_0^1 \Vert Y(t)^{-1/2} \dot{Z}(t) Y(t)^{-1/2} \Vert\ dt,
\]
with $Z_{1}, Z_{2}\in\mathfrak{S}_l$, where the infimum is taken over all differentiable paths $Z(t)$ joining $Z_{1}$ to $Z_{2}$. By realizing the transfer matrices as complex symplectic matrices acting on $\mathfrak{S}_l$ via fractional linear transformations, we prove that their action is distance decreasing under this metric. Here, we follow the ideas presented in [\ref{McBride2}].

In order to extend the Breimesser-Pearson Value Distribution Theory, the primary tool used by Remling in order to obtain the Oracle Theorem, to the matrix setting, we introduce specific matrix-valued harmonic measures. For a matrix-valued Herglotz function $F$ and for a Borel set $S$, we define the functions $\omega_{F(z)}^{I}(S)$ and $\omega_{F(z)}^{R}(S)$ as customized probability densities:
\begin{equation}\label{117}
\omega_{F(z)}^{I}(S) = \frac{1}{\pi}\int_S \operatorname{Im}\Big(\lambda I-i((\operatorname{Im}F(z))^{1/2}+I)\Big)^{-1}\ d\lambda,
\end{equation}
\begin{equation}\label{107}
\omega_{F(z)}^{R}(S)= \frac{1}{\pi}\int_S \operatorname{Im}\left(\frac{1}{\Vert\operatorname{Im}F(z)+I\Vert}\operatorname{Re}F(z)-\lambda I+iI\right)^{-1}\ d\lambda.
\end{equation}
It is crucial to emphasize that the functions $\omega_{F(z)}^{I}$ and $\omega_{F(z)}^{R}$ were deliberately defined separately, and constructed exactly in this manner in order to guarantee that the resulting matrices are symmetric and positive definite. 

A critical distinction in our framework is the treatment of the absolutely continuous spectrum. In the matrix-valued case ($l>1$), the spectrum is stratified by multiplicity. Here, one has to consider $\Sigma_{ac,l}$, which denotes the essential support of the absolutely continuous spectrum of \textit{maximal} multiplicity $l$. This full-rank condition is analytically indispensable to establish the strict reflectionless property, given that the $m$-functions must belong to $\mathfrak{S}_l$.


Firstly, by generalizing the Breimesser-Pearson Theory, we establish in Theorem~\ref{25} the convergence of the harmonic measures associated with the half-line $m$-functions. More specifically, we prove that for each $C>0$, each $A,S\in\mathcal{B}(\mathbb{R})$ so that $A \subset \Sigma_{ac,l}$, $|A|>0$, and $S\subset(-C,C)$, the identity
\[
\lim_{n \to \infty} \left\Vert \int_A \omega_{M_-(n,t)}^\alpha(-S)\ dt - \int_A \omega_{M_+(n,t)}^\alpha(S)\ dt \right\Vert = 0
\]
holds. By building directly upon this convergence and the continuous dependence of the harmonic measures, Theorem~\ref{27} establishes the spectral rigidity of the system. Namely, we prove that the $\omega$-limit set of any bounded matrix-valued potential in $\mathcal{V}^C_+$ consists entirely of reflectionless matrix potentials on $\Sigma_{ac,l}$, the essential support of the absolutely continuous spectrum of multiplicity $l$ of $J_+^V$:
\[
\omega(V) \subset \mathcal{R}(\Sigma_{ac,l});
\]
recall that a potential $W$ is said to be reflectionless on $A$ if $M_+^W(t)=-\overline{M_-^W(t)}$ for a.e. $t\in A$.

Finally, this rigidity allows us to obtain the Oracle Theorem for matrix-valued Jacobi operators, (Theorem~\ref{98}),  which reads almost identical to the Oracle Theorem for scalar Jacobi operators: 
for each $\epsilon>0$, there exists a smooth function $\Delta$ such that for each potential $U\in \mathcal{V}^C_+$ so that $\Sigma_{ac,l}(U)\supset A$ (a Borel set of positive Lebesgue measure), there exists $n_0\in\mathbb{N}$ such that for each $n\ge n_0$,
\[
\rho \big( U(n+1), \Delta(U(n-L), \dots, U(n)) \big) < \epsilon.
\]
This result confirms that, as long as the absolutely continuous spectrum of maximal multiplicity has a positive Lebesgue measure, the complex, multichannel nature of matrix-valued quantum systems does not destroy their deterministic predictability.

The power of this predictability becomes even more evident when one considers the direct consequences originally established by Remling in the scalar setting [\ref{remling}]. Namely, if a potential takes only finitely many values and possesses some absolutely continuous spectrum, it must be eventually periodic. Furthermore, the structural rigidity of the $\omega$-limit set leads to the semicontinuity of the absolutely continuous spectrum, which proves that limit potentials retain at least as much absolutely continuous spectrum as the original potential. Such results have direct counterparts for the matrix-valued case, which we prove in Section~\ref{CTO}; see Theorem~\ref{110a} and Corollary~\ref{111a}. Other scalar consequences, such as Denisov-type asymptotics for the free interval and finite-gap convergence, require additional results and remain outside the scope of this paper.

\subsection*{Organization of the paper}

This paper is organized as follows. Section~\ref{101} introduces the half-line Weyl-Titchmarsh $m$-functions on the Siegel upper half-space $\mathfrak{S}_l$ for matrix-valued operators, proving the distance-decreasing property of transfer matrices under the Finsler metric $d_\infty$ via fractional linear transformations. We further investigate the symplectic structure of the transfer matrices and their role in the evolution of the eigenvalue equation along the lattice. Additionally, we characterize the $m$-functions as matrix-valued Herglotz functions, establishing their analytical properties and boundary behavior.

Section~\ref{108} presents convergence properties of the matrix-valued harmonic measures $\omega^I$ and $\omega^R$, and develops the matrix analogue of value distribution theory. More specifically, we prove that these measures converge in norm and in the value distribution sense as the spectral parameter approaches the real axis. These convergences provide the analytical background for the generalized Breimesser-Pearson Theory in the matrix-valued setting.

Section~\ref{102} presents the appropriate topological spaces for matrix-valued potentials and their product metrics. In Section~\ref{114}, we analyze the matrix-valued Weyl-Titchmarsh $m$-functions for the half-line restrictions $J_{\pm}$ and provide the necessary link between the restricted operators and the spectral measures. This section also houses the proof of the reflectionless property of the $\omega$-limit set, Theorems~\ref{25} and~\ref{27}, which shows that the  connection between the absolutely continuous spectrum of multiplicity $l$ and the reflectionless asymptotic behavior is still valid in the matrix-valued setting.

Section~\ref{104} states and proves the Oracle Theorem for matrix-valued Jacobi operators, Theorem~\ref{98}. In Section~\ref{115}, we study some properties of reflectionless potentials, which provide tools for the proof of  the Oracle Theorem, presented in the subsequent section. Finally, in Section~\ref{CTO}, we extend some of the consequences of the scalar Oracle Theorem to matrix-valued operators.

Appendix~\ref{92} contains the detailed proof of the continuous dependence of $\omega_{F(z)}^{I}$ and $\omega_{F(z)}^{R}$ with respect to $F(z)$, by considering the Finsler metric and the operator norm. This result is of fundamental importance for the stability of the spectral estimates and for the consistency of the main arguments.

\section[Half-line Weyl-Titchmarsh m-functions]{Half-line Weyl-Titchmarsh $m$-functions of matrix-valued Jacobi operators}\label{101}

    We show in this section that the half-line $m$-functions associated with the matrix-valued Jacobi operators assume values in the Siegel upper half-space. We also introduce a metric on the space of $m$-functions associated with the matrix-valued Jacobi operators. Then, we show that the action of transfer matrices on these $m$-functions is distance decreasing. The results here are all based on the results presented in~[\ref{McBride2}].

    \vspace{0.3cm}
    \subsection[Half-line Weyl-Titchmarsh m-functions]{Half-line Weyl-Titchmarsh $m$-functions}\label{SMM}
    
    In what follows, $J$ is the matrix-valued Jacobi operator defined by~\eqref{fabricio2.1} such that $((D_n,V_n))_{n\in\mathbb{Z}}$ satisfies~\eqref{CSP}.

    An $l\times 2l$ matrix-valued sequence $Y(n)$ is said to be a fundamental matrix solution to equation 
    \begin{equation}\label{1.1}
    D_n \mathbf{w}_{n+1}+D_{n-1}\mathbf{w}_{n-1}+V_n\mathbf{w}_n=z\mathbf{w}_n, \quad z\in\mathbb{C}.
    \end{equation}  if the columns of $Y(n)$ form a set of linearly independent solutions to equation (\ref{1.1}).

    Write $Y=(U,V)$, where $U$ and $V$ are $l\times l$ matrix-valued solutions to equation (\ref{1.1}). Our goal here is to discuss the half-line Weyl-Titchmarsh $m$-functions. 
    Let, for each $n\in\mathbb{Z}$,  $\mathbb{Z}_-^n=\{-\infty,\ldots,n\}$ and $\mathbb{Z}_+^n=\{n+1,n+2,\ldots\}$. Suppose that $U$ and $V$ satisfy the following initial conditions at $n$:
    \begin{align*}
    U(n) &= -D_n^{-1}, & V(n) &= 0, \\
    U(n+1) &= 0, & V(n+1) &= I.
    \end{align*}
    
    One defines the Weyl-Titchmarsh $m$-functions on $\mathbb{Z}_{-}^n$ and $\mathbb{Z}_{+}^n$ by requiring that the sequences 
    \begin{equation} \label{3.1}
    F_{\pm}(n)=U(n)\pm M_{\pm}(n)V(n)
    \end{equation}
 satisfy $F_{\pm}\in \ell^{2}(\mathbb{Z}_\pm,\mathbb{C}^{l\times l})$. 
    Note that the columns of $F_{\pm}(n)$ form a linearly independent set of solutions to equation (\ref{1.1}), and therefore for each $n\in\mathbb{Z}$, $F_{\pm}(n)$ are invertible. For each $z\in\mathbb{C}_{+}$, these half-line Weyl-Titchmarsh matrix-valued $m$-functions are uniquely determined (since $J_\pm$ are in the limit-point case at $\pm\infty$; see Proposition~2.4 in~\cite{fabricioSilasA}) and satisfy
    \begin{equation} \label{3.2}
    M_{+}(n,z)=-F_{+}(n+1,z)F_{+}(n,z)^{-1}D_n^{-1}; \quad M_{-}(n,z)=F_{-}(n+1,z)F_{-}(n,z)^{-1}D_n^{-1}.
    \end{equation}

    Firstly, we show that $M_{\pm}(n)$ can be expressed in terms of a resolvent operator. 

    Let $J_{\pm}$ be the restriction of $J$ to the space $\ell^{2}(\mathbb{Z}^n_{\pm},\mathbb{C}^{l})$, that is,
      \begin{equation*}
    (J_+\mathbf{u})_k=\begin{cases}D_{k-1}\mathbf{u}_{k-1}+V_k\mathbf{u}_k+D_k\mathbf{u}_{k+1} & \text{if } k>n+1\\ V_{k}\mathbf{u}_{k}+D_k\mathbf{u}_{k+1} & \text{if } k=n+1\end{cases},
      \end{equation*}
      \begin{equation*}
    (J_-\mathbf{u})_k=\begin{cases}D_{k-1}\mathbf{u}_{k-1}+V_k\mathbf{u}_k+D_k\mathbf{u}_{k+1} & \text{if } k<n\\ V_{k}\mathbf{u}_k+D_{k-1}\mathbf{u}_{k-1} & \text{if } k=n\end{cases},
      \end{equation*}
   
    \noindent and denote the Dirac delta-type vector-valued sequences by $\delta_{j}^{k}$, where for each $k \in \mathbb{Z}$ and each $j \in \{1, \ldots, l\}$, its $m$-th entry is given by
    \[ \delta_{j}^{k}(m) = \delta_{k,m} e_j = \begin{cases} e_j & \text{if } m=k \\ 0 & \text{if } m \neq k \end{cases}; \]
    here, $\{e_1, \ldots, e_l\}$ stands for the canonical standard basis of $\mathbb{C}^l$. 
    
    Similarly, for each $k \in \mathbb{Z}$, let $\Delta_{k}$ denote the matrix-valued sequence which $m$-th entry is given by
    \[ \Delta_{k}(m) = \delta_{k,m} I = \begin{cases} I & \text{if } m=k \\ 0 & \text{if } m \neq k \end{cases}, \]
    where $I$ and $0$ are the $l\times l$ identity and null matrices, respectively.

    \begin{obs}
    For a linear operator $A$ acting on $\ell^2(\mathbb{Z}, \mathbb{C}^l)$ (or $\ell^{2}(\mathbb{Z}^n_{\pm},\mathbb{C}^{l})$),  
    \begin{equation}\label{block_operator_notation}
        \langle \Delta_p, A \Delta_m \rangle
    \end{equation}
denotes, for each $p,m\in\mathbb{Z}^n_{\pm}$,  the $l \times l$ matrix whose $(i,j)$-th entry is given by the standard vector inner product $$\langle \delta_i^p, A \delta_j^m \rangle_{\ell^2(\mathbb{C}^l)}\ \ .$$
    
    Equivalently, this notation represents the matrix-valued evaluation (often related to the Gramian, or Gram matrix, defined as the matrix of pairwise inner products of a set of vectors), yielding an $l \times l$ matrix rather than a scalar. In the following result, the bracket notation $\langle \Delta_p, \cdot \, \Delta_m \rangle$ is used to represent the $l \times l$ block-matrix evaluation of the resolvent operator, as introduced in \eqref{block_operator_notation}. 
    \end{obs}

    \begin{pro}\label{3.3}
    Let $z\in\mathbb{C}_{+}$. Then, for each $n\in\mathbb{Z}$,
    \begin{eqnarray*} 
      M_{+}(n,z)&=&\langle\Delta_{n+1},(J_{+}-zI)^{-1}\Delta_{n+1}\rangle\, ,\\ 
      M_{-}(n,z)&=&-(D_{n}\langle\Delta_{n},(J_{-}-zI)^{-1}\Delta_{n}\rangle D_{n})^{-1}\,.
    \end{eqnarray*}
    \end{pro}
    
    \begin{proof}
    Set $g_{j}(k,z):=(J_{+}-zI)^{-1}\delta_{j}^{n+1}(k)$. Note that for each $j=1,\ldots,l$ and each $k\ge n+1$, $g_{j}(k,z)$ is a solution to the equation (\ref{1.1}) that belongs to $\ell^{2}(\mathbb{Z}^n_{+},\mathbb{C}^{l})$. Set, for each $k\ge n+1$,
    \[ G(k,z):=(g_{1}(k,z),g_{2}(k,z),...,g_{l}(k,z)); \]
    $G(k,z)$ is a matrix-valued solution to equation (\ref{1.1}) for $k \ge n+1$ and $G(\cdot,z)\in \ell^{2}(\mathbb{Z}^n_{+},\mathbb{C}^{l\times l})$. Thus, there exists an invertible constant matrix $C\in\mathbb{C}^{l\times l}$ such that for each $k\ge n+1$,
    \begin{equation} \label{3.4}
    G(k,z)=F_+(k,z)C.
    \end{equation}

    In order to determine $C$, one needs only to compare $G(n+1,z)$ and $F_+(n+1,z)$. Note that evaluating the action of the operator at the boundary node $n+1$ yields
    \[ \left((J_{+}-zI)G\right)(n+1,z)=(\delta_{1}^{n+1}(n+1),\delta_{2}^{n+1}(n+1),...,\delta_{l}^{n+1}(n+1)), \]
    and so $\left((J_{+}-zI)G\right)(n+1,z)=I$. Therefore, by the definition of $J_+$ at node $n+1$, one gets
    $$V_{n+1}G(n+1,z) + D_{n+1}G(n+2,z) - zG(n+1,z) = I$$
    \begin{equation} \label{3.5}
     \implies D_{n+1}G(n+2,z)=(z-V_{n+1})G(n+1,z)+I.
    \end{equation}

    On the other hand, since $F_+(k,z)$ is a fundamental solution to equation (\ref{1.1}) everywhere, evaluating it at $k=n+1$ gives
    \[ D_{n+1}F_+(n+2,z)=(z-V_{n+1})F_+(n+1,z)-D_n F_+(n,z). \]

    It follows from equation (\ref{3.4}) that
    \begin{equation} \label{3.6}
    D_{n+1}G(n+2,z)=(z-V_{n+1})F_+(n+1,z)C-D_n F_+(n,z)C.
    \end{equation}

    Finally, by comparing equations (\ref{3.5}) and (\ref{3.6}) and by noting that $G(n+1,z) = F_+(n+1,z)C$, one concludes that $-D_n F_+(n,z)C=I$. Since $F_+(n,z)=-D_n^{-1}$, it follows that $C=I$, and so, $G(n+1,z)=F_+(n+1,z)$.

    Hence, one gets 
    $M_{+}(n,z) = -F_{+}(n+1,z) F_{+}(n,z)^{-1} D_n^{-1} = -G(n+1,z) (-D_n) D_n^{-1} = G(n+1,z)$. It remains to obtain the entries of the matrix $G(n+1,z)$. Namely, for each $i,j=1,...,l$, the components $g_{ij}(n+1,z)$ are given by
    \[ g_{ij}(n+1,z)=\langle\delta_{i}^{n+1},(J_{+}-zI)^{-1}\delta_{j}^{n+1}\rangle, \]
    from which follows that
    \[ M_{+}(n,z)=\langle\Delta_{n+1},(J_{+}-zI)^{-1}\Delta_{n+1}\rangle\, .\]

    The proof for the identity for $M_-(n,z)$ follows a similar geometric reasoning but adjusted for the right-end boundary. Set $h_{j}(k,z):=(J_{-}-zI)^{-1}\delta_{j}^{n}(k)$. Note that for each $j=1,\ldots,l$ and each $k \le n$, $h_{j}(k,z)$ is a solution to the equation (\ref{1.1}) that belongs to $\ell^{2}(\mathbb{Z}^n_{-},\mathbb{C}^{l})$. Set, for each $k\le n$,
    \[ H(k,z):=(h_{1}(k,z),h_{2}(k,z),...,h_{l}(k,z)); \]
    $H(k,z)$ is a matrix-valued solution to equation (\ref{1.1}) and $H(\cdot,z)\in \ell^{2}(\mathbb{Z}^n_{-},\mathbb{C}^{l\times l})$. Thus, there exists an invertible constant matrix $K\in\mathbb{C}^{l\times l}$ such that for each $k\le n$,
    \begin{equation} \label{3.4a}
    H(k,z)=F_-(k,z)K.
    \end{equation}
    
    The evaluation of the action of $J_-$ at the boundary node $n$ yields
    \[ \left((J_{-}-zI)H\right)(n,z)=I \implies V_nH(n,z) + D_{n-1}H(n-1,z) - zH(n,z) = I. \]
    It follows from relation (\ref{3.4a}) that
    \begin{equation} \label{3.5a}
    V_nF_-(n,z)K + D_{n-1}F_-(n-1,z)K - zF_-(n,z)K = I.
    \end{equation}
    On the other hand, since $F_-(k,z)$ is a solution to equation (\ref{1.1}) everywhere, evaluating at $k=n$ gives
    \[ D_nF_-(n+1,z) + V_nF_-(n,z) + D_{n-1}F_-(n-1,z) = zF_-(n,z), \]
    which can be rewritten as
    \begin{equation} \label{3.6a}
    V_nF_-(n,z) + D_{n-1}F_-(n-1,z) - zF_-(n,z) = -D_nF_-(n+1,z).
    \end{equation}

    By multiplying both members of~\eqref{3.6a} to the right by $K$ and by comparing the resulting equation with~\eqref{3.5a}, one gets $-D_nF_-(n+1,z)K = I$, which leads to $K = -F_-(n+1,z)^{-1}D_n^{-1}$. It follows from~\eqref{3.4a} that
    \[ H(n,z) = -F_-(n,z)F_-(n+1,z)^{-1}D_n^{-1}. \]

    Since we know from (\ref{3.2}) that $M_-(n,z) = F_-(n+1,z)F_-(n,z)^{-1}D_n^{-1}$,  one gets $M_-(n,z)^{-1} = D_n F_-(n,z)F_-(n+1,z)^{-1}$. Therefore,
    \[ H(n,z) = -D_n^{-1} M_-(n,z)^{-1} D_n^{-1}, \]
and so, 
    \[ M_{-}(n,z) = -(D_nH(n,z)D_n)^{-1}. \]
    By extracting the entries of $H(n,z)$ exactly as before yields the identity 
    \[ M_{-}(n,z)=-(D_{n}\langle\Delta_{n},(J_{-}-zI)^{-1}\Delta_{n}\rangle D_{n})^{-1}.\]
    \end{proof}

    Proposition~(\ref{3.3}) shows that for each $n\in\mathbb{Z}$, $M_{\pm}(n,z)$ are matrix-valued analytic functions on $\mathbb{C}_+$.
    In fact, one can say even more: $M_{\pm}(n,z)$ are matrix-valued Herglotz functions.

    \begin{lema}\label{84}
    For each $z\in\mathbb{C}_{+}$ and each $n\in\mathbb{Z}$, $M_{\pm}(n,z)$ are symmetric and satisfy the relation
    \begin{equation} \label{3.7}
    D_n M_{\pm}(n,z) D_n+M_{\pm}(n-1,z)^{-1}\pm(zI-V_n)=0.
    \end{equation}
    \end{lema}
    
    \begin{proof}
    Let us show that $M_{\pm}(n,z)$ are symmetric. Namely, for $M_{-}(n,z)$, consider the following calculation:
    \begin{align*}
    M_{-}&(n,z)^{t}-M_{-}(n,z) = \\ &=(F_{-}(n+1,z)F_{-}(n,z)^{-1}D_n^{-1})^{t}-F_{-}(n+1,z)F_{-}(n,z)^{-1}D_n^{-1} \\&= D_n^{-1}(F_{-}(n,z)^{-1})^{t}[F_{-}(n+1,z)^{t}D_n F_{-}(n,z)-F_{-}(n,z)^{t}D_n F_{-}(n+1,z)]F_{-}(n,z)^{-1}D_n^{-1} \\&= -D_n^{-1}(F_{-}(n,z)^{-1})^{t}W_{n+1}(F_{-}(\cdot,z),F_{-}(\cdot,z))F_{-}(n,z)^{-1}D_n^{-1}.
    \end{align*}
    
    \noindent Since $F_{-}(\cdot,z)$ is a solution to equation (\ref{1.1}), $W_{n+1}(F_{-}(\cdot,z),F_{-}(\cdot,z))$ does not depend on $n$. Hence, it follows that $W_{n+1}(F_{-}(\cdot,z),F_{-}(\cdot,z))= 0$ 
    (the condition $F_{-} \in \ell^{2}(\mathbb{Z}^n_-,\mathbb{C}^{l\times l})$ guarantees that the boundary term of the Wronskian vanishes at $-\infty$).
    This in turn implies that $M_{-}(n,z)$ is symmetric. A similar calculation shows that $M_{+}(n,z)$ is also symmetric.
    
    It remains to show that $M_{\pm}(n,z)$ satisfy relation (\ref{3.7}). Namely, it follows from equation (\ref{1.1}) that
    \[ D_n F_{+}(n+1,z)+D_{n-1} F_{+}(n-1,z)+(V_n-zI)F_{+}(n,z)=0; \]
    by multiplying both of its members to the right by $F_{+}(n,z)^{-1}$, one gets  
    \[ D_n (-M_{+}(n,z) D_n) + D_{n-1} F_{+}(n-1,z)F_{+}(n,z)^{-1} + (V_n-zI)=0, \]
    given that  $F_+(n+1,z) F_+(n,z)^{-1} = -M_+(n,z) D_n$.
    
    It also follows from the definition of $M_+(n-1,z)$ that $F_{+}(n-1,z)F_{+}(n,z)^{-1} = -D_{n-1}^{-1} M_+(n-1,z)^{-1}$; by replacing this into 
    the previous relation, one gets 
    \[ -D_n M_{+}(n,z) D_n - D_{n-1} \left( D_{n-1}^{-1} M_{+}(n-1,z)^{-1} \right) + V_n - zI = 0, \]
    that is,$$ D_n M_{+}(n,z) D_n + M_{+}(n-1,z)^{-1} + zI - V_n = 0. $$
    
  The proof for $M_{-}(n,z)$ follows the same steps, so we omit it. 
%
    \end{proof}

    Recall that the imaginary part of $M_{\pm}(n,z)$ is given by
    \[ \text{Im } M_{\pm}(n,z)=\frac{1}{2i}(M_{\pm}(n,z)-M_{\pm}(n,z)^{*}). \]
 
    Inwhat follows, we also need the so-called Green Formula (see Section~2.3 in~\cite{fabricioSilasA} for a discussion): let $n>m\ge 0$, then, 
    \begin{equation}\label{2.4}
 \sum_{k=m}^n(A_k^TJB_k-JA_k^TB_k)=W_{n+1}(A,B)-W_m(A,B),
\end{equation}
where
\[
W_m(A,B):=A_{m-1}^TD_{m-1}B_m-A_m^TD_{m-1}B_{m-1}.
\]

    \begin{pro}\label{86}
    For each $z\in\mathbb{C}_{+}$ and each $n\in\mathbb{Z}$, $\text{Im } M_{\pm}(n,z)>0$.
    \end{pro}
      \begin{proof}
      We first show that $\text{Im } M_{+}(n,z)>0$. Recall that 
      $F_{+} \in \ell^{2}(\mathbb{Z}_+, \mathbb{C}^{l \times l})$. This square-summability condition guarantees that the boundary term of the Wronskian vanishes at infinity, i.e., that $\lim_{m \to \infty} W_{m}(\overline{F_+}, F_+) = 0$. By applying Green’s Formula (\ref{2.4}) to the interval $[n+1, \infty)$, and by noting that $(J F_{+})(j,z) = z F_{+}(j,z)$, one gets
    \begin{align*}
    \sum_{j=n+1}^{\infty} &\left( F_{+}(j,z)^{*} (z F_{+}(j,z)) - (\overline{z} F_{+}(j,z)^{*}) F_{+}(j,z) \right)= \\ &= W_{n}(\overline{F_+}, F_+) - \lim_{m \to \infty} W_{m}(\overline{F_+}, F_+) (z-\overline{z}) \\ &= W_{n}(\overline{F_{+}(n,z)}, F_{+}(n,z)).
    \end{align*}
    Since $z - \overline{z} = 2i\ \text{Im}(z)$, this simplifies to
    \begin{equation} \label{3.8_new}2i\ \text{Im}(z) \sum_{j=n+1}^{\infty} F_{+}(j,z)^{*} F_{+}(j,z) = W_{n}(\overline{F_{+}(n,z)}, F_{+}(n,z)).\end{equation}

    It follows from Definition (\ref{3.2}) that $F_{+}(n+1,z) = -M_{+}(n,z) D_n F_{+}(n,z)$. Moreover, given that $D_n$ is real and symmetric ($D_n^* = D_n$), the Wronskian on the right side of (\ref{3.8_new}) can be written as 
    \begin{align*}W_{n}(\overline{F_+}, F_+) &= \overline{F_+(n+1,z)}^{t} D_n F_{+}(n,z) - \overline{F_+(n,z)}^{t} D_n F_{+}(n+1,z) \\&= \overline{(-M_{+} D_n F_{+})}^{t} D_n F_{+} - \overline{F_{+}}^{t} D_n (-M_{+} D_n F_{+}) \\&= -\overline{F_{+}}^{t} D_n M_{+}^{*} D_n F_{+} + \overline{F_{+}}^{t} D_n M_{+} D_n F_{+} \\
    &= \overline{F_{+}(n,z)}^{t} D_n \left( M_{+}(n,z) - M_{+}(n,z)^{*} \right) D_n F_{+}(n,z).
    \end{align*} By using the identities $F_+(n,z)=U(n)+M_+(n,z)V(n)=-D_n^{-1}$ and  $M_{+} - M_{+}^{*} = 2i\ \text{Im } M_{+}$, one gets 
    \begin{equation} \label{3.10_new}\text{Im}(z) \sum_{j=n+1}^{\infty} F_{+}(j,z)^{*} F_{+}(j,z) = 
      \text{Im } M_{+}(n,z),\end{equation} and so, $\text{Im } M_{+}(n,z)>0$.
    
   Since we use the same arguments for the left half-line function $M_{-}(n,z)$, we omit the proof. 
%
%
%
  %
    \end{proof}
    
    \vspace{0.5cm}
    Note that Proposition~\ref{86}  highlights that the different sign conventions in the initial definitions of $M_+$ and $M_-$ in (\ref{3.2}) perfectly explain the orientations of the limit, ensuring that both matrices behave consistently as Herglotz functions.

    \vspace{0.5cm}
    
    \subsection{Fractional linear transformations and distance decreasing actions}\label{FLTD}
    
    The fractional linear transformations have been successfully used in the study of Weyl-Titchmarsh $m$-functions. More specifically, the transfer matrices associated with Jacobi operators are fractional linear transformations (for a more general treatment, see~\cite{Freitas,Froese,Siegel}). For matrix-valued Jacobi operators, their transfer matrices may also be considered as matrix-valued fractional linear transformations. Namely, for each $T\in\mathbb{C}^{2l\times 2l}$ of the form
    \[ T=\begin{pmatrix}A&B\\ C&D\end{pmatrix}, \]
    where $A, B, C,$ and $D$ are $l\times l$ matrices, a matrix-valued fractional linear transformation is a map $T:\mathbb{C}^{l\times l}\rightarrow\mathbb{C}^{l\times l}$ defined by the law
    \begin{equation} \label{3.12}
    T(Z)=(AZ+B)(CZ+D)^{-1},
    \end{equation}
    where $Z\in\mathbb{C}^{l\times l}$ is such that $\det(CZ+D)\neq 0$. 

    It is particularly true that the transfer matrices associated with equation (\ref{1.1}) define matrix-valued fractional linear transformations of the form (\ref{3.12}). Namely, if $G(n)$ is a matrix solution to equation (\ref{1.1}), then for each $n\in\mathbb{Z}$, the matrix identity
    \begin{equation} \label{3.13}
    \begin{pmatrix}G(n+1)\\ \mp D_n G(n)\end{pmatrix}=\begin{pmatrix}D_n^{-1}(zI-V_n)&\pm D_n^{-1}\\ \mp D_n&0\end{pmatrix}\begin{pmatrix}G(n)\\ \mp D_{n-1} G(n-1)\end{pmatrix}
    \end{equation}
    follows. 
    The matrices
    \begin{equation} \label{3.14}
    T_{\pm}(n,z)=\begin{pmatrix}D_n^{-1}(zI-V_n)&\pm D_n^{-1}\\ \mp D_n&0\end{pmatrix}
    \end{equation}
    are the so-called transfer matrices associated with $J$, which describe the evolution of the state vectors $$\begin{pmatrix}G(n+1)\\ \mp D_n G(n)\end{pmatrix}$$ under iteration of $T_{\pm}(n,z)$.

    Then, each $T_{\pm}(n):=T_{\pm}(n,z)$ can be considered as a complex matrix-valued linear transformation (\ref{3.12}) acting on the set of the Weyl-Titchmarsh matrix-valued $m$-functions $M_{\pm}(n)$: 
    $$T_{\pm}(n)M_{\pm}(n)=\left((D_n^{-1}(zI-V_n))M_{\pm}(n)\pm D_n^{-1}\right)\left(\mp D_nM_{\pm}(n)\right)^{-1}\ .$$

    The transfer matrices also relate $M_{\pm}(n)$ and $M_{\pm}(n-1)$. 
    
    \begin{lema}\label{90}
    For each $z\in\mathbb{C}$ and each $n\in\mathbb{Z}$, $M_{\pm}(n)=T_{\pm}(n)M_{\pm}(n-1)$.
    \end{lema}
     \begin{proof}
      Since $F_{\pm}(n)$ are solutions to equation (\ref{1.1}), one has
      \[ D_n F_{\pm}(n+1,z)+D_{n-1} F_{\pm}(n-1,z)+V_n F_{\pm}(n,z)=zF_{\pm}(n,z). \]
    Consider the following calculation:
    \begin{eqnarray*}
    \begin{pmatrix}D_n^{-1}(zI-V_n)& D_n^{-1}\\-D_n&0\end{pmatrix}M_+(n-1)&=&\begin{pmatrix}D_n^{-1}(zI-V_n)& D_n^{-1}\\ -D_n&0\end{pmatrix} \begin{pmatrix}F_{+}(n,z)\\ -D_{n-1}F_{+}(n-1,z)\end{pmatrix} \\ &=& \begin{pmatrix}D_n^{-1}((zI-V_n)F_{+}(n,z)-D_{n-1}F_{+}(n-1,z))\\ -D_n F_{+}(n,z)\end{pmatrix} \\ 
    &=& \begin{pmatrix}F_{+}(n+1,z)\\ -D_n F_{+}(n,z)\end{pmatrix}=-F_{+}(n+1,z)F_{+}(n,z)^{-1}D_n^{-1}\\&=&M_+(n)\, ,
    \end{eqnarray*}
    where the penultimate line follows from the usual interpretation of matrix-valued linear transformation. This shows the result for $M_{+}(n)$. A similar calculation shows that
    \[ M_{-}(n)=T_{-}(n)M_{-}(n-1). \]
    \end{proof}

    \begin{lema}\label{85}
     
      For each $n\in\mathbb{Z}$ and each $z\in\mathbb{C}$, 
      the transfer matrices $T_{\pm}(n)=T_{\pm}(n,z)$ satisfy the following symplectic identities:
    \begin{enumerate}
        \item $T_{\pm}(n)^{t}\begin{pmatrix}0&I\\ -I&0\end{pmatrix}T_{\pm}(n)=\begin{pmatrix}0&I\\ -I&0\end{pmatrix}$
        \item $\begin{pmatrix}I&0\\ 0&-I\end{pmatrix}T_{+}(n)\begin{pmatrix}I&0\\ 0&-I\end{pmatrix}=T_{-}(n)$.
    \end{enumerate}
    \end{lema}
    
    \begin{proof}
    These identities follow from matrix multiplication:
\begin{enumerate}\item    \begin{align*}
    T_{\pm}(n)^{t}\begin{pmatrix}0&I\\ -I&0\end{pmatrix}T_{\pm}(n) &= \begin{pmatrix}(zI-V_n)D_n^{-1}&\mp D_n\\ \pm D_n^{-1}&0\end{pmatrix}\begin{pmatrix}0&I\\ -I&0\end{pmatrix}\begin{pmatrix}D_n^{-1}(zI-V_n)&\pm D_n^{-1}\\ \mp D_n&0\end{pmatrix} \\
    &= \begin{pmatrix}\pm D_n&(zI-V_n)D_n^{-1}\\ 0&\pm D_n^{-1}\end{pmatrix}\begin{pmatrix}D_n^{-1}(zI-V_n)&\pm D_n^{-1}\\ \mp D_n&0\end{pmatrix} \\
    &= \begin{pmatrix}0&I\\ -I&0\end{pmatrix}
    \end{align*}

   \item  \begin{align*}
    \begin{pmatrix}I&0\\ 0&-I\end{pmatrix}T_{+}(n)\begin{pmatrix}I&0\\ 0&-I\end{pmatrix} &= \begin{pmatrix}I&0\\ 0&-I\end{pmatrix}\begin{pmatrix}D_n^{-1}(zI-V_n)&D_n^{-1}\\ -D_n&0\end{pmatrix}\begin{pmatrix}I&0\\ 0&-I\end{pmatrix} \\
    &= \begin{pmatrix}D_n^{-1}(zI-V_n)&-D_n^{-1}\\ D_n&0\end{pmatrix} = T_{-}(n)
    \end{align*}
    \end{enumerate} 
    \end{proof}
    
    The matrix in the first identity is a special symplectic matrix, which is usually denoted by 
    \[ J=\begin{pmatrix}0&I\\ -I&0\end{pmatrix} \]

    It follows from Lemma~\ref{85} that for each $n\in\mathbb{Z}$, $T_{\pm}(n)\in SL(2l,\mathbb{C})$, which is the group of $2l\times 2l$ complex matrices with determinat equal to one. Therefore, these are matrix-valued functions that map the complex upper half plane $\mathbb{C}_{+}$ into $SL(2l,\mathbb{C})$. It follows from Lemma \ref{84} and Proposition \ref{86} that for each $n\in\mathbb{Z}$ and each $z\in\mathbb{C}_{+}$, $M_{\pm}(n,z)$ are symmetric and have positive definite imaginary part; therefore, $M_{\pm}(n,z)\in\mathfrak{S}_l$, where 
    \[ \mathfrak{S}_l = \{Z \in \mathbb{C}^{l\times l}: Z = X +iY, X^{t}=X, Y^{t} = Y, Y > 0\} \]
    stands for the Siegel upper half plane.
    
    If $T\in Sp(2l,\mathbb{R})$, then the mapping $Z\mapsto T(Z)$ presented in equation (\ref{3.12}) is well defined, and it is a group action on $\mathfrak{S}_l$\footnote{While the mapping $Z \mapsto T(Z)$ is often discussed in the context of $SL(2l, \mathbb{R})$, it is essential to specify that for $l > 1$, the natural group of automorphisms of the Siegel upper half-plane $\mathfrak{S}_l$ is a subgroup of $SL(2l, \mathbb{R})$, namely the real symplectic group $Sp(2l,\mathbb{R})$. A matrix $T = \begin{pmatrix} A & B \\ C & D \end{pmatrix}$ belongs to $Sp(2l, \mathbb{R})$ if, and only if, $T^t J T = J$. 
    This condition is equivalent to the following relations:$$ A^t C = C^t A, \quad B^t D = D^t B, \quad \text{and} \quad A^t D - C^t B = I.$$ These constraints are necessary to ensure the two defining properties of $\mathfrak{S}_l$: (a) Symmetry. The identities $A^t C = C^t A$ and $B^t D = D^t B$ guarantee that for each $Z\in\mathfrak{S}_l$, $T(Z) = (AZ + B)(CZ + D)^{-1}$ is a symmetric matrix. Without the symplectic requirement, the symmetry of $Z$ would not be preserved under the action of a general $SL(2l, \mathbb{R})$ matrix. (b) Positivity. it follows from the symplectic identity $T^t J T = J$ that 
    $$ \text{Im}(T(Z)) = ((CZ+D)^*)^{-1} \text{Im}(Z) (CZ+D)^{-1}. $$ Since $Z \in \mathfrak{S}_l$ implies $\text{Im}(Z) > 0$, this congruence transformation guarantees that $\text{Im}(T(Z))$ is also positive definite. In the specific case of $l=1$, one has $SL(2, \mathbb{R}) = Sp(2,\mathbb{R})$, which justifies the interchangeability of the terms in lower dimensions. However, for $l > 1$, the symplectic structure is the required framework to keep the mapping well-defined within $\mathfrak{S}_l$.}. The transfer matrices $T_{\pm}(n)$, however, are complex symplectic matrices acting on the space of $m$-functions. The reference  \ref{McBride} provides a necessary condition that a complex symplectic matrix $T\in SL(2l,\mathbb{C})$ must satisfy for the map $Z\mapsto T(Z)$ to be well defined. Here, we show that the maps $Z\mapsto T_{\pm}(n,z)(Z)$ are well defined on $\mathfrak{S}_l$. 
    
    \begin{lema}\label{87}
      Let, for each $n\in\mathbb{Z}$, $V_n$ be symmetric and $D_n$ be symmetric and invertible, and let $z\in\mathbb{C}_{+}\cup\mathbb{R}$. Then, 
      $T_{\pm}(n,z)(Z)\in \mathfrak{S}_l$ for each $Z\in\mathfrak{S}_l$.
    \end{lema}
    
    \begin{proof}
      The second identity in Lemma \ref{85} shows that for each $n\in\mathbb{Z}$ and each $z\in\mathbb{C}_+$, $T_{+}(n,z)$ and $T_{-}(n,z)$ are conjugated. Therefore, it is enough to show that $T_{-}(n,z)(Z)\in \mathfrak{S}_l$. Write the transfer matrix $T_{-}(n,z)$ as the following product of matrices:
    \[ T_{-}(n,z)=\begin{pmatrix}D_n^{-1}&0\\ 0&D_n\end{pmatrix}\begin{pmatrix}I&-V_n\\ 0&I\end{pmatrix}\begin{pmatrix}I&zI\\ 0&I\end{pmatrix}\begin{pmatrix}0&-I\\ I&0\end{pmatrix}=:T_{D}T_{V}T_{z}T_{J}. \]

          Since $J$ is a symplectic matrix, so it is $T_J$. Moreover, 
          $T_{D}, T_{V}\in Sp(2l,\mathbb{R});$ hence $T_{J}(Z)$, $T_{V}(Z)$, $T_{D}(W)\in \mathfrak{S}_l$ for each $Z,W\in\mathfrak{S}_l$ (see the previous footnote).

          Now, we affirm that for each $Z\in\mathfrak{S}_l$, $T_{z}(Z)\in\mathfrak{S}_l$. In order to prove this, note that 
          $$ T_z(Z) = (I \cdot Z + zI)(0 \cdot Z + I)^{-1} = Z + zI.$$ By decomposing $Z$ into its real and imaginary parts as $Z = X + iY$, and by writing $z = \text{Re}(z) + i\text{Im}(z)$, one gets $$ T_z(Z) = X + iY + (\text{Re}(z)I + i\text{Im}(z)I) = (X + \text{Re}(z)I) + i(Y + \text{Im}(z)I). $$

          The resulting real part, $X + \text{Re}(z)I$, is symmetric,  since both $X$ and $I$ are symmetric. Now, since 
          $Y > 0$ and $\text{Im}(z)I \ge 0$, it follows that $Y + \text{Im}(z)I>0$. 
          Consequently, $T_z(Z) \in \mathfrak{S}_l$. This concludes the proof, since $$T_-(n,z)(Z)=(T_D\circ T_V\circ T_z)(T_J(Z))=(T_D\circ T_V)(T_z(W))=T_D(T_V(U))=T_D(K) \in \mathfrak{S}_l,$$ with $W=T_J(Z) \in \mathfrak{S}_l$, $U=T_z(W) \in \mathfrak{S}_l$, $K=T_V(U) \in \mathfrak{S}_l$.
    \end{proof}
    
    \vspace{0.5cm}
    Next, we define a notion of distance between Weyl-Titchmarsh $m$-functions evaluated at $z\in\mathbb{C}_+$: let 
    $d_{\infty}:\mathfrak{S}_l\times\mathfrak{S}_l\rightarrow\mathbb{R}$ be given by the law
    \begin{equation} \label{3.15}
    d_{\infty}(Z_{1},Z_{2})=\inf_{Z(t)}\int_{0}^{1}F_{Z(t)}(\dot{Z}(t))dt\,,
    \end{equation}
 where 
    \begin{equation} \label{3.16}
    F_{Z}(\cdot)=\|Y^{-1/2}(\cdot)Y^{-1/2}\|,\ \ \ \ Z=X+iY\ ;
    \end{equation}
  the  infimum is taken over all differentiable paths $Z(t)$ joining $Z_{1}$ to $Z_{2}$. Here, $Y$ is a positive definite matrix, and so $Y^{-1/2}=(Y^{1/2})^{-1}=(Y^{-1})^{1/2}.$ The norm in equation (\ref{3.16}) is the operator norm of matrices acting on $\mathbb{C}^{l}$. It is straightforward to show that $d_{\infty}$ is a metric on $\mathfrak{S}_l$ (the so-called Finsler metric), from which follows that $(\mathfrak{S}_l,d_{\infty})$ is a metric space (see~\cite{Froese} for more details).
    
    \begin{lema}\label{89}
      Let $D_0 \in \mathbb{R}^{l \times l}$ be a symmetric invertible matrix, let $A\in \mathbb{R}^{l \times l}$ be positive definite and let 
      $\lambda>0$. Then,
    \[ A^{1/2}(A+\lambda I)^{-1/2}=(A(A+\lambda I)^{-1})^{1/2}=((I+\lambda A^{-1})^{-1})^{1/2}. \]
    Moreover, if $D_0^2-\lambda A$ is also positive definite, then
    \[ \|A^{1/2}(A+\lambda I)^{-1/2}\|^{2}<\frac{\|D_0\|^2}{\|D_0\|^2+\lambda^{2}}\ . \]
    \end{lema}
    
    \begin{proof}
    The first set of identities $$ A^{1/2}(A+\lambda I)^{-1/2}=(A(A+\lambda I)^{-1})^{1/2}=((I+\lambda A^{-1})^{-1})^{1/2} $$ follows from the continuous functional calculus for self-adjoint operators: since $A$ is a positive definite matrix and $\lambda > 0$, the operator $A+\lambda I$ is invertible and $\sigma(A+\lambda I)\subset(\lambda, \infty)$. Thus, the functions $f(x) = \sqrt{x(x+\lambda)^{-1}}$ and $g(x) = \sqrt{(1+\lambda x^{-1})^{-1}}$ are well defined on the spectrum of $A$. 
    
    Now,  note that since $A$ is positive definite, $\sigma(A)\subset\mathbb{R}_+$. Moreover, 
   the fact that $D_0^2 - \lambda A$ is  positive definite implies that $A < \lambda^{-1}D_0^2$. Consequently, the eigenvalues of $A$ are strictly bounded by the norm of $\lambda^{-1}D_0^2$:$$ \sigma < \frac{\|D_0\|^2}{\lambda}. $$
    
   Now, consider the squared norm of the operator $A^{1/2}(A+\lambda I)^{-1/2}$. Since $A$ is positive definite, one has 
   $$ \|A^{1/2}(A+\lambda I)^{-1/2}\|^2 = \|A(A+\lambda I)^{-1}\| = \sup_{\sigma \in \sigma(A)} \frac{\sigma}{\sigma + \lambda}. $$
    
    The function $h(\sigma) = \frac{\sigma}{\sigma + \lambda}$ is strictly increasing for $\sigma > 0$, given that $h'(\sigma) = \frac{\lambda}{(\sigma + \lambda)^2}>0$ for each $\sigma>0$. Therefore, the maximum value of $h(\sigma)$ is $h(\sigma_{\max})$, with 
    $\sigma_{max} < \|D_0\|^2/\lambda$:$$ \|A(A+\lambda I)^{-1}\| < h\left(\frac{\|D_0\|^2}{\lambda}\right) = \frac{\|D_0\|^2/\lambda}{\|D_0\|^2/\lambda + \lambda}=\frac{\|D_0\|^2}{\|D_0\|^2+\lambda^{2}}.$$ 
       \end{proof}
    
    For each $T\in SL(2l,\mathbb{R})$, when considering the map from (\ref{3.12}) as a group action, [\ref{McBride}] and [\ref{Freitas}] have shown that this action is, in fact, distance preserving with respect to the metric $d_{\infty}$. 
    
    \begin{teo}\label{88}
    Let $T\in SL(2l,\mathbb{R})$ and $W_{1}, W_{2}\in\mathfrak{S}_l$. Then,
    \[ d_{\infty}(T(W_{1}),T(W_{2}))=d_{\infty}(W_{1},W_{2}). \]
    \end{teo}
%
    %
   %
%
    %
    %
    
    The result stated in Theorem \ref{88} may not be true if $T\in SL(2l,\mathbb{C})$. Namely, for the 
    transfer matrices $T_{-}(n)$, the map defined by (\ref{3.12}) is in fact distance decreasing. 
    
    \begin{teo}\label{91}
    Let $z\in\mathbb{C}_{+}$, $y=\operatorname{Im}z>0$, and set $W_{j}=T_{-}(0,z)Z_{j}$ for $Z_{1},Z_{2}\in\mathfrak{S}_l$. Then, for each $n\in\mathbb{N}$,
    \begin{equation}
    d_{\infty}(T_{-}(n,z)W_{1},T_{-}(n,z)W_{2})\le\frac{\|D_0\|^2}{\|D_0\|^2+y^{2}}\ d_{\infty}(W_{1},W_{2}).
    \end{equation}
    \end{teo}
    
    \begin{proof}
      As in the proof of Lemma \ref{87}, write $T_{-}(n,z)=T_{D}T_{V}T_{z}T_{J}$. Since $T_{D},T_{V},T_{J}\in SL(2l, \mathbb{R})$ and since $T_{z}$ has range in $\mathfrak{S}_l$, it follows from Theorem \ref{88} and Lemma \ref{87} 
      that
    \begin{align} \label{3.17}
    d_{\infty}(T_{-}(n,z)(W_{1}), T_{-}(n, z)(W_{2})) &= d_{\infty}(T_{D}T_{V}T_{z}T_{J} (W_{1}), T_{D}T_{V}T_{z}T_{J}(W_{2})) \nonumber \\
    &= d_{\infty}(T_{z}T_{J}(W_{1}),T_{z}T_{J}(W_{2})).
    \end{align}
    Moreover, since $T_{J}\in SL(2l,\mathbb{R})$, one gets
    \begin{equation} \label{3.18}
    d_{\infty}(T_{J}(W_{1}),T_{J}(W_{2}))=d_{\infty}(W_{1},W_{2}).
    \end{equation}

    Let $U_{j}=T_{J}(W_{j}) \in \mathfrak{S}_l$ for $j=1,2$. Thus, $T_{z}(U_{j})=U_{j}+zI$. Let $U(t)$ be a length minimizing path between $U_{1}$ and $U_{2}$, entirely contained within $\mathfrak{S}_l$. Then, $U(t)+zI$ is a path between $U_{1}+zI$ and $U_{2}+zI$, and
    \[ \frac{d}{dt}(U(t)+zI)=\dot{U}(t). \]
    
    Set $V(t)=\operatorname{Im} U(t)>0$, so $\operatorname{Im} (U(t)+zI)=V(t)+yI.$ It follows from the definition of the Finsler metric (\ref{3.16}) that
    \begin{align*}
    F_{(U(t)+zI)}(\dot{U}(t)) &= \|(V(t)+yI)^{-1/2}\dot{U}(t)(V(t)+yI)^{-1/2}\| \\
    &= \|(V(t)+yI)^{-1/2}V(t)^{1/2}V(t)^{-1/2}\dot{U}(t)V(t)^{-1/2}V(t)^{1/2}(V(t)+yI)^{-1/2}\| \\
    &\le \|(V(t)+yI)^{-1/2}V(t)^{1/2}\|\cdot\|V(t)^{-1/2}\dot{U}(t)V(t)^{-1/2}\|\cdot\|V(t)^{1/2}(V(t)+yI)^{-1/2}\|.
    \end{align*}
    
    Since $V(t)^{1/2}$ and $(V(t)+yI)^{-1/2}$ are symmetric matrices, one gets 
    \begin{equation} \label{3.19}
    F_{(U(t)+zI)}(\dot{U}(t))\le\|(V(t)+yI)^{-1/2}V(t)^{1/2}\|^{2}F_{U(t)}(\dot{U}(t)).
    \end{equation}

    Now, since 
    $W_{j}=T_{-}(0,z)Z_{j}$, it follows that for each path $W(t)\in \mathfrak{S}_l$ between $W_1$ and $W_2$, there exists a path $Z(t) \in \mathfrak{S}_l$ such that $W(t) = T_{-}(0,z)Z(t)$; hence, $U(t)$ is of the form $U(t) = T_J(W(t))$. By treating the multiplication $W(t) = T_{-}(0,z)Z(t)$ as the group action and by using the matrix representation of $T_{-}(0,z)$, one gets
    \begin{align*}
    W(t) &= T_{-}(0,z)Z(t) \\ &= \begin{pmatrix}D_0^{-1}(zI-V_0)&-D_0^{-1}\\ D_0&0\end{pmatrix}Z(t) \\ \\[-1.5em] &= (D_0^{-1}(zI-V_0)Z(t) - D_0^{-1})(D_0 Z(t))^{-1} \\ 
    &= D_0^{-1}(zI-V_0)D_0^{-1} - D_0^{-1} Z(t)^{-1} D_0^{-1}.
    \end{align*}

    By taking the imaginary part of both members of this identity and by noting that $\operatorname{Im}(-Z(t)^{-1}) > 0$ for $Z(t) \in \mathfrak{S}_l$, it follows that
    \[ \operatorname{Im} W(t) = y D_0^{-2} + D_0^{-1} \operatorname{Im}(-Z(t)^{-1}) D_0^{-1} > y D_0^{-2}, \]
  and so (by the continuous functional calculus for positive operators),  
    \begin{equation*}
    (\operatorname{Im} W(t))^{-1} < y^{-1} D_0^{2}\ .
    \end{equation*}

    Since $U(t) = T_J(W(t)) = -W(t)^{-1}$, by combining the previous relation with 
    $\operatorname{Im}(-W^{-1}) \le (\operatorname{Im} W)^{-1}$ ,
    one gets
    \[ V(t) = \operatorname{Im} U(t) \le (\operatorname{Im} W(t))^{-1} < y^{-1} D_0^{2}\ ; \]
    then, $D_0^2-yV(t)>0$, and since $V(t)>0$, it follows from 
    Lemma \ref{89} that 
    \[ \|(V(t) + yI)^{-1/2}V(t)^{1/2}\|^2 < \frac{\|D_0\|^2}{\|D_0\|^2 + y^2}\ . \]

    By replacing this into (\ref{3.19}), one gets
    \[ F_{(U(t)+zI)}(\dot{U}(t)) \le \frac{\|D_0\|^2}{\|D_0\|^2 + y^2}\ F_{U(t)}(\dot{U}(t)),  \]
   and integration over $t \in [0,1]$ yields
    \[ \int_{0}^{1} F(U(t)+zI)(\dot{U}(t)) dt \le \frac{\|D_0\|^2}{\|D_0\|^2 + y^2}\int_{0}^{1} F(U(t))(\dot{U}(t)) dt. \]

    Finally, by taking the infimum over all such paths $U(t)$ (remember that $\frac{d}{dt}(U(t)+zI)=\dot{U}(t)$), it follows from~\eqref{3.18} that 
    \begin{eqnarray*}
      d_{\infty}(T_{-}(n, z)W_{1}, T_{-}(n, z)W_{2})&=& d_{\infty}(T_{z}U_{1}, T_{z}U_{2}) \\ &\le& \frac{\|D_0\|^2}{\|D_0\|^2 + y^2}\ d_{\infty}(U_{1}, U_{2})\\
      &=&\frac{\|D_0\|^2}{\|D_0\|^2 + y^2}\ d_{\infty}(W_{1}, W_{2}). \end{eqnarray*}
  
    \end{proof}

    Set, for each $n\in\mathbb{N}$ and each $z\in\mathbb{C}$, $P_{\pm}(n, z) := T_{\pm}(n, z) \cdots T_{\pm}(1, z)$. Since for each $j = 1,\ldots, n$, $T_{\pm}(j, z)\in SL(2l, \mathbb{C})$, and since $SL(2l, \mathbb{C})$ is a group, it follows that $P_{\pm}(n, z) \in SL(2l, \mathbb{C})$ ($SL(2l, \mathbb{R})$ if $z\in\mathbb{R}$). 

    \begin{lema} \label{85A} Let $n\in\mathbb{Z}$ and let $z \in \mathbb{C}$. Then, for each $M\in\mathfrak{S}_l$,
      \[P_+(n,z)M=P_-(n,z)M.\]
      \end{lema}
    \begin{proof} It follows from Lemma~\ref{85}-(2) that for each $n\in\mathbb{Z}$ and each $z\in\mathbb{C}$,
\begin{eqnarray*}\left(\begin{array}{cc}
			I & \ 0 \\
			0 & -I \\
	   \end{array}\right)P_+(n,z)&=&\left(\begin{array}{cc}
			I & \ 0 \\
			0 & -I \\
	   \end{array}\right)T_+(n)\left(\begin{array}{cc}
			I & \ 0 \\
			0 & -I \\
	   \end{array}\right)\left(\begin{array}{cc}
			I & \ 0 \\
			0 & -I \\
  \end{array}\right)T_+(n-1)\cdots T_+(1)\\
  &=&T_-(n)\cdots T_-(1)\left(\begin{array}{cc}
			I & \ 0 \\
			0 & -I \\
	   \end{array}\right)= P_-(n,z)\left(\begin{array}{cc}
			I & \ 0 \\
			0 & -I \\
  \end{array}\right),\end{eqnarray*}
and so, 
       $$-P_+(n,z)M=\left(\begin{array}{cc}
			I & \ 0 \\
			0 & -I \\
	   \end{array}\right)P_+(n,z)M=
P_-(n,z)\left(\begin{array}{cc}
			I & \ 0 \\
			0 & -I \\
	   \end{array}\right)M=-P_-(n,z)M\, .$$
      \end{proof}

    The next result is a direct consequence of Theorem~\ref{88}.
    
\begin{cor}\label{88P}
    Let $n\in\mathbb{N}$ and let $t \in \mathbb{R}$. Then, for each $W_{1}, W_{2}\in\mathfrak{S}_l$, 
    \[ d_{\infty}(P_{+}(n, t)W_{1}, P_{+}(n, t)W_2) = d_{\infty}(W_{1},W_{2}). \]
    \end{cor}

    \begin{cor}\label{DDEC}
    Let $n\in\mathbb{N}$, $z \in \mathbb{C}_{+}$ and $W\in\mathfrak{S}_l$. Then,
    \[ d_{\infty}(M_{-}(n, z), P_{-}(n, z)W) \le \frac{\|D_0\|^{2n}}{\left(\|D_0\|^2+y^{2}\right)^n}\ d_{\infty}(M_{-}(0, z), W). \]
    \end{cor}
    
    \begin{proof}
      It follows from Lemmas~\ref{84}, \ref{90} and Proposition \ref{86} that for each $z \in \mathbb{C}_{+}$ and each $n \in \mathbb{N}$, $M_{-}(n, z) = P_{-}(n, z)M_{-}(0, z)$, with 
      $M_{-}(n,z) \in \mathfrak{S}_l$. The result is now a direct consequence of 
      Theorem \ref{91}.
    \end{proof}


    \section[Harmonic measures and value distribution]{Matrix-valued harmonic measures, value distribution and convergence}\label{108}

   This section shifts the focus from the geometric contraction of transfer matrices to a measure-theoretic description of the boundary behavior (as $\epsilon\downarrow 0$) of the related $m$-functions. More specifically, by recalling the matrix-valued harmonic measures defined in Introduction, we provide the tools needed to analyze how the iterative dynamics in the Siegel upper half-plane translates into specific value distributions for the $m$-functions near the real axis.

   \vspace{0.3cm}

    \subsection{Matrix-valued harmonic measures and value distribution}
   
   Recall from (\ref{117}) the \textit{matrix-valued harmonic measures}:  for each $S \in \mathcal{B}(\mathbb{R})$, each $z\in\mathbb{C}_+$ and each Herglotz matrix-valued function $F\in\mathcal{H}$,
    $$\omega_{F(z)}^{I}(S):= \frac{1}{\pi}\int_S \operatorname{Im}\Big(\lambda I-i((\operatorname{Im}F(z))^{1/2}+\text{I})\Big)^{-1}\ d\lambda$$
    and
    $$\omega_{F(z)}^{R}(S):= \frac{1}{\pi}\int_S \operatorname{Im}\left(\frac{1}{\Vert\operatorname{Im}F(z)+I\Vert}\operatorname{Re}F(z)-\lambda I+iI\right)^{-1}\ d\lambda.$$

    \begin{lema}\label{47}
      Let $A, S \in \mathcal{B}(\mathbb{R})$, with $|A|<\infty$, let $F\in\mathcal{H}$  and let $\alpha=R,I$ stand for real or imaginary. Then, 
      the following identities hold: 
      
    \begin{enumerate}
    \item $$\int_A\omega_{F(t)}^\alpha(S)\ dt=\int_S\rho_\alpha^{(\lambda)}(A)\ \frac{1}{1+\lambda^2}\ d\lambda,$$ 
        where $d\rho_\alpha^{(\lambda)}(A)=(1+t^2)\chi_A(t)d\upsilon_\alpha^{(\lambda)}(t)$ ($\upsilon_\alpha^{(\lambda)}$ stands for the measure in the Herglotz representation of $F_\alpha^{(\lambda)}$), and for each $z\in\mathbb{C}_+$, $\lambda\in S$,
 
 $$F_I^{(\lambda)}(z)=(I+i\lambda((\operatorname{Im}F(z))^{1/2}+I))(\lambda I-i((\operatorname{Im}F(z))^{1/2}+I))^{-1}\,,$$
 $$F_R^{(\lambda)}(z)= \left(I+\lambda\left(\frac{1}{\Vert\operatorname{Im}F(z)+I\Vert}\operatorname{Re}F(z)+iI\right)\right)\left(\lambda I-\left(\frac{1}{\Vert\operatorname{Im}F(z)+I\Vert}\operatorname{Re}F(z)+iI\right)\right)^{-1}.$$ 
    \item For each $z=\zeta + i\lambda\in\mathbb{C}_+$, one has 
    
    $$\omega_{F(z)}^{\alpha}(S)=\frac{1}{\pi}\int_{-\infty}^{\infty}\omega_{F(t)}^{\alpha}(S)\frac{\lambda}{(t-\zeta)^2+\lambda^2}\ dt=\int_{-\infty}^{\infty} \omega_{F(t)}^{\alpha}(S)\ d\omega_z(t),$$ 

    \noindent where $\omega_z(\cdot)$ is the distribution function of the harmonic measure given by$$\omega_z(t) = \frac{1}{\pi} \arctan\left(\frac{t-\zeta}{\lambda}\right) + \frac{1}{2},\qquad t\in\mathbb{R}.$$
        \end{enumerate}
    \end{lema}

    \begin{proof}
      \textbf{1.}        In the first part of the proof, we consider the cases $\alpha=I$ and $\alpha=R$ separately. Let $\alpha=I$. In order to prove that for each $y\in S$,
      $F_I^{(y)}\in\mathcal{H}$, one just needs to show that for each $z\in\mathbb{C}_+$, $\operatorname{Im}F_I^{(y)}(z)>0$, since $F_I^{(y)}$  
      is analytic on $\mathbb{C}_+$; from now on, we  use the notation $1/(B^2+C^2)=(B-iC)^{-1}(B+iC)^{-1}$ for matrices $B$ and $C$ that commute.

        Since $((\operatorname{Im}F(z))^{1/2}+I)$ is a real symmetric matrix, there exist an orthogonal matrix $P$  and a diagonal matrix $D$ such that $((\operatorname{Im}F(z))^{1/2}+I)=PDP^T$; then, it follows from the functional calculus for self-adjoint matrices that 
        $$F_I^{(y)}(z)=P(I+iyD)P^T(P(yI-iD)P^T)^{-1}=P(I+iyD)(yI-iD)^{-1}P^T=P\widetilde{D}P^T.$$
Since $\widetilde{D}=(I+iyD)(yI-iD)^{-1}$ is a diagonal matrix, one concludes that $F_I^{(y)}(z)$ is a normal matrix. Moreover, it follows from the functional calculus that for each $i\in\{1,\ldots,l\}$,
        $$\widetilde{D}_{ii}=\frac{1+iy(1+\sqrt{\lambda^I_{i}})}{y-i(1+\sqrt{\lambda^I_{i}})}=\frac{(1+iy(1+\sqrt{\lambda^I_{i}}))(y+i(1+\sqrt{\lambda^I_{i}}))}{y^2+(1+\sqrt{\lambda^I_{i}})^2}$$
      (here, $\{\lambda_i^I\}_{i=1}^l=\{\lambda_i^I(z)\}_{i=1}^l$ denotes the set of the eigenvalues of $\operatorname{Im}F(z)$), and so $\operatorname{Im}\widetilde{D}_{ii}=\dfrac{(1+y^2)(1+\sqrt{\lambda^I_{i}})}{y^2+(1+\sqrt{\lambda^I_{i}})^2}>0$.  Since                $$(P\widetilde{D}P^T)^*=P\widetilde{D}^*P^T \ \ \ \ \text{and} \ \ \ \ \operatorname{Im}(P\widetilde{D}P^T)=\frac{P\widetilde{D}P^T-P\widetilde{D}^*P^T}{2i}=P\operatorname{Im}\widetilde{D}P^T,$$ it follows that for each $z\in\mathbb{C}_+$,
                \begin{equation}\label{10}
            \operatorname{Im}F_I^{(y)}(z)=\frac{(1+y^2)((\operatorname{Im}F(z))^{1/2}+I)}{y^2\text{I}+((\operatorname{Im}F(z))^{1/2}+I)^2}>0\ .
        \end{equation}        

                Hence, $F_I^{(y)}$ is Herglotz and
        $$\omega_{F(z)}^{I}(S)= \frac{1}{\pi}\int_S \frac{(\operatorname{Im}F(z))^{1/2}+I}{\lambda^2\text{I}+((\operatorname{Im}F(z))^{1/2}+I)^2}\ d\lambda=\frac{1}{\pi}\int_S \frac{1}{1+\lambda^2}\ \operatorname{Im}F_I^{(\lambda)}(z)\ d\lambda.$$ 

        The proof for the case $\alpha=R$ follows the same arguments presented above, so we just omit it. 

 One can write down both identities obtained above as follows:
       \begin{equation}\label{80}
           \omega_{F(z)}^{\alpha}(S)=\frac{1}{\pi}\int_S \frac{1}{1+\lambda^2}\ \operatorname{Im}F_\alpha^{(\lambda)}(z)\ d\lambda ,\ \ \ \ \ \alpha=R,I\ .
       \end{equation}

        Note also that from the previous discussion, one can also write
        \[\omega_{F(z)}^{I}(S)=P_I(z)D_{I}(z,S)P_I(z)^T,\]
        where $P_I(z)$ is an orthogonal matrix that diagonalizes $\operatorname{Im} F(z)$, and $D_I(z,S)$ is the diagonal matrix such that for each $i\in\{1,\ldots,l\}$,
        \[(D_{I}(z,S))_{ii}:=\frac{1}{\pi}\int_S\frac{1+\lambda_{i}^I(z)^{1/2}}{(1+\lambda_{i}^I(z)^{1/2})^2+\lambda^2}\ d\lambda.\]
        Moreover, one has for each $S\in\mathcal{B}(\mathbb{R})$, each $z\in\mathbb{C}_+$ and each $i\in\{1,\ldots,l\}$ that $(D_{I}(z,S))_{ii}\le (D_{I}(z,\mathbb{R}))_{ii} =1${\footnote{A direct calculation shows that $\frac{1}{\pi} \int_{\mathbb{R}} \frac{1+\lambda_{i}(z)^{1/2}}{(1+\lambda_{i}(z)^{1/2})^2+\lambda^2} d\lambda = 1$, as the integrand corresponds to the probability density function of a Cauchy distribution with scale parameter $A = 1+\lambda_{i}(z)^{1/2}$.}}, from which follows that $\Vert \omega_{F(z)}^{I}(S)\Vert\le 1$.

        By the same reasoning, one concludes that for each $S\in\mathcal{B}(\mathbb{R})$ and each $z\in\mathbb{C}_+$, there exists an orthogonal matrix $P_R(z)$ and a diagonal matrix  $D_R(z,S)$ such that 
        \[\omega_{F(z)}^{R}(S)=P_R(z)D_{R}(z,S)P_R(z)^T\ ,\] where 
        $$(D_{R}(z,S))_{ii} := \frac{1}{\pi} \int_S \frac{1}{(\lambda - \lambda_{i}^R(z))^2 + 1}\ d\lambda$$ 
        
        \noindent and $(D_{R}(z,S))_{ii}\le (D_{R}(z,\mathbb{R}))_{ii} = 1$\footnote{Similarly, a direct calculation shows that $\frac{1}{\pi} \int_{\mathbb{R}} \frac{1}{(\lambda - \lambda_{i}^R(z))^2 + 1} d\lambda = 1$, as the integrand corresponds to the probability density function of a Cauchy distribution with location parameter $\lambda_i^R(z)$ and unit scale parameter.} for each $i\in\{1,\ldots,l\}$; then, $\Vert \omega_{F(z)}^{R}(S)\Vert\le 1$.

        Now, it follows from this discussion, relation
        ~\eqref{80} and Fubini Theorem (recall that $|A|<\infty$) that 
                \begin{equation}\label{PEL}
            \int_A \omega_{F(t+iy)}^{\alpha}(S)\ dt =\frac{1}{\pi}\int_S \frac{1}{1+\lambda^2}\left(\int_A \operatorname{Im}F_\alpha^{(\lambda)}(t+iy)\ dt \right) d\lambda.
        \end{equation}
        
                Since
                \[\operatorname{Im}F_\alpha^{(\lambda)}(t+iy)=y\int_{\mathbb{R}_{\infty}} \frac{1+\zeta^2}{|\zeta-t-iy|^2}\ d\upsilon_\alpha^{(\lambda)}(\zeta)\] (by the Herglotz representation of $F_\alpha^{(\lambda)}$\footnote{Namely, it follows from the definition of a matrix-valued Herglotz function that
        $$\operatorname{Im}F_\alpha^{(\lambda)}(t+iy)=\int_{\mathbb{R}_{\infty}}\operatorname{Im}\left(\frac{1+\zeta (t+iy)}{\zeta-(t+iy)}\right)\ d\upsilon_\alpha^{(\lambda)}(\zeta)=\int_{\mathbb{R}_{\infty}}\frac{y(1+\zeta^2)}{(\zeta-t)^2+y^2}\ d\upsilon_\alpha^{(\lambda)}(\zeta)\ .$$}), one gets
                \begin{equation*}
            \int_A \omega_{F(t+iy)}^{\alpha}(S)\ dt =\frac{1}{\pi}\int_S \frac{1}{1+\lambda^2}\left( \int_A y \left(\int_{\mathbb{R}_{\infty}} \frac{1+\zeta^2}{|\zeta-t-iy|^2}\ d\upsilon_\alpha^{(\lambda)}(\zeta)\right)\ dt\right) d\lambda\ .\nonumber 
                \end{equation*}

                Now, for each $\lambda\in S$, it follows from Fubini Theorem that
                $$\frac{1}{\pi}\int_A\int_{\mathbb{R}_{\infty}} \frac{y(1+\zeta^2)}{|\zeta-t-iy|^2}\ d\upsilon_\alpha^{(\lambda)}(\zeta)\ dt = \int_{\mathbb{R}_{\infty}} (1+\zeta^2)\left( \frac{1}{\pi}\int_A\frac{y}{|\zeta-t-iy|^2}\ dt\right) d\upsilon_\alpha^{(\lambda)}(\zeta)\ ; $$
                thus, by combining the identity\footnote{If $A=(a,b)\subset\mathbb{R}$, then $$\lim_{y\to 0^+}\frac{1}{\pi}\int_a^b \frac{y}{(\zeta-t)^2+y^2}\ dt=\lim_{y\to 0^+}\frac{1}{\pi}\left[\arctan\left(\frac{b-\zeta}{y}\right)-\arctan\left(\frac{a-\zeta}{y}\right)\right]=\chi_{(a,b)}(\lambda).$$ The extension to the case where $A$ is a Borel set of finite measure follows from the fact that $A$ can be approximated by a finite union of disjoint intervals, or, equivalently, from the property of the Poisson kernel $P_y(x) = \frac{1}{\pi}\frac{y}{x^2+y^2}$ as an approximate identity, which ensures that the convolution $(P_y * \chi_A)(\zeta)$ (which is equal to the original integral in $A$) converges to $\chi_A(\zeta)$ almost everywhere as $y \to 0^+$.}
      
        \begin{equation*}
            \lim_{y\to 0^+}\frac{1}{\pi}\int_A\frac{y}{|\zeta-t-iy|^2}\ dt=\chi_A(\zeta)
        \end{equation*}
        with Fubini and Dominated Convergence Theorems\footnote{Note that there exists a positive matrix  $W^{(\lambda)}(\zeta)$ such that $d\upsilon_\alpha^{(\lambda)}(\zeta) = W^{(\lambda)}(\zeta)\ d\upsilon^{(\lambda)}_{\alpha\ tr}(\zeta)$, where $d\upsilon^{(\lambda)}_{\alpha\ tr}(\zeta)$ is the (scalar) trace measure of $d\upsilon_\alpha^{(\lambda)}(\zeta)$. Since the integral of a matrix-valued function can be viewed as a matrix whose entries are integrals of complex-valued functions, one can apply Dominated Convergence Theorem to each entry of the matrix individually.}, one gets
        \begin{align}\label{14}
          &\lim_{y\to 0^+} \frac{1}{\pi}\int_S \frac{1}{1+\lambda^2}\left(\int_{\mathbb{R}_{\infty}}  (1+\zeta^2)\left( \int_A \frac{y}{|\zeta-t-iy|^2}\ dt \right) d\upsilon_\alpha^{(\lambda)}(\zeta)\right) d\lambda \nonumber \\  &= \int_S \frac{1}{1+\lambda^2}\left( \int_{\mathbb{R}_{\infty}}  (1+\zeta^2)\ \chi_A(\zeta)\ d\upsilon_\alpha^{(\lambda)}(\zeta)\right) d\lambda \nonumber \\ &= \int_S \frac{1}{1+\lambda^2}\ \rho_\alpha^{(\lambda)}(A)\ d\lambda\,.
        \end{align}

        On the other hand, it follows from dominated convergence that for each $A,S\in\mathcal{B}(\mathbb{R})$, with $|A|<\infty$, and each $F\in\mathcal{H}$,
        \begin{equation}\label{14a}
          \lim_{y\to 0^+} \int_A \omega_{F(t+iy)}^{\alpha}(S)\ dt=\int_A \omega_{F(t)}^{\alpha}(S)\ dt.
          \end{equation}
        Namely, note that $\Vert \omega_{F(t+iy)}^{\alpha}(S)\Vert\le 1$ for each $t\in\mathbb{R}$ and each $y>0$; note also that since the eigenvalues and eigenvectors of $\operatorname{Im}F(z)$ and $\operatorname{Re}F(z)$ depend continuously on $z\in\mathbb{C}_+$, and that $\lim_{y\downarrow 0} F(t+iy)=F(t)$ exists for a.e. $t\in A$ (given that $F\in\mathcal{H}$), it follows from dominated convergence that $\omega_{F(t)}^{\alpha}(S)=\lim_{y\downarrow 0}\omega_{F(t+iy)}^{\alpha}(S)$ for a.e. $t\in A$, yielding~\eqref{14a}. 

        Thus, by combining relations~\eqref{PEL},~\eqref{14} and \eqref{14a}, one concludes that 
        $$\int_A\omega_{F(t)}^\alpha(S)\ dt=\int_S\rho_\alpha^{(\lambda)}(A)\ \frac{1}{1+\lambda^2}\ d\lambda,$$ 
        where $d\rho_\alpha^{(\lambda)}(A)=(1+t^2)\chi_A(t)d\upsilon_\alpha^{(\lambda)}(t)$ and $\upsilon_\alpha^{(\lambda)}$ is the measure in the Herglotz representation of $F_\alpha^{(\lambda)}$.

       \textbf{2.} In order to prove the second result, one needs to extend the previous identity to the case where $\chi_A(t)$ is replaced by a Borel function $p(t)$ (which can be done by dominated convergence and by noting the fact that $p(t)$ can be seen as the pointwize limit of a sequence $(s_n(t))$ of simple functions; again, it suffices to use the identity $d\upsilon_\alpha^{(\lambda)}(t) = W^{(\lambda)}(t)\ d\upsilon^{(\lambda)}_{\alpha\ tr}(t)$): 
                \begin{equation}\label{81}
          \int_{-\infty}^{\infty}p(t)\ \omega_{F(t)}^\alpha(S)\ dt=  \int_S \frac{1}{1+\lambda^2}\left(\int_{\mathbb{R}_{\infty}}(1+t^2)p(t)d\upsilon_\alpha^{(\lambda)}(t)\right)d\lambda\, .
        \end{equation}
        
         By setting $p(t)=\frac{1}{\pi}\operatorname{Im}\frac{1}{t-z}=\frac{1}{\pi}\frac{\lambda}{(t-\zeta)^2+\lambda^2}$ in~\eqref{81}, it follows from relation~\eqref{80} and from the following Herglotz representation of $F_\alpha^{(\lambda)}$, 
                         \begin{equation*}
          \frac{1}{\pi}\operatorname{Im}F_\alpha^{(\lambda)}(\zeta+iy)=  \frac{1}{\pi}\int_{\mathbb{R}_{\infty}}(1+t^2)\frac{y}{(t-\zeta)^2+y^2}d\upsilon_\alpha^{(\lambda)}(t),
        \end{equation*}
that, for each $z=\zeta+iy\in\mathbb{C}_+$,        
        %
        \begin{align}
            \omega_{F(z)}^{\alpha}(S)&=\frac{1}{\pi}\int_S \frac{1}{1+\lambda^2}\ \operatorname{Im}F_\alpha^{(\lambda)}(z)\ d\lambda \nonumber \\ &=\frac{1}{\pi}\int_S\frac{1}{1+\lambda^2}\left(\int_{\mathbb{R}_{\infty}}(1+t^2)\frac{y}{(t-\zeta)^2+y^2}\ d\upsilon_\alpha^{(\lambda)}(t)\right) d\lambda\ \nonumber \\ &=\frac{1}{\pi}\int_{-\infty}^{\infty}\omega_{F(t)}^{\alpha}(S)\frac{y}{(t-\zeta)^2+y^2}\ dt \nonumber \\ &=\int_{-\infty}^{\infty}\omega_{F(t)}^{\alpha}(S)\ d\omega_z(t)\,.
        \end{align}
   \end{proof}

   \subsection{Convergence of matrix-valued harmonic measures}\label{120}

    In this subsection, we discuss two important results concerning the convergence of matrix-valued harmonic measures.
    
    The next result says that given $A\in\mathcal{B}(\mathbb{R})$ with $|A|<\infty$, the difference $$\int_A \omega_{F(t+iy)}^{\alpha}(S)\ dt\ -\int_A \omega_{F(t)}^{\alpha}(S)\ dt$$ converges to zero (when $y\to 0^+$) uniformly over $S\in\mathcal{B}(\mathbb{R})$ and $F\in\mathcal{H}$.

    \begin{pro}\label{19}
        Let $A\in\mathcal{B}(\mathbb{R})$, with $|A|<\infty$. Then,  $$\lim_{y\to0^+}\sup_{F\in\mathcal{H};S\subset\mathbb{R}}\left\|\int_A \omega_{F(t+iy)}^{\alpha}(S)\ dt\ -\int_A \omega_{F(t)}^{\alpha}(S)\ dt\right\| = 0,$$ 
with $\alpha=I$ or $\alpha=R$.
    \end{pro}
    \begin{proof}
It follows from Lemma~\ref{47} that for each $S\in\mathcal{B}(\mathbb{R})$ and each $z\in\mathbb{C}_+$, $$\omega_{F(t+iy)}^{\alpha}(S)= \frac{1}{\pi}\int_{-\infty}^{\infty} \omega_{F(u)}^{\alpha}(S)\frac{y}{(u-t)^2+y^2}\ du.$$

   Now, one has from Fubini Theorem that for each $A\in\mathcal{B}(\mathbb{R})$ with $|A|<\infty$, 
        \begin{align*}
            \int_A \omega_{F(t+iy)}^{\alpha}(S)\ dt &=\frac{1}{\pi}\int_A \left(\int_{-\infty}^{\infty} \frac{y}{(u-t)^2+y^2}\ \omega_{F(u)}^{\alpha}(S)\ du\right)dt \nonumber\\ &=\int_{-\infty}^{\infty}  \omega_{F(u)}^{\alpha}(S)\left( \frac{1}{\pi}\int_A \frac{y}{(u-t)^2+y^2}\ dt\right)du \nonumber\\ &=\int_{-\infty}^{\infty} \omega_{u+iy}(A)\ \omega_{F(u)}^{\alpha}(S)\ du\ .
        \end{align*}

   Therefore,
\begin{align*}
    &\left\|\int_A \omega_{F(t+iy)}^{\alpha}(S)\ dt\ - \int_A \omega_{F(t)}^{\alpha}(S)\ dt\right\| 
    =\left\|\int_{-\infty}^\infty \omega_{F(t)}^{\alpha}(S)\ \Big(\omega_{t+iy}(A)-\chi_A(t)\Big)\ dt \right\|
\end{align*} \vspace{-0.7cm}
\begin{align*}
    \ \ \ \ \ \ \ \ \ \ \ \ \ \ \ \ \ &=\left\|\int_{A^c} \omega_{F(t)}^{\alpha}(S)\ \omega_{t+iy}(A)\ dt - \int_{A} \omega_{F(t)}^{\alpha}(S)\ \omega_{t+iy}(A^c)\ dt \right\| \\
    &\le \max \left( \left\|\int_{A^c} \omega_{F(t)}^{\alpha}(S)\ \omega_{t+iy}(A)\ dt \right\| , \left\|\int_{A} \omega_{F(t)}^{\alpha}(S)\ \omega_{t+iy}(A^c)\ dt \right\| \right) \\
    &\le \max \left( \int_{A^c} \omega_{t+iy}(A)\ dt , \int_{A} \omega_{t+iy}(A^c)\ dt \right) \\
    &=\int_A \omega_{t+iy}(A^c)\ dt\ ;
\end{align*}
\noindent namely, the second equality is obtained by splitting the integral over $\mathbb{R}$ into $A$ and $A^c$, and by noting that for each $t\in A$, $\omega_{t+iy}(A) - \chi_A(t) = \omega_{t+iy}(A) - 1 = -\omega_{t+iy}(A^c)$, since $\omega_z$ is a probability measure; the first inequality follows from the fact that both integrals represent positive semi-definite matrices; the second inequality follows because $\Vert \omega_{F(z)}^{\alpha}(S)\Vert\le 1$; finally, by Fubini Theorem and the definition of $\omega_z$, one gets $\int_{A^c} \omega_{t+iy}(A) dt = \int_A \omega_{t+iy}(A^c) dt$, making both terms in the maximum equal to each other.

Set, for each $y>0$, $\epsilon_A(y):=\int_A \omega_{t+iy}(A^c)\ dt\ ;$ one needs to show that 
$\lim_{y\downarrow 0}\epsilon_A(y)= 0$. Recall, from Lebesgue’s Differentiation Theorem,
\begin{align}
        \omega_{t+iy}(A^c) &\leqslant \frac{1}{\pi}\int_{A^c \cap (t-Ny,t+Ny)}\frac{y}{(s-t)^2+y^2}\ ds + \frac{1}{\pi}\int_{|s-t|\geqslant Ny}\frac{y}{(s-t)^2+y^2}\ ds \nonumber \\ &=No(1)+1-\frac{2}{\pi}\arctan N. \nonumber
    \end{align}

Thus, by taking $y$ small enough and by noting that $N > 0$ is arbitrary, one gets  \linebreak $\lim_{y\downarrow 0}\omega_{t+iy}(A^c)=0$ for a.e. $t \in A$, from which follows, by dominated convergence,  that $\epsilon_A(y)\to 0$. 
    \end{proof}

    \vspace{0.5cm}

     The next result establishes a sufficient condition for a sequence of Herglotz functions $(F_n)$ to converge to a Herglotz function $F$ \textit{in the value distribution sense}, that is, for each $A\in\mathcal{B}(\mathbb{R})$ with $|A|<\infty$, and each bounded $S\in\mathcal{B}(\mathbb{R})$,
      \[\lim_{n\to\infty}\int_A\omega^\alpha_{F_n(t)}(S)\ dt=\int_A\omega^\alpha_{F(t)}(S)\ dt.\]

\begin{obs} We highlight that this definition of convergence in value distribution is weaker than the one presented by Pearson and Breimesser in~[\ref{Breimesser1}]; here, $S\in\mathcal{B}(\mathbb{R})$ must be bounded, and in their, it can be any Borel subset of $\mathbb{R}$. Such restriction is due to the fact that the matrix-valued harmonic measures $\omega^\alpha_{F(z)}(S)$ are only continuous in $\mathfrak{S}_l$ (with respect to the matrix norm), as far as we know, in case $S\in\mathcal{B}(\mathbb{R})$ is bounded. We refer to the proof of Lemma~\ref{28.} for details. 
\end{obs}

    \begin{pro}\label{CVD} Let $\{F,(F_n)\}\subset\mathcal{H}$ be such that for each compact $K\subset \mathbb{C}_+$ and each $z\in K$, $\lim_{n\to\infty}\Vert F_n(z)-F(z)\Vert=0$. Then, $(F_n)$ converges to $F$ in the value distribution sense.
      \end{pro}
    \begin{proof}    Let $A\in\mathcal{B}(\mathbb{R})$ be such that $|A|<\infty$, and let $0<\epsilon<1/2$. It follows from Proposition \ref{19} that there exists $y_0=y_0(\epsilon,A)>0$ such that for each $y\in(0,y_0]$, each $n\in\mathbb{N}$ and each $S\in\mathcal{B}(\mathbb{R})$,
                \begin{equation}\label{129}\left\|\int_A \left(\omega_{F_n(t)}^{\alpha}(S)- \omega_{F_n(t+iy)}^{\alpha}(S)\right)\ dt\right\|<\epsilon\ ,\ \ \ \left\|\int_A \left(\omega_{F(t)}^{\alpha}(S)- \omega_{F(t+iy)}^{\alpha}(S)\right)\ dt\right\|<\epsilon\ .\end{equation}

        Fix $y_0>0$ as above,  choose a compact set $K_\mathbb{R}\subset A$ so that $|A\setminus K_\mathbb{R}|<\epsilon$ (such compact set exists, since the Lebesgue measure is regular), and set $K:=K_\mathbb{R}+iy_0=\{z=x+iy_0\mid x\in K_{\mathbb{R}}\}$. This set $K \subset \mathbb{C}_+$ is clearly compact.

        Then, by hypothesis, there exists $n_0=n_0(\epsilon,K)$ such that for each $n\geqslant n_0$ and each $z\in K$, $\Vert F_n(z)-F(z)\Vert<\epsilon$. It follows  from Lemma \ref{28.} that there exists a function $v_\alpha:\mathbb{R}_+\times\mathbb{R}_+\rightarrow\mathbb{R}_+$, with $\lim_{s\downarrow 0}v_\alpha(s,w)=0$ for each $w\in\mathbb{R}_+$, such that for each $x\in K_{\mathbb{R}}$, each  $C>0$ and each $S\in\mathcal{B}((-C,C))$, 
                \begin{equation}\label{29}
            \left\|\omega_{F_n(x+iy_0)}^\alpha(S)-\omega_{F(x+iy_0)}^\alpha(S)\right\|<v_\alpha(\epsilon,C).
        \end{equation}

               Therefore, it follows from relations~\eqref{129} and~\eqref{29} that for each fixed $n\ge n_0$, $C>0$, and $S\in\mathcal{B}((-C,C))$, 
        \begin{align}
             &\left\|\int_A \left(\omega_{F_n(t)}^{\alpha}(S)- \omega_{F(t)}^{\alpha}(S)\right)\ dt\right\| \leqslant \nonumber \\ &\leqslant\left\|\int_A \left(\omega_{F_n(t)}^{\alpha}(S)- \omega_{F_n(t+iy_0)}^{\alpha}(S)\right)\ dt\right\|+\left\|\int_{A\setminus K_\mathbb{R}} \left(\omega_{F_n(t+iy_0)}^{\alpha}(S)- \omega_{F(t+iy_0)}^{\alpha}(S)\right)\ dt\right\| \nonumber \\ &+\left\|\int_{K_\mathbb{R}} \left(\omega_{F_n(t+iy_0)}^{\alpha}(S)- \omega_{F(t+iy_0)}^{\alpha}(S)\right)\ dt\right\|+\left\|\int_A \left(\omega_{F(t+iy_0)}^{\alpha}(S)- \omega_{F(t)}^{\alpha}(S)\right)\ dt\right\|\leqslant \nonumber \\ &\leqslant 2\epsilon+2|A\setminus K_\mathbb{R}|+v_\alpha(\epsilon,C)|K_\mathbb{R}|\ , \nonumber
        \end{align}
        where one has used that for each $z\in\mathbb{C}_+$ and each $H\in\mathcal{H}$, $\|\omega_{H(z)}^\alpha(S)\|\leqslant1$. 
        Since $0<\varepsilon<1/2$ and $A\in\mathcal{B}(\mathbb{R})$, with $|A|<\infty$, were arbitrarily chosen, this proves that
        $$\lim_{n\to\infty}\int_A\omega^\alpha_{F_n(t)}(S)\ dt=\int_A\omega^\alpha_{F(t)}(S)\ dt.$$
    \end{proof}


     \section[Potential spaces and the reflectionless property of the omega-limit set]{Potential spaces and the reflectionless property of the $\omega$-limit set}\label{102}


    \vspace{0.3cm}
    \subsection{Topology of potential spaces}\label{97}
    
   Let, for each $C>0$, $\mathcal{V}^C$ be the space of bounded sequences $(D_n,V_n)$ such that for each $n\in\mathbb{Z}$, $$(C+1)^{-1}\leqslant s_l(D_n)\leqslant s_1(D_n)\leqslant C+1, \qquad 0\le s_1(V_n)\leqslant C.$$
    
     Let $\mathcal{P}_l(\mathbb{R},(C+1)^{-1},C+1)$ denote the set of $l\times l$ positive definite matrices whose singular values{\footnote{In this case, their singular values coincide with their eigenvalues.}} belong to the compact interval $[(C+1)^{-1},C+1]$. 

     We induce on $\mathcal{P}_l(\mathbb{R},(C+1)^{-1},C+1)$ the topology of $\mathbb{R}^{l(l+1)/2}$. This set is bounded (as a subset of $\mathbb{R}^{l(l+1)/2}$) and closed, since the set of positive definite matrices whose singular values belong to the compact interval $[(C+1)^{-1},C+1]$ is closed (by the continuity of the singular values with respect to the matrix elements. Therefore, $\mathcal{P}_l(\mathbb{R},(C+1)^{-1},C+1)$ is a compact subset of $(M(l,\mathbb{R}),\rho)$ (here, $\rho$ can be taken as the metric induced by the Frobenius norm).

     Similarly, $\mathcal{S}_l(\mathbb{R},C)$ (the set of $l\times l$ bounded symmetric matrices  whose norms are bounded by $C$) is a compact subset of $(M(l,\mathbb{R}),\rho)$.

     \begin{pro}\label{convexity_pot}The sets $\mathcal{P}_l(\mathbb{R}, (C+1)^{-1}, C+1)$ and $\mathcal{S}_l(\mathbb{R}, C)$ are convex subsets of $(M(l, \mathbb{R}), \rho)$.\end{pro}\begin{proof}Firstly, let $V_1, V_2 \in \mathcal{S}_l(\mathbb{R}, C)$ and $\theta \in [0, 1]$. By the triangle inequality and by the absolute homogeneity of the Frobenius norm $\rho$, one has $$\rho(\theta V_1 + (1-\theta)V_2, 0) \leq \theta \rho(V_1, 0) + (1-\theta) \rho(V_2, 0).$$ Since $\rho(V_j, 0) \leq C$ for $j=1,2$, it follows that $\rho(\theta V_1 + (1-\theta)V_2, 0) \leq C$, thus $\theta V_1 + (1-\theta)V_2 \in \mathcal{S}_l(\mathbb{R}, C)$.
     
     Now, let $D_1, D_2 \in \mathcal{P}_l(\mathbb{R}, (C+1)^{-1}, C+1)$. The condition $(C+1)^{-1} I \le D_j \le (C+1) I$ is equivalent, for each unit vector $x \in \mathbb{R}^l$,  to
    $$(C+1)^{-1} \leq \langle x, D_j x \rangle \leq C+1, \quad j=1,2.$$

     For the convex combination $D_\theta = \theta D_1 + (1-\theta)D_2$, the linearity of the inner product yields
    $$\langle x, D_\theta x \rangle = \theta \langle x, D_1 x \rangle + (1-\theta) \langle x, D_2 x \rangle.$$
    Since $\langle x, D_\theta x \rangle$ is a convex combination of two values in $[(C+1)^{-1}, C+1]$, it remains in the same interval. As this holds for every unit vector $x$, the singular values (which coincide with the eigenvalues here) of $D_\theta$ are bounded by $(C+1)^{-1}$ and $C+1$, which proves that $D_\theta \in \mathcal{P}_l(\mathbb{R}, (C+1)^{-1}, C+1)$.
    \end{proof}
    
    \vspace{0.5cm}
    Thus, $\mathcal{V}^C = \prod_{n\in\mathbb{Z}}\mathcal{P}_l(\mathbb{R},(C+1)^{-1},C+1)\times\mathcal{S}_l(\mathbb{R},C)$ is compact in the product topology. Such  topology is metrizable, and a possible choice for the metric is given by 
    \begin{equation}\label{62}
        \rho_{\infty}(V,W)=\sum_{n=-\infty}^{\infty}2^{-|n|}\left(\rho\left(D_n^V,D_n^W\right) + \rho\left(V_n^V,V_n^W\right)\right)\ ,
    \end{equation}
  with $V,W:\mathbb{Z}\to \mathcal{P}_l(\mathbb{R},(C+1)^{-1},C+1)\times\mathcal{S}_l(\mathbb{R},C)$ given by $V(n)=(D_n^V,V_n^V)$ for each $n\in\mathbb{Z}$ (and similarly for $W$).
    
  Now, for each $V:A\to \mathcal{P}_l(\mathbb{R},(C+1)^{-1},C+1)\times\mathcal{S}_l(\mathbb{R},C)$,  $W:B\to \mathcal{P}_l(\mathbb{R},(C+1)^{-1},C+1)\times\mathcal{S}_l(\mathbb{R},C)$, with $A, B \subset\mathbb{Z}$, one defines 
    \begin{equation}\label{63a}
        \rho_{\infty}(V,W)=\sum_{n\in A\cap B}2^{-|n|}\left(\rho\left(D_n^V,D_n^W\right) + \rho\left(V_n^V,V_n^W\right)\right)\ . 
    \end{equation}

    \begin{obs}
        In order to make the notation less cumbersome, set for each $n\in\mathbb{Z}$
      $$\rho\left(V(n),W(n)\right):=\rho\left(D_n^V,D_n^W\right) + \rho\left(V_n^V,V_n^W\right).$$
        Then, one can write 
         $$\rho_{\infty}(V,W)=\sum_{n=-\infty}^{\infty}2^{-|n|}\ \rho\left(V(n),W(n)\right)\ .$$      
    \end{obs}

    One may refer to each $V\in\mathcal{V}^C$ as a potential, or even as a whole-line potential. If $V\in\mathcal{V}^C$, one writes $V_{\pm}$ for the restrictions of $V$ to $\mathbb{Z}_\pm$, and one denotes the sets of such restrictions by
    \[\mathcal{V}_{\pm}^C=\{V_\pm\mid V\in\mathcal{V}^C\}.\]

    Naturally, \eqref{63a} becomes, in this case, 
\begin{equation}\label{63b}
  \rho_{\infty}(V_\pm,W_\pm)=\sum_{n\in \mathbb{Z}_\pm}2^{-|n|}\rho\left(V_\pm(n),W_\pm(n)\right)\ . 
    \end{equation}

 Now, let $V\in\mathcal{V}^C$ be a potential defined in $\mathbb{Z}_+$ (that is, a half-line potential). Then, set $$\omega(V):=\{W\in \mathcal{V}^C\mid \text{there is a sequence}\ n_j\to\infty\ \text{such that}\ \rho_{\infty}(S^{n_j}V,W)\to 0\},$$ where for each $k,n\in\mathbb{Z}_+$, $(S^kV)(n)=V(n+k)$.

    \begin{pro}\label{34}
 Let $V\in\mathcal{V}^C_+$ for some $C>0$. Then, the set $\omega(V)\subset\mathcal{V}^C$ is compact, nonempty, and S is a homeomorphism in $\omega(V)$. Furthermore, $\lim_{n \to \infty}\rho_{\infty}(S^{n}V,\omega(V))\to 0$.
    \end{pro}

    \begin{proof} Since the proof is identical to the proof of Proposition~1.3 in~\cite{remling}, we omit it.
%
  %
    \end{proof}

    \vspace{0.3cm}
    \subsection[Weyl-Titchmarsh matrix-valued m-functions for the restrictions J+ and J-]{Weyl-Titchmarsh matrix-valued $m$-functions for the restrictions $J_{\pm}$}\label{114}
    
    We are now interested in showing that 
    the maps $$\mathcal{V}^C_\pm\to \mathcal{H},\ \ \ \ W_\pm \mapsto M_\pm^W(0,\cdot)\,,$$ are homeomorphisms onto their images, where we endow $\mathcal{V}^C_\pm$ with the metric defined by~\eqref{63b} and $\mathcal{H}$ with the metric of the uniform convergence on compact sets of $\mathbb{C}_+$.

    Here, $M_\pm^W(0,z)$ stand for the Weyl-Titchmarsh $m$-functions $M_{\pm}(n,z)$ for $n=0$, $z\in\mathbb{C}_+$, and for the Jacobi matrix-valued operator  $J$ associated with $W\in\mathcal{V}^C$. Such functions $M_{\pm}^W(0,z)$ are the Weyl-Titchmarsh $m$-functions of the operator $J$ restricted to the half-lines $\mathbb{Z}_\pm$, and will be denoted simply by $M_\pm(z)$ (see Subsection~\ref{SMM} for details). 

Recall from Proposition~\ref{3.3} and Lemma~\ref{84} (in case $n=0$) that for each $z\in\mathbb{C}_+$,
\begin{equation}\label{MeJ} M_+(z)=\langle\Delta_{1},(J_{+}-zI)^{-1}\Delta_{1}\rangle\, ,
\end{equation}
and (after performing some algebraic manipulations) that
\begin{equation*} M_-(z)=D_0^{-1}\left(zI-V_0+D_{-1}\langle\Delta_{-1},(J_{-}-zI)^{-1}\Delta_{-1}\rangle\ D_{-1}\right)D_0^{-1} 
\end{equation*}
(compare such identities with the ones presented in page 15 in~[\ref{remling}]).

Then, by Spectral Theorem, there exists a positive definite matrix-valued Borel measure of compact support $\mu_+$ such that for each $i,j\in\{1,\ldots,l\}$,
$$\left(M_+(z)\right)_{ij}=\langle\delta^1_i,(J_+-zI)^{-1}\delta^1_j\rangle=\int \frac{d\langle e_i,\mu_+(\lambda)e_j\rangle}{\lambda-z}\, ,$$ where $\{e_1,\ldots,e_l\}$ stands for the canonical basis of $\mathbb{C}^l$. It follows again from Spectral Theorem that for each $n\in\mathbb{N}$ and each $a,b\in\mathbb{C}^l$, $$\langle\delta_1\otimes a,(J_+)^n\delta_1 \otimes b\rangle=\int \lambda^n d\langle a,\mu_+(\lambda)b\rangle\, .$$

\begin{obs}\label{OBSFUND} In particular, the set
\begin{equation*}
  \Sigma_{ac,l}(J_+):=\{\lambda\in\mathbb{R}\mid\lim_{\epsilon\to 0^+} M_+(\lambda+i\epsilon)\ \textrm{exists\ finitely}, \ \operatorname{rank}(\operatorname{Im}(M_+(\lambda+i0))) = l\}
  \end{equation*}
is an essential support for the absolutely continuous spectrum of multiplicity $l$ of $J_+$. 
This fact plays a fundamental role in the results of this section and of the following one.
\end{obs}

In the same  vein, there exists another positive definite matrix-valued Borel measure of compact support $\mu_-$ such that for each $n\in\mathbb{N}$ and each $a,b\in\mathbb{C}^l$, $$\langle\delta_{-1}\otimes a,(J_-)^n(\delta_{-1} \otimes b)\rangle=\int \lambda^n d\langle a,\mu_-(\lambda)b\rangle\, ,$$
and so (by Spectral Theorem), \[M_-(z)=D_0^{-1}\left(zI-V_0+D_{-1}\int\frac{d\mu_-(\lambda)}{\lambda-z}\ D_{-1}\right)D_0^{-1}.\]
    
    \begin{lema}\label{32}
      The maps $$\mathcal{V}^C_\pm\to \mathcal{H},\ \ \ \ W_\pm \mapsto M_\pm(\cdot),$$ are homeomorphisms onto their images. 
    \end{lema}

    \begin{proof}
      We just prove the result for $W_+$ and $M_+$; the proof for $W_-$ and $M_-$ is almost identical to this one. Firstly, we show that the map is bijective.

      Naturally, the map $W_+\mapsto M_+$ is well defined: given a (half-line) potential $W_+\in\mathcal{V}_+^C$, there exists a unique Jacobi operator $J_+$ defined by $W_+$, which by its turn determines, by relation~\eqref{MeJ}, a unique matrix-valued Weyl-Titchmarsh $m$-function $M_+$. Now, one needs to show that the map $M_+\mapsto W_+$ is also well defined. 

      Given that the matrix-valued measure $\mu_+$ has bounded support, it follows from relation~\eqref{MeJ} that one can write 
                \begin{equation}\label{30}
            M_+(w^{-1})=-\sum_{n\geqslant0}w^{n+1}\int \lambda^n\ d\mu_+(\lambda)\ 
        \end{equation}
                on a disk centered at $w=0$, and so $M_+(w^{-1})$ determines $\langle\delta_1\otimes a,(J_+)^n\delta_1 \otimes b\rangle$ for each $n\in\mathbb{N}$.

                On the other hand, 
                $$J_+(\delta_1 \otimes a)=\delta_2 \otimes D_1a+\delta_1 \otimes V_1a\, ,$$ and so, $$\langle b, V_1a\rangle\ =\ \langle\delta_1 \otimes b, J_+(\delta_1 \otimes a)\rangle\ ;$$hence, $\langle\Delta_1, J_+\Delta_1\rangle$ determines $V_1$. Now, one has for each $a\in\mathbb{C}^l$ that
\begin{align}
            J_+^{2}(\delta_1 \otimes a)&=\delta_1\otimes D_1^2a+\delta_3\otimes D_2D_1a \nonumber \\ &+\delta_2\otimes D_1V_1a+\delta_2\otimes V_2D_1a+\delta_1\otimes V_1^2a \nonumber \\ &=\delta_1\otimes \left(D_1^2+V_1^2\right)a \nonumber \\ &+\delta_2\otimes \left(D_1V_1+V_2D_1\right)a+\delta_3\otimes D_2D_1a\, , \nonumber
        \end{align}
and so        
$$\langle b,D_1^2a\rangle\ =\ \langle\delta_1 \otimes b, J_+^{2}(\delta_1 \otimes a)\rangle-\langle b,V_1^2a\rangle\ ,\qquad b\in\mathbb{C}^l\, .$$ 

Since $D_1$ is positive definite, it is the unique solution to the equation $A^2=D_1^2$. 
Therefore, $D_1^2$ completely determines $D_1$, hence $\langle\Delta_1, J^2_+\Delta_1\rangle$  and $V_1$ determine $D_1$.

One has, for each $a\in\mathbb{C}^l$, 
\begin{align}
            J^{3}_+(\delta_1 \otimes a)&=\delta_1\otimes (D_1^2V_1+D_1V_2D_1)a+\delta_3\otimes (D_2D_1V_1+D_2V_2D_1)a \nonumber \\ &+\delta_2\otimes (D_1^3+D_1V_1^2)a+\delta_2\otimes (D_2^2D_1)a+\delta_4\otimes (D_3D_2D_1)a \nonumber \\ &+\delta_1\otimes (V_1D_1^2+V_1^3)a+\delta_2\otimes (V_2D_1V_1+V_2^2D_1)a+\delta_3\otimes (V_3D_2D_1)a \nonumber \\ &=\delta_1\otimes (D_1^2V_1+V_1D_1^2+V_1^3+D_1V_2D_1)a \nonumber \\ &+\delta_2\otimes (D_1^3+D_1V_1^2+D_2^2D_1+V_2D_1V_1+V_2^2D_1)a \nonumber \\ &+\delta_3\otimes (D_2D_1V_1+D_2V_2D_1+V_3D_2D_1)a \nonumber \\ &+\delta_4\otimes (D_3D_2D_1)a\ , \nonumber 
\end{align}
from which follows that for each $b\in\mathbb{C}^l$,
           \begin{equation*}
            \langle\delta_1 \otimes b, J^3_+(\delta_1 \otimes a)\rangle\ =\ \langle b,D_1^2V_1a\rangle+\langle b,V_1D_1^2a\rangle+\langle b,V_1^3a\rangle+\langle b,D_1V_2D_1a\rangle\, ; 
        \end{equation*} 
hence
        \begin{equation*}
            \langle b,D_1V_2D_1a\rangle\ =\ \langle\delta_1 \otimes b, J_+^{3}(\delta_1 \otimes a)\rangle-\langle b,V_1^3a\rangle-\langle b,D_1^2V_1a\rangle-\langle b,V_1D_1^2a\rangle\ . 
        \end{equation*}

        This shows that $V_2=D_1^{-1}\left(\langle\Delta_1, J^{3}_+\Delta_1\rangle-V_1^3-D_1^2V_1-V_1D_1^2\right)D_1^{-1}$.  One can show, in the same way, that for each $n\in\mathbb{N}$, 
        \begin{align}
            &\langle b, D_1\cdots D_{n-1}V_nD_{n-1}\cdots D_1a\rangle\ =\nonumber \\ &=\ \langle\delta_1\otimes b, (J_+)^{2n-1}(\delta_1 \otimes a)\rangle+\ \text{some function of}\ \ V_1,\ldots, V_{n-1},D_1,\ldots, D_{n-1} \nonumber
        \end{align}
and
        \begin{align}
            &\langle b, D_1\cdots D_{n-1}D_n^2D_{n-1}\cdots D_1a\rangle\ =\nonumber \\ &=\ \langle \delta_1\otimes b, (J_+)^{2n}(\delta_1 \otimes a)\rangle+\ \text{some function of}\ \ V_1,\ldots, V_n,D_1,\ldots, D_{n-1}\ . \nonumber
        \end{align}
      
 Therefore, by the same arguments presented above, one concludes that $\big\{\langle\Delta_1,(J_+)^n\Delta_1\rangle\big\}_{n\ge 1}$ inductively determines $\{V_n\}_{n\geqslant1}$ and $\{D_n\}_{n\geqslant1}$ .

 It remains to show that the map $W_+ \mapsto M_+$ is a homeomorphism; its continuity follows from identity~\eqref{30} and the fact that for each $0\le n\le 2N$, the moments $\mu_n=\int\lambda^n\ d\mu_+(\lambda)$ depend continuously on $W(1),\ldots, W(N)$. 

 Since the inverse of a continuous map between two compact metric spaces is continuous, it follows that $M_+\mapsto W_+$ is also continuous. 
    \end{proof}

    \vspace{0.3cm}
    \subsection[The reflectionless property of the omega-limit set]{The reflectionless property of the $\omega$-limit set}\label{103}

    Our main goal in this subsection is to extend Theorem 1.4 and Proposition 4.1 of [\ref{remling}] to the case of matrix-valued Jacobi operators. In order to do this, one needs to show that the matrix-valued harmonic measures $\omega^{\alpha}_{F(z)}(S)$ depend continuously on $F(z)\in\mathfrak{S}_l$  with respect to the metric $d_\infty$ and to the spectral norm. Since the proofs of these results are quite long, we have opted to present then in Appendix \ref{92}.

     \begin{lema}\label{21.}
       Let $0<\epsilon<1/2$, let $F,H \in\mathcal{H}$ and let $z\in\mathbb{C}_+$ be such that 
       $d_{\infty}(F(z),H(z)) < \epsilon$. Then,
       \begin{enumerate}
         \item there exists a function $u:\mathbb{R}_+\times\mathbb{R}_+\rightarrow\mathbb{R}_+$, with $\lim_{x\downarrow 0}u(x,y)=0$ for each $y\in\mathbb{R}_+$, such that for each $S\in\mathcal{B}(\mathbb{R})$ with $|S|<\infty$,
            $$\left\|\omega_{F(z)}^{I}(S)-\omega_{H(z)}^{I}(S)\right\|\leqslant u(\epsilon, |S|)\,;$$
            
            \item there exists a function $v:\mathbb{R}_+\times\mathbb{R}_+\rightarrow\mathbb{R}_+$, with $\lim_{x\downarrow 0}v(x,y)=0$ for each $y\in\mathbb{R}_+$, such that for each $C>0$ and each $S\in\mathcal{B}((-C,C))$, 
            $$\left\|\omega_{F(z)}^{R}(S)-\omega_{H(z)}^{R}(S)\right\|\leqslant v(\epsilon, C).$$
        \end{enumerate}
    \end{lema}

    \begin{lema}\label{28.}
      Let $0<\epsilon<1/2$, let $F,H \in\mathcal{H}$ and let $z\in\mathbb{C}_+$ be such that 
      $\|F(z)-H(z)\|<\epsilon$. Then,
        \begin{enumerate} 
        \item  there exists a function $u:\mathbb{R}_+\times\mathbb{R}_+\rightarrow\mathbb{R}_+$, with $\lim_{x\downarrow 0}u(x,y)=0$ for each $y\in\mathbb{R}_+$, such that for each $S\in\mathcal{B}(\mathbb{R})$ with $|S|<\infty$,
          \[\left\|\omega_{F(z)}^{I}(S)-\omega_{H(z)}^{I}(S)\right\|\leqslant u(\epsilon, |S|);\]
        \item there exists a function $v:\mathbb{R}_+\times\mathbb{R}_+\rightarrow\mathbb{R}_+$, with $\lim_{x\downarrow 0}v(x,y)=0$ for each $y\in\mathbb{R}_+$, such that for each $C>0$ and each $S\in\mathcal{B}((-C,C))$, 
          \[\left\|\omega_{F(z)}^{R}(S)-\omega_{H(z)}^{R}(S)\right\|\leqslant v(\epsilon, C).\]
        \end{enumerate}
    \end{lema}


    \begin{teo}\label{25}
        Let $J$ be a matrix-valued Jacobi operator with bounded coefficients. Then, for each $C>0$, each $A\in\mathcal{B}(\Sigma_{ac,l})$ such that $|A|>0$, and each $S\in\mathcal{B}((-C,C))$, one has
        \begin{equation*}
            \lim_{n\to\infty}\left\|\int_A\omega_{M_-(n,t)}^{\alpha}(-S)\ dt-\int_A\omega_{M_+(n,t)}^{\alpha}(S)\ dt\right\|=0. 
        \end{equation*}
        Moreover, the convergence is uniform in $S$.
    \end{teo}

    \begin{proof}
      Fix $C>0$, $A\in\mathcal{B}(\Sigma_{ac,l})$ with $|A|>0$, and $S\in\mathcal{B}((-C,C))$. Let 
      $\epsilon>0$ and decompose $A=\bigcup_{j=0}^N A_j$ so that, for each $t\in A\setminus A_0$, $M_+(t)=\lim_{y\to 0^+}M_+(t+iy)$ exists and $M_+(t)\in\mathfrak{S}_l$. 
      Moreover, we require the existence of $M_j \in\mathfrak{S}_l$ such that
 for each $t\in A_j$,        
        \begin{equation*}
            d_{\infty}(M_+(t),M_j)<\epsilon
        \end{equation*}
        and $|A_0|<\epsilon$. Finally, for each $j\in\{1,\ldots,N\}$, $A_j$ must be bounded.
        
        Firstly, put in $A_0$ each $t\in A$ for which $M_+(t)$ does not exist (for now, $|A_0|=0$, since $M_+\in\mathcal{H}$ (see [\ref{FGET}]) and $A\in\mathcal{B}(\Sigma_{ac,l})$).

        Now, take (sufficiently large) compact subsets $K\subset\mathfrak{S}_l$ and $K'\subset\mathbb{R}$ such that $$A_0=\{t\in A\mid M_+(t) \notin K\ \ \text{or}\ \ t \notin K'\}$$
        satisfies  $|A_0|<\epsilon$ (such subsets $K$ and $K'$ exist, given that both $\mathfrak{S}_l$ and $\mathbb{R}$ are $\sigma$-compact spaces{\footnote{$\mathfrak{S}_l$ is a homogeneous space of the symplectic group $Sp(2l,\mathbb{R}$), which is a Lie group and thus, $\sigma$-compact.}}). Subdivide $K$ into a finite number of subsets of diameter less than $\epsilon$, take inverse images of such sets by $M_+$, and intersect with $K'$ to obtain $A_j$, $j\in\{1,\ldots,N\}$. We may choose the matrices $M_j$ as follows: for each $j$, select a point $t_j \in A_j$ and set $M_j = M_+(t_j)$. By the construction of the sets $A_j$ described above, it follows that $M_j \in \mathfrak{S}_l$ and, for every $t \in A_j$, $d_{\infty}(M_+(t), M_j) < \epsilon$.

 It follows from Lemma~\ref{87} that for each $n\in\mathbb{N}$ and each $t\in A\setminus A_0$, $M_+(n,t)=P_+(n,t)M_+(t)\in\mathfrak{S}_l$. Since $P_+(n,t)$ is an isometry in $\mathfrak{S}_l$ (by Corollary~\ref{88P}), it follows from each $n\in\mathbb{N}$ and each $t\in A_j$ that $$d_{\infty}(M_+(n,t), P_+(n,t)M_j)<\epsilon\,.$$

        Hence, by Lemma~\ref{21.}, there exists a function $v_\alpha:\mathbb{R}_+\times\mathbb{R}_+$, with $\lim_{x\downarrow 0}v_\alpha(x,y)=0$ for each $y> 0$, so that 
        $$\left\|\omega_{M_+(n,t)}^{\alpha}(S)-\omega_{P_+(n,t)M_j}^{\alpha}(S)\right\|\leqslant v_\alpha\big(\epsilon, C\big);$$ thus,
        \begin{equation}\label{M+O}
          \left\| \int_{A_j}\omega_{M_+(n,t)}^{\alpha}(S)\ dt-\int_{A_j}\omega_{P_+(n,t)M_j}^{\alpha}(S)\ dt\right\| \leqslant v_\alpha\big(\epsilon, C\big)|A_j|.\end{equation}

        The next step consists in showing that for each $n\in\mathbb{N}$ and each $t\in A\setminus A_0$, 
        \begin{equation}\label{O+-}
          \omega_{P_+(n,t)M_j}^{\alpha}(S)=\omega_{-\overline{P_+(n,t)M_j}}^{\alpha}(-S)=\omega_{P_-(n,t)(-\overline{M_j})}^{\alpha}(-S).
        \end{equation}
      
        Namely, note that for each $W=X+iY\in\mathfrak{S}_l$ and each $S\in\mathcal{B}(\mathbb{R})$, 
        \begin{align}\label{122}
          \omega_{-\overline{W}}^{I}(-S)&=\frac{1}{\pi}\int_{-S}\frac{(\operatorname{Im}(-\overline{W}))^{1/2}+I}{((\operatorname{Im}(-\overline{W}))^{1/2}+I)^2+\lambda^2I}\ d\lambda \nonumber \\ &= \frac{1}{\pi}\int_{-S}\frac{Y^{1/2}+I}{(Y^{1/2}+I)^2+\lambda^2I}\ d\lambda \nonumber 
          \nonumber \\ &=\omega_{W}^{I}(S) 
        \end{align}
        \noindent and
         \begin{align}\label{123}
           \omega_{-\overline{W}}^{R}(-S)&=\frac{1}{\pi}\int_{-S}\frac{I}{\left({\|\operatorname{Im}(-\overline{W})+I\|}^{-1}\operatorname{Re}(-\overline{W})-\lambda I\right)^2+I}\ d\lambda \nonumber \\ 
           &=\frac{1}{\pi}\int_S\frac{I}{\left(\|Y+I\|^{-1}X-uI\right)^2+I}\ du;\nonumber \\
           &=\omega_{W}^{R}(S)\ ; 
        \end{align}
 in particular, %
 \begin{equation}\label{OP+P-}
   \omega_{P_+(n,t)M_j}^{\alpha}(S)=\omega_{-\overline{P_+(n,t)M_j}}^{\alpha}(-S).
   \end{equation}

Now, it follows from Lemma~\ref{85A} that for each $n\in\mathbb{N}$ and each $t\in\mathbb{R}$, 
\begin{equation}\label{OP+P-1}-\overline{P_+(n,t)M_j}=-P_+(n,t)\overline{M_j}=P_+(n,t)(-\overline{M_j})=P_-(n,t)(-\overline{M_j}).
\end{equation}

Relation~\eqref{O+-} is now  a consequence of relations~\eqref{OP+P-} and~\eqref{OP+P-1}. 

   It follows from Proposition \ref{19} that there exists $y=y(\epsilon)>0$ such that for each $F\in\mathcal{H}$ and each $j=1,\ldots,N$,
        \begin{equation}\label{20}
           \left\|\int_{A_j} \omega_{F(t+iy)}^{\alpha}(-S)\ dt\ -\int_{A_j} \omega_{F(t)}^{\alpha}(-S)\ dt\right\|\leqslant\epsilon|A_j|.
       \end{equation}

        Fix such $y>0$. Since, for each $j\in\{1,\ldots,N\}$, $K_j:=\{t+iy\mid t\in \overline{A_j}\}$ is a compact subset of $\mathbb{C}_+$ (each $A_j$ is a bounded subset of $A$, by construction), and since  for each $j\in\{1,\ldots,N\}$,  
        $\{M_-(0,z)\mid z\in K_j\}$ is a compact subset of $\mathfrak{S}_l$ (since $\mathbb{C}_+\ni z\mapsto M_-(0,z)\in\mathfrak{S}_l$ is a continuous map), it follows from Corollary~\ref{DDEC} that there exists $n_j=n_j(\epsilon,K_j)$ such that for each $n\ge n_j$ and each $t\in A_j$, 
       $$d_\infty\left(M_-(n,t+iy),P_-(n,t+iy)(-\overline{M_j})\right)<\epsilon$$
(here, one uses the fact that $\sup_{z\in K_j}d_\infty\left(M_-(0,z),-\overline{M_j}\right)<\infty$, which allows $n_j$ to be chosen uniformly over $K_j$).  Then, by Lemma~\ref{21.} again, one has for each $j\in\{1,\ldots,N\}$ and each $n\ge n_j$,  
       \begin{equation}\label{20A}
         \left\|\int_{A_j} \omega_{P_-(n,t+iy)(-\overline{M_j})}^{\alpha}(-S)\ dt\ -\int_{A_j} \omega_{M_-(n,t+iy)}^{\alpha}(-S)\ dt\right\|\leqslant v_\alpha(\epsilon,C)|A_j|.
         \end{equation}
         
 By combining~\eqref{20} with~\eqref{20A}, one gets for each $j\in\{1,\ldots,N\}$, 
\begin{eqnarray}\label{M-O}&&\left\|\int_{A_j} \omega_{F(t)}^{\alpha}(-S)\ dt\ -\int_{A_j} \omega_{H(t)}^{\alpha}(-S)\ dt\right\|\nonumber\\&\le&\left\|\int_{A_j} \omega_{F(t)}^{\alpha}(-S)\ dt\ -\int_{A_j} \omega_{F(t+iy)}^{\alpha}(-S)\ dt\right\|+\left\|\int_{A_j} \omega_{F(t+iy)}^{\alpha}(-S)\ dt\ -\int_{A_j} \omega_{H(t+iy)}^{\alpha}(-S)\ dt\right\|\nonumber\\&+&\left\|\int_{A_j} \omega_{H(t+iy)}^{\alpha}(-S)\ dt\ -\int_{A_j} \omega_{H(t)}^{\alpha}(-S)\ dt\right\|\leqslant 2\epsilon|A_j|+v_\alpha(\epsilon,C)|A_j|,\end{eqnarray}
\noindent where $F(t+iy)=P_-(n,t+iy)(-\overline{M_j})$  and $H(t+iy)=M_-(n,t+iy)$.

Finally, it follows from~\eqref{M+O},~\eqref{O+-} and~\eqref{M-O} that for each $n\geqslant n_0:=\max\{n_j\mid j\in\{1,\ldots,N\}\}$, 
        $$\left\|\int_A\left(\omega_{M_+(n,t)}^{\alpha}(S)-\omega_{M_-(n,t)}^{\alpha}(-S)\right)dt\right\|\leqslant \sum_j\left\|\int_{A_j}\left(\omega_{M_+(n,t)}^{\alpha}(S)-\omega_{M_-(n,t)}^{\alpha}(-S)\right)dt\right\|$$ $$\leqslant\sum_j\left(\left\|\int_{A_j}\left(\omega_{M_+(n,t)}^{\alpha}(S)-\omega_{P_+(n,t)M_j}^{\alpha}(S)\right)dt\right\|+\left\|\int_{A_j}\left(\omega_{P_-(n,t)(-\overline{M}_j)}^{\alpha}(-S)-\omega_{M_-(n,t)}^{\alpha}(-S)\right)dt\right\|\right)$$ $$\le \sum_j\left(2\epsilon|A_j|+2v_\alpha(\epsilon,C)|A_j|\right)\leqslant \left(2\epsilon+2v_\alpha(\epsilon,C)\right)|A|.$$
   \end{proof}

    
 Before we present the next result, recall that for each $A\in\mathcal{B}(\mathbb{R})$,
    $$\mathcal{R}(A)=\left\{W\in\bigcup_{C>0}\mathcal{V}^C\mid W\ \text{reflectionless on}\ A\right\}.$$

    \begin{teo}[The reflectionless property of the $\omega$-limit set]\label{27}
      Let $V\in\mathcal{V}^C_+$ 
      and let $\emptyset\neq\Sigma_{ac,l}$ be the essential support of the absolutely continuous spectrum of multiplicity $l$ of $J_+^V$. Then, $$\omega(V)\subset\mathcal{R}(\Sigma_{ac,l})\ .$$ 
    \end{teo}
    \begin{proof}
      Let $W\in\omega(V)$. 
      Then, there exists $n_j\xrightarrow[j\to\infty]{}\infty$ such that $\rho_\infty(S^{n_j}V,W)\xrightarrow[j\to\infty]{}0$. By Lemma \ref{32}, 
      $\|M_{\pm}(n_j,z)-M_{\pm}(z)\|\xrightarrow[j\to\infty]{}0$ uniformly on compact subsets of $\mathbb{C}_+$, where $M_{\pm}(n_j,z):=M_{\pm}^{S^{n_j}V}(0,z)$ and $M_{\pm}(z):=M_{\pm}^W(0,z)$.    It follows from Proposition~\ref{CVD} that for each $A\in\mathcal{B}(\Sigma_{ac,l})$ such that $|A|>0$, each $C>0$ and each $S\in\mathcal{B}((-C,C))$, 
        $$\lim_{j\to\infty}\int_A\omega^\alpha_{M_\pm(n_j,t)}(S)\ dt=\int_A\omega^\alpha_{M_\pm(t)}(S)\ dt.$$

        By combining this result with Theorem~\ref{25}, one gets 
        $$\int_A\omega^\alpha_{M_-(t)}(-S)\ dt=\int_A\omega^\alpha_{M_+(t)}(S)\ dt\ ,$$
        and so, given that $A\in\mathcal{B}(\Sigma_{ac,l})$ was arbitrarily chosen,  one concludes that $\omega^\alpha_{M_-(t)}(-S)=\omega^\alpha_{M_+(t)}(S)$ for a.e. $t\in\Sigma_{ac,l}$. 

        Now, for each $z\in\mathbb{C}_+$, each $F\in\mathcal{H}$ and each $U\in\mathcal{B}(\mathbb{R})$, it follows from Spectral Theorem (as discussed in the proof of Lemma \ref{47}) that one can write
        $$\omega^I_{F(z)}(U)=P_I(z)D_I(z,U)P_I(z)^T\ ,$$
where $P_I(z)$ is an orthogonal matrix that diagonalizes $\operatorname{Im} F(z)$ (whose columns are formed by the eigenvectors of $\operatorname{Im} F(z)$) and $D_I(z,U)$ is a diagonal matrix such that for each $i\in\{1,\ldots,l\}$,
        \[(D_{I}(z,U))_{ii}:=\frac{1}{\pi}\int_S\frac{1+\lambda_{i}^I(z)^{1/2}}{(1+\lambda_{i}^I(z)^{1/2})^2+\lambda^2}\ d\lambda,\]
where $\{\lambda^I_1(z),\ldots,\lambda^I_l(z)\}$ are the eigenvalues of $\operatorname{Im}F(z)$.        

Since $\lim_{y\downarrow 0}\operatorname{Im}F(t+iy)=\operatorname{Im}F(t)$ for a.e. $t\in \Sigma_{ac,l}$ (see Remark~\ref{OBSFUND}),  it follows that $\lim_{y\downarrow 0}\lambda^I_j(t+iy)=\lambda^I_j(t)$ and that  $$\lim_{y\downarrow 0}\Vert P_I(t+iy)-P_I(t)\Vert=0 ,\ \ \ \ \lim_{y\downarrow 0}\Vert D_I(t+iy,U)-D_I(t,U)\Vert=0$$ 
for a.e. $t\in\Sigma_{ac,l}$.

Moreover, one has for each $i\in\{1,\ldots,l\}$ that $(D_{I}(z,U))_{ii}\le (D_{I}(z,\mathbb{R}))_{ii} =1$, and so, by dominated convergence, one gets 

        $$\lim_{y\to0^+}\left(D_I(t+iy,U)\right)_{jj}=\frac{1}{\pi}\int_S\frac{(\lambda^I_j(t))^{1/2}+1}{((\lambda^I_j(t))^{1/2}+1)^2+\lambda^2}\ d\lambda=\left(D_I(t,U)\right)_{jj} >0.$$

By summing up this discussion, one concludes that for a.e. $t\in\Sigma_{ac,l}$,
\begin{eqnarray*}\frac{1}{\pi}\int_S\frac{(\operatorname{Im}M_+(t))^{1/2}+I}{((\operatorname{Im}M_+(t))^{1/2}+I)^2+\lambda^2I}\ d\lambda &=&P_I^+(t)D_I^+(t,S)\left(P_I^+(t)\right)^T\\
 &=&\omega^I_{M_+(t)}(S)=\omega^I_{M_-(t)}(-S)\\ &=&P_I^-(t)D_I^-(t,-S)\left(P_I^-(t)\right)^T\\
  &=&\frac{1}{\pi}\int_S\frac{(\operatorname{Im}M_-(t))^{1/2}+I}{((\operatorname{Im}M_-(t))^{1/2}+I)^2+\lambda^2I}\ d\lambda\,;
\end{eqnarray*}  
in particular, by setting $S=(-a,a)$ (with $0<a<C$) and by setting, for each $\lambda\in(-a,a)$ and a.e. $t\in\Sigma_{ac,l}$,
$$f_t(\lambda):=\frac{(\operatorname{Im}M_+(t))^{1/2}+I}{((\operatorname{Im}M_+(t))^{1/2}+I)^2+\lambda^2I}\ ,\ \ \ \ \  g_t(\lambda):=\frac{(\operatorname{Im}M_-(t))^{1/2}+I}{((\operatorname{Im}M_-(t))^{1/2}+I)^2+\lambda^2I}\ , $$ one gets
        $$\int_{-a}^af_t(\lambda)\ d\lambda=\int_{-a}^ag_t(\lambda)\ d\lambda\ .$$

   Given that $f_t(\lambda)$ and $g_t(\lambda)$ are matrix-valued continuous functions, it follows from the arbitrary choice of $0<a<C$ and from the Mean Value Theorem for integrals that $f_t(0)=g_t(0)$ for a.e. $t\in\Sigma_{ac,l}$; alternatively, it follows that there exists  $N_I\subset\Sigma_{ac,l}$, with $|N_I|=0$, such that for each $t\in \Sigma_{ac,l}\setminus N_I$, 
        \begin{equation}\label{26}
            \operatorname{Im}M_+(t)=\operatorname{Im}M_-(t). 
        \end{equation}

        Similarly, one can show that for each $t\in\Sigma_{ac,l}\setminus N_R$ (where $N_R$ is a set of zero Lebesgue measure), 
        \begin{eqnarray*}
            &&\int_S\frac{I}{\left((\|\operatorname{Im}M_+(t)+I\|)^{-1}\operatorname{Re}M_+(t)-\lambda I\right)^2+I}\ d\lambda\\ &=& \int_{-S}\frac{I}{\left((\|\operatorname{Im}M_-(t)+I\|)^{-1}\operatorname{Re}M_-(t)-\lambda I\right)^2+I}\ d\lambda \nonumber \\&= &\int_S\frac{I}{\left(-(\|\operatorname{Im}M_-(t)+I\|)^{-1}\operatorname{Re}M_-(t)-\lambda I\right)^2+I}\ d\lambda\ . \nonumber
        \end{eqnarray*}

        By taking $S=(a,b)$, with $-C<a<b<C$, the previous identity can be rewritten as
        $$\int_a^b \Big(h_t(\lambda)-l_t(\lambda)\Big)\ d\lambda\ ,$$
        where for each $t\in \Sigma_{ac,l}\setminus N_R$,
        $$h_t(\lambda):=\left(\left(\frac{1}{\|\operatorname{Im}M_+(t)+I\|}\operatorname{Re}M_+(t)-\lambda I\right)^2+I\right)^{-1}\ ,$$
        $$l_t(\lambda):=\left(\left(\frac{-1}{\|\operatorname{Im}M_-(t)+I\|}\operatorname{Re}M_-(t)-\lambda I\right)^2+I\right)^{-1}\ .$$ 

        It follows from the same argument presented before that for each $\lambda\in(-C,C)$ and each $t\in \Sigma_{ac,l}\setminus N_R$, 
        $$\left(\frac{1}{\|\operatorname{Im}M_+(t)+I\|}\operatorname{Re}M_+(t)-\lambda I\right)^2=\left(\frac{-1}{\|\operatorname{Im}M_-(t)+I\|}\operatorname{Re}M_-(t)-\lambda I\right)^2\ .$$

        By expanding both sides of the identity and by simplifying, one gets    
        \begin{align}\label{AB}
            &\left(\frac{\operatorname{Re}M_+(t)}{\|\operatorname{Im}M_+(t)+I\|}\right)^2-\left(\frac{\operatorname{Re}M_-(t)}{\|\operatorname{Im}M_-(t)+I\|}\right)^2 \nonumber \\ &-2\lambda\left(\left(\frac{\operatorname{Re}M_+(t)}{\|\operatorname{Im}M_+(t)+I\|}\right)+\left(\frac{\operatorname{Re}M_-(t)}{\|\operatorname{Im}M_-(t)+I\|}\right)\right)=0\ . 
        \end{align}

        By combining~\eqref{26} with~\eqref{AB}, one concludes that for each $t\in \Sigma_{ac,l}\setminus \{N_I\cup N_R\}$,~\eqref{AB} is valid for each $\lambda\in(-C,C)$ if, and only if, 
        $\operatorname{Re}M_+(t)=-\operatorname{Re}M_-(t)$.  

        Therefore, one concludes that for each $t\in\Sigma_{ac,l} \setminus \{N_I\cup N_R\}$ (with $|N_I\cup N_R|=0$), 
        $$M_+(t)=-\overline{M_-(t)}\,,$$ 
        which proves that $W\in\mathcal{R}(\Sigma_{ac,l})$. 
        
    \end{proof}


    \section[The Oracle Theorem]{The Oracle Theorem for matrix-valued Jacobi operators}\label{104}

    This section provides the main result of this work, 
    consolidating the analytical and geometric results established thus far into a definitive statement on the deterministic nature of matrix-valued Jacobi operators $J$ such that $J_+$ has nontrivial absolutely continuous spectrum of multiplicity $l$. More specifically, by employing the reflectionless property as a structural constraint, we show that the presence of an absolutely continuous component of multiplicity $l$ in the spectral measure is not merely a local property, but a global signature that dictates the asymptotic evolution of the potential. 

    \vspace{0.3cm}
    \subsection{Properties of reflectionless potentials}\label{115}

    This subsection presents some properties of reflectionless potentials. 

    \begin{lema}\label{61}
        Let $W\in\mathcal{V}^C$, let $A\in\mathcal{B}(\mathbb{R})$,  and set $M_\pm(\cdot)=M_\pm^W(0,\cdot)$. Then, $W\in\mathcal{R}(A)$ if, and only if, for each $B\in\mathcal{B}(A)$ so that $|B|<\infty$, and for each bounded $S\in\mathcal{B}(\mathbb{R})$, 
                \begin{equation}\label{31}
            \int_B \omega^\alpha_{M_-(t)}(-S)\ dt = \int_B \omega^\alpha_{M_+(t)}(S)\ dt. 
        \end{equation}
   
    \end{lema}

    \begin{proof}
      It follows from the second part of the proof of Theorem \ref{27} that if~\eqref{31} is valid for each $W\in\mathcal{V}^C$, each $A,B\in\mathcal{B}(\mathbb{R})$ such that $B\subset A$, $|A|<\infty$, and  each bounded $S\in\mathcal{B}(\mathbb{R})$, then $W\in\mathcal{R}(A)$.

      Now, if $W\in\mathcal{R}(A)$, then for a.e. $t\in A$, $M_+(t)=-\overline{M_-(t)}$, from which follows that for each $S\in\mathcal{B}(\mathbb{R})$,\\[-0.5em]      
      \[\omega^\alpha_{M_+(t)}(S)=\omega^\alpha_{-\overline{M_-(t)}}(S)=\omega^\alpha_{M_-(t)}(-S)\]
          
      \noindent (see (\ref{122}) and (\ref{123})), with $\Vert\omega^\alpha_{M_+(t)}(S)\Vert,\Vert\omega^\alpha_{-\overline{M_-(t)}}(S)\Vert\le 1$ (namely, just replace $F(z)$ with $M_+(t)$ or with $-\overline{M_-(-t)}$ in the definition of $\omega^\alpha_{F(z)}(S)$, and note that $\omega^\alpha_{F(z)}(S)$ is well defined even if $F(z)\notin\mathfrak{S}_l$, although $\operatorname{Re}(F(z))$ must be symmetric and $\operatorname{Im}(F(z))\ge 0$). The result follows by integrating this identity over $B\in\mathcal{B}(A)$.
    \end{proof}

    \begin{pro}\label{33}
        Let $A\in\mathcal{B}(\mathbb{R})$ so that $|A|>0$, and let $W\in\mathcal{R}(A)$. Then:
        \begin{enumerate}[label=(\alph*)]
            \item for each $n\in\mathbb{Z}$ and for a.e. $t\in A$, $M_+^W(n,t)=-\overline{M_-^W(n,t)}$ (that is, $S^nW\in\mathcal{R}(A)$ for each $n\in\mathbb{Z}$);
            \item $A\subset\Sigma_{ac,l}(W_\pm)$;
            \item $W_\pm$ uniquely determine $W$, that is, the restriction maps $$\Phi_\pm:\mathcal{R}(A)\to\mathcal{R}_\pm(A)\ ,\ \ \ \ W\mapsto W_\pm$$ are injective;
            \item for each $C>0$, $\mathcal{R}^C(A)=\mathcal{R}(A)\cap\mathcal{V}^C$ is compact;
            \item set $\mathcal{R}_\pm^C(A):=\{W_\pm\mid W\in\mathcal{R}^C(A)\}$; then, the map 
            $$\mathcal{R}_-^C(A) \to \mathcal{R}_+^C(A)\ ,\ \ \ \ \ W_-\to W_+$$
            \noindent is uniformly continuous.
        \end{enumerate}
    \end{pro}

    \begin{proof}
        \begin{enumerate}[label=(\alph*)]
            \item If $W\in\mathcal{R}(A)$, then for a.e. $t\in A$, $M^W_+(0,t)=-\overline{M^W_-(0,t)}$. It follows from Lemmas~\ref{90} and~\ref{85A} that for each $n\in\mathbb{N}$ and for a.e. $t\in A$, 
            $$M_+^W(n,t)=P_+(n,t)M_+^W(0,t)=P_-(n,t)\left(-\overline{M_-^W(0,t)}\right)=-\overline{M_-^W(n,t)}$$
(note that the results presented in Lemma~\ref{90} are still valid for a.e. $t\in A$).

            \item 
              Since $M_\pm^W(\cdot)\in\mathcal{H}$, it follows  that $M_+^W(\cdot)+M_-^W(\cdot)\in\mathcal{H}$ and  that 
              $-\left(M_+^W(\cdot)+M_-^W(\cdot)\right)^{-1}\in\mathcal{H}$. 
              Since the normal limit of a matrix-valued Herglotz function exists Lebesgue a.e.,  it follows that $\left(M_+^W(t)+M^W_-(t)\right)$ and that 
              $\left(M_+^W(t)+M_-^W(t)\right)^{-1}$ exist for a.e. $t\in\mathbb{R}$; consequently, $\operatorname{Ker}\left(M_+^W(t)+M_-^W(t)\right)=\{0\}$ for a.e. $t\in\mathbb{R}$.

              Now, since $W\in\mathcal{R}(A)$, one has $M_+^W(t)+\overline{M_-^W(t)}=0$ for a.e.~$t\in A$, that is, $\operatorname{Re}(M_+^W(t))=-\operatorname{Re}(M_-^W(t))$ and $\operatorname{Im}(M_+^W(t))=\operatorname{Im}(M_-^W(t))$ for a.e. $t\in A$.              Hence, $\left(M_+^W(t)+M_-^W(t)\right)= 2\operatorname{Im}M_\pm^W(t)$, and so \[\operatorname{Ker}\operatorname{Im}M^W_\pm(t)=\operatorname{Ker}\left(M_+^W(t)+M_-^W(t)\right)=\{0\}\] for a.e. $t\in A$. 
                          This shows that 
              $\operatorname{Im}M_\pm^W (t)>0$ for a.e. $t\in A$, that is, $A\subset\Sigma_{ac,l}(W_\pm)$.

            \item It follows from Lemma \ref{32} that $W_-$ determines $M_-^W(\cdot)$, 
              and if $W\in\mathcal{R}(A)$, then $M_-^W(t)$ determines $M_+^W(t)$ for a.e. $t\in A$.  Since $M_+^W(\cdot)\in\mathcal{H}$ and $|A|>0$, it follows that $M_+^W(z)$ is completely determined for each $z\in\mathbb{C}_+$, and then one can obtain $W_+$ from $M_+^W(\cdot)$ by another application of Lemma \ref{32}.

            \item Since $\mathcal{V}^C$ is compact, it suffices to show that $\mathcal{R}^C(A)$ is closed. Let $(W_j)\subset \mathcal{R}^C(A)$ and let $W\in\mathcal{V}^C$ be such that $\lim_{j\to\infty}\rho_\infty(W_j,W)=0$. It follows from Proposition~\ref{CVD} and Lemma~\ref{32} that for each $A\in\mathcal{B}(\mathbb{R})$ with $|A|<\infty$, and each bounded $S\in\mathcal{B}(\mathbb{R})$, 
              \begin{eqnarray*}\int_A \omega^\alpha_{M_-(t)}(-S)\ dt&=&\lim_{j\to\infty}\int_A \omega^\alpha_{M_-^j(t)}(-S)\ dt=\lim_{j\to\infty}\int_A \omega^\alpha_{M_+^j(t)}(S)\ dt\\
                &=&\int_A \omega^\alpha_{M_+(t)}(S)\ dt\, ,\end{eqnarray*}
where $M_\pm^j(\cdot)$ stand for the $m$-functions of $J_\pm^{W_j}$.
              The result is now a direct consequence of Lemma~\ref{61}.

            \item The restrictions of the maps $\Phi_\pm: W\mapsto W_\pm$ to $\mathcal{R}^C(A)$, denoted by $\Phi_\pm^C$ and defined between the metric spaces $\mathcal{R}^C(A)$ and $\mathcal{R}^C_\pm(A)$, are injective (by item (c)) and continuous. Therefore, the inverse maps $$(\Phi_\pm^C)^{-1}:\mathcal{R}_\pm^C(A)\to\mathcal{R}^C(A)\ ,\ \ \ \ \ \ (\Phi_\pm^C)^{-1}(W_\pm)=W,$$ are continuous, along with $$\Phi_+^C\circ(\Phi_-^C)^{-1}:\mathcal{R}_-^C(A)\to\mathcal{R}^C_+(A)\ ,\ \ \ \ \ \ \left(\Phi_+^C\circ(\Phi_-^C)^{-1}\right)(W_-)=W_+.$$ 

              Since $\mathcal{R}^C(A)$ is compact and since $(\Phi^C_\pm)^{-1}$ are continuous maps, if follows from item (d) that both $\mathcal{R}_-^C(A)$ and $\mathcal{R}_+^C(A)$ are compact, and so, $\Phi_+^C\circ(\Phi_-^C)^{-1}$ is uniformly continuous.
        \end{enumerate}
    \end{proof}

    \subsection{The Oracle Theorem}

    This subsection provides the main result of this paper, that is, a version of the Oracle Theorem for matrix-valued Jacobi operators. Most of the proof follows Remling's original proof, but for the sake of the reader, we present the details.

    \begin{teo}{\textnormal{(The Oracle Theorem)}}\label{98}
        Let $C>0,\ \epsilon>0$, and let $A\subset\mathbb{R}$ be a Borel set of positive Lebesgue measure. Then, there exist $L\in\mathbb{N}$ and a smooth function 
        $$\Delta:\left(\mathcal{P}_l(\mathbb{R},(C+1)^{-1},C+1)\times\mathcal{S}_l(\mathbb{R},C)\right)^{L+1}\rightarrow \mathcal{P}_l(\mathbb{R},(C+1)^{-1},C+1)\times\mathcal{S}_l(\mathbb{R},C),$$
        the so-called oracle, such that the following holds: for each $U\in\mathcal{V}^C_+$ so that $\Sigma_{ac,l}(U)\supset A$, there exists $n_0\in\mathbb{N}$ such that for each $n\geqslant n_0$,
        $$\rho\big(U(n+1),\ \Delta\left(U(n-L),...,U(n)\right)\big)<\epsilon\ .$$
    \end{teo}
    \begin{proof}
      Let $A\subset\mathbb{R}$, $C>0$ and $\epsilon>0$ be as in the hypotheses of the theorem. Then, there exists $\delta>0$ so that 
                  \begin{equation}\label{35}
            \rho\left(W(1),\widetilde{W}(1)\right)<\frac{\epsilon}{3} \ \ \ \ \text{if}\ \ \ \ W,\widetilde{W} \in \mathcal{R}^C(A),\ \ \rho_\infty(W_-,\widetilde{W}_-)<5\delta.
        \end{equation}

                  This is a consequence of Proposition \ref{33}-(e), since the map $$\mathcal{R}^C_-(A) \to \mathcal{R}_+^C(A), \ \ \ \ \ W_-\mapsto W_+,$$  is a uniformly continuous bijection with a uniformly continuous inverse. We also let $\delta<\epsilon/3$ for technical reasons that are discussed later.

        Now, define the $2\delta$-neighborhood of $\mathcal{R}_-^C(A)$ as the set $$\mathcal{U}_{2\delta}:=\{U_- \in \mathcal{V}_-^C : \rho_\infty(U_-,\mathcal{R}_-^C(A))\leqslant 2\delta\},$$ 
      where    $\rho_\infty(U_-,\mathcal{R}_-^C(A)):=\inf\{\rho_{\infty}(U_-,W)\mid W\in\mathcal{R}_-^C(A)\}$; this set is obviously closed, and since it is compact (by Proposition \ref{33}-(d)), one can cover it with a finite number of balls of size $3\delta$ (say, $M$), whose centers are elements $W_-^{(j)} \in \mathcal{R}_-^C(A)$: 
        $$\mathcal{U}_{2\delta} \subset B_{3\delta}(W_-^{(1)})\cup...\cup B_{3\delta}(W_-^{(M)}).$$

        We may choose $L$ large enough so that $(4C+2)\Sigma_{j>L} 2^{-j}<\delta$. This choice of $L$ guarantees that $\rho_{\infty}(U_-,\widetilde{U}_-)<\delta$ whenever $U(n)=\widetilde{U}(n)$ for $n=0,-1,\ldots,-L$ (and $U_-, \widetilde{U}_- \in \mathcal{V}_-^C)$. Namely, since for each $n<-L$, $\rho\left(U(n),\widetilde{U}(n)\right)<(4C+2)$ (by definition), one gets 
         \begin{equation*}
            \rho_{\infty}(U_-,\widetilde{U}_-)=\sum_{n<-L}2^{-|n|}\ \rho\left(U(n),\widetilde{U}(n)\right)<\delta \, .
        \end{equation*}

         We now present a preliminary definition of the oracle function $\Delta$ as follows: if for each $j\in\{0,\ldots,L\}$, 
         $$A_{-j}\in \mathcal{P}_l(\mathbb{R},(C+1)^{-1},C+1)\times \mathcal{S}_l(\mathbb{R},C)$$
         and 
         $$[\ldots, 0, 0, 0,\ldots, 0, A_{-L},\ldots, A_0]\in B_{3\delta}(W_-^{(1)}),$$ then set $$\Delta(A_{-L},..., A_0):=W^{(1)}(1).$$

        Now, one can focus on the next open ball: if $$[\ldots, 0, 0, 0,\ldots, 0, A_{-L},\ldots, A_0]\in B_{3\delta}(W_-^{(2)})\setminus B_{3\delta}(W_-^{(1)}),$$ then set $$\Delta(A_{-L},\ldots, A_0):=W^{(2)}(1).$$

         Proceed in this way up to $B_{3\delta}(W_-^{(M)})\setminus B_{3\delta}(W_-^{(M-1)})$. If $(A_{-L},\ldots,A_0)\notin B_{3\delta}(W_-^{(M)})$, then set $\Delta (A_{-L},..,A_0):= 0$.

        Let $U\in\mathcal{V}^C_+$ be so that $A\subset\Sigma_{ac,l}(U)$. It remains to show that $\Delta$ predicts $U(n+1)$. In order to do this, let $n_0 \geqslant L$ be sufficiently large so that, for each $n\ge n_0$, $\rho_\infty(S^nU,\omega(U))<\delta$. This is possible thanks to Proposition \ref{34}. Hence, it follows from Theorem \ref{27} that for each $n\geqslant n_0$, there exists $\widetilde{W}\in \mathcal{R}^C(A)$ \footnote{Since $A \subset \Sigma_{ac,l}(U)$, if $\widetilde{W}$ is reflectionless on $\Sigma_{ac,l}(U)$, then $\widetilde{W}$ is evidently reflectionless on $A$; it follows that $\mathcal{R}(\Sigma_{ac,l}(U)) \subset \mathcal{R}(A)$.} such that
        \begin{equation}\label{36}
            \rho_\infty(S^nU,\widetilde{W})<\delta.
        \end{equation}

        By the choice of $L$, one has $$Z=[\ldots, 0, 0, 0,\ldots, 0, (S^nU)(-L),\ldots, (S^nU)(0)]\in \mathcal{U}_{2\delta},$$ from which follows that there exists $j\in \{1,\ldots,M\}$ such that $$[\ldots, 0, 0, 0,\ldots, 0, U(n-L),\ldots, U(n)]\in B_{3\delta}(W_-^{(j)}).$$ Fix the minimum value of $j$ for which this property holds. For this choice of $j$, one gets 
                \begin{equation}\label{37}
            \Delta (U(n-L), U(n-L+1),\ldots, U(n))=W^{(j)}(1),
        \end{equation}
        \noindent by the definition of $\Delta$. One also has that 
        \begin{eqnarray*}\rho_\infty(W^{(j)}_-, \widetilde{W}_-)&\leqslant& \rho_\infty(W^{(j)}_- , [\ldots, 0, U(n-L),\ldots, U(n)])\\&+& \rho_\infty([\ldots, 0, U(n-L),\ldots, U(n)], (S^nU)_-)\\ &+& \rho_\infty((S^nU)_-, \widetilde{W}_-) < 3\delta + \delta + \delta = 5\delta,\end{eqnarray*}  and so, it follows from relation~\ref{35} that 
        \begin{equation}\label{38}
            \rho\left(W^{(j)}(1),\widetilde{W}(1)\right)<\frac{\epsilon}{3}.
        \end{equation}
       
        On the other hand, it follows from inequality \ref{36} that $$\rho\left(U(n+1),\widetilde{W}(1)\right) < 2\delta <\frac{2\epsilon}{3},$$ and if one combines this with \ref{37} and \ref{38}, one gets $$\rho\big(U(n+1),\ \Delta (U(n-L), U(n-L+1),\ldots, U(n))\big)<\epsilon\ ,$$ as desired.
        
        \vspace{0.3cm}
        The function $\Delta$ constructed above is a step function and, therefore, it is not continuous. In order to obtain a smooth oracle, one can employ a smooth partition of unity. Let $\{O_j\}_{j=1}^M$ be the open cover of the compact set $\mathcal{U}_{2\delta}$ given by $O_j := B_{3\delta}(W_-^{(j)})$, and let $\{\phi_j\}_{j=1}^M$ be a smooth partition of unity subordinate to this cover. One can then redefine $\Delta$ as the convex combination $$\Delta(X) := \sum_{j=1}^M \phi_j(X) W^{(j)}(1)\ .$$ Since the space $\mathcal{P}_l(\mathbb{R},(C+1)^{-1},C+1)\times\mathcal{S}_l(\mathbb{R},C)$ is convex (see Proposition \ref{convexity_pot}), this smooth function $\Delta$ is well-defined and maps strictly into the correct space. Furthermore, for a given $X \in \mathcal{U}_{2\delta}$, if $\phi_j(X) > 0$, then $X \in O_j$, meaning the bound~\eqref{38} holds for the corresponding $W^{(j)}(1)$. By the triangle inequality (and by the convexity of the metric $\rho$), it follows that the error of this smooth prediction is strictly bounded by convex combinations of errors smaller than $\epsilon/3$, ensuring that the overall accuracy $\rho\big(U(n+1),\ \Delta(\dots)\big) < \epsilon$ is preserved. This concludes the proof of the Oracle Theorem.
    \end{proof}

    \vspace{0.3cm}
    \subsection{Some consequences of Theorems~\ref{27} and~\ref{98}}\label{CTO}

    Here, we discuss the counterparts of Theorem~1.1, Corollaries~1.6 and~1.7 in~\cite{remling}  for matrix-valued Jacobi operators.

    \begin{teo}\label{110a}
      Suppose that the (half-line) potential $U$ 
      takes only finitely many values and that $\Sigma_{ac,l}(U)\neq\emptyset$. Then, $U$ is eventually periodic: there exist $n_0, p \in \mathbb{N}$ so that for each $n\ge n_0$, $U(n+p)=U(n)$.
    \end{teo}
    \begin{proof} The proof of this result is identical to the proof of Theorem~1.1 in~\cite{remling}, which we present for completeness.

      By choosing $\epsilon>0$ small enough, one can use an oracle to
      (eventually) predict $U(n)=(D_n,V_n)$ exactly, given the previous $L + 1$ values of $U$. But there are only finitely many different blocks of size $L + 1$, so after a while, things must start repeating themselves.
    \end{proof}
    

\begin{cor}\label{cor1.5} Let $C>0$, $U\in\mathcal{V}^C$, and suppose that $W$ is a pertubation such that $U+W\in\mathcal{V}^C$, and for which there exists a subsequence $n_j\rightarrow\infty$ so that $$\limsup_{j\to\infty}\Vert W (n_j)\Vert > 0\,$$
(recall that $\Vert U(n)\Vert=\rho(U(n),0)=\rho(D_n,0)+\rho(V_n,0)=\Vert D_n\Vert+\Vert V_n\Vert$), but for each $k\in\mathbb{N}$,
$$\lim_{j\to\infty}\Vert U (n_j - k)\Vert = 0\,.$$
 Then, $\Sigma_{ac,l}(U)\cap\Sigma_{ac,l}(U+W) = \emptyset$. In particular, this conclusion holds for every bounded perturbation $W$ of the form
$$W (n) =\sum_{j\ge1}(D_j,V_j)\delta_{n,n_j},\qquad n_j-n_{j-1}\rightarrow\infty,\qquad \limsup_{j\to\infty}\Vert (D_j,V_j)\Vert > 0\,,$$ such that $U+W\in\mathcal{V}^C$.
\end{cor}

\begin{proof} Since, by hypothesis, there exists $C>0$ so that $U,$ $U+W\in\mathcal{V}^C$, then $A =\Sigma_{ac,l}(U)\cap\Sigma_{ac,l}(U+W)$ cannot have positive measure, because then Oracle Theorem would provide oracles that work for both $U$ and $U+W$. This is impossible, since $\Vert W (n_j )\Vert\ge\epsilon> 0$ on a suitable subsequence, but $W$ is small on long intervals
  to the left of these points, so no (continuous) oracle with sufficiently high accuracy can predict both $U$ and $U+W$ correctly.
\end{proof}

\begin{obs}
  \begin{enumerate}
  \item It is not necessary to use Oracle Theorem here. Corollary~\ref{cor1.5} also follows directly from Theorem~\ref{27} if one uses the standard uniqueness property of reflectionless potentials presented in Proposition~\ref{33}-(c).
  \item In Corollary~\ref{cor1.5}, one demands that there exists $C>0$ so that $U,$ $U+W\in\mathcal{V}^C$, a much stricter condition than the original result proved by Remling (Corollary~1.5 in [\ref{remling}]), which follows for any (half-line) potential $U$. This is due to the fact that, by Simon-Spencer Theorem for scalar Jacobi operators, if either $W\notin\cup_{C>0}\mathcal{V}^C$ or $U+W\notin\cup_{C>0}\mathcal{V}^C$, then the corresponding operator has empty absolutely continuous spectrum. Since we do not have a version of Simon-Spencer Theorem for matrix-valued Jacobi operators, one has to demand, in Corollary~\ref{cor1.5}, that  $U,$ $U+W\in\mathcal{V}^C$ for some $C>0$.
    \end{enumerate}
\end{obs}

As another immediate consequence of Theorem~\ref{27}, one partially recovers a result on the semicontinuity of $\Sigma_{ac,l}$ (which extends Last-Simon Theorem for matrix-valued Jacobi matrices).

\begin{cor}[Proposition~5.4 in \ref{fabricioSilasA}] \label{111a}
If $W \in \omega(V)$, then $\Sigma_{ac,l}(W_\pm) \supset \Sigma_{ac,l}(V)$.
\end{cor}
\begin{proof}
  It follows from Theorem~\ref{27} that $W\in\mathcal{R}(\Sigma_{ac,l}(V))$. The result now follows from Proposition~\ref{33}-(b). 
  \end{proof}


\appendix

\section{Continuous dependence of the measures \texorpdfstring{$\omega^{I,R}$}{omega I,R}}\label{92}

    In this appendix we present, in details, the proofs of Lemmas~\ref{21.} and~\ref{28.}. For the sake of the readers, we present the statements again. 

    Recall that for each $z\in\mathbb{C}_+$, each $F\in\mathcal{H}$ and each $S\in\mathcal{B}(\mathbb{R})$, one defines 
\[\omega_{F(z)}^{I}(S)=\frac{1}{\pi}\int_S\frac{(\operatorname{Im}F(z))^{1/2}+I}{((\operatorname{Im}F(z))^{1/2}+I)^2+\lambda^2I}\ d\lambda\]
    and
\[\omega_{F(z)}^{R}(S)=\frac{1}{\pi}\int_S\frac{I}{\left({\|\operatorname{Im}F(z)+I\|}^{-1}\operatorname{Re}F(z)-\lambda I\right)^2+I}\ d\lambda\, .\]

      %
\begin{proof1}
  \normalfont
      Write $F(z)= Z_1 =X_1+iY_1$  and  $H(z)=Z_2 =X_2+iY_2$. Firstly, one needs some auxiliary results. 
      It follows from Theorem~A.1 in~\cite{Froese} that for each $Z_1,Z_2\in\mathfrak{S}_l$, one has  
        \begin{equation}\label{15}
            d_{\infty}(Z_1,Z_2)=2\ln \left\|S_{Z_1}^{-1}S_{Z_2}\right\|, \ \ \ \text{where} \ S_Z:=\left(\begin{array}{cc}
			Y^{1/2} & \ XY^{-1/2} \\
			0 & Y^{-1/2} \\
	   \end{array}\right), \ \ Z=X+iY \in \mathfrak{S}_l.
        \end{equation}

        Given that $S_{Z_1}^{-1}S_{Z_2}$ is a real symplectic matrix (see~\cite{Froese}), there exist $\mathcal{O}, \mathcal{O}' \in Sp(2l,\mathbb{R})\cap O(2l,\mathbb{R})$ and $\mathcal{D}=\left(\begin{array}{cc}
			D & \ 0 \\
			0 & D^{-1} \\
	   \end{array}\right)$ so that $S_{Z_1}^{-1}S_{Z_2}=\mathcal{O}\mathcal{D}\mathcal{O}'$, where $D$ is a diagonal matrix formed by the singular values of $S_{Z_1}^{-1}S_{Z_2}$. 
       
       One has from the identity in~(\ref{15}) that if $d_{\infty}(Z_1,Z_2)<\epsilon$, then 
       \[1\leqslant\left\|S_{Z_1}^{-1}S_{Z_2}\right\|=\left\|\mathcal{D}\right\|\leqslant e^{\epsilon/2}< 1+\epsilon\]       
       \noindent (which is true for $0<\epsilon<1/2$). Therefore,
       \[1\leqslant\left\|\mathcal{D}^{-1}\right\|< 1+\epsilon\ \ \ \  \text{or}\ \ \ \  s_l(\mathcal{D}^{-1})>\frac{1}{1+\epsilon}\ .\]
       
 In particular, one gets 
       \[\left\|S_{Z_1}^{-1}S_{Z_2}-\mathcal{O}\mathcal{O}'\right\|=\left\|\mathcal{O}(\mathcal{D}-I)\mathcal{O}'\right\|=\left\|\mathcal{D}-I\right\|=\left\|\mathcal{D}\right\|-1<\epsilon\ ,\]       
       \noindent since $\left\|\mathcal{D}-I\right\|
       =\left\|\mathcal{D}\right\|-1$.

       Note that since $\mathcal{O}, \mathcal{O}' \in Sp(2l,\mathbb{R})\cap O(2l,\mathbb{R})$, then 
       $\mathcal{O}\mathcal{O}' \in Sp(2l,\mathbb{R})\cap O(2l,\mathbb{R})$, 
       which implies that there exists $V\in U(l)$ such that 
       \[\mathcal{O}\mathcal{O}'=\left(\begin{array}{cc}
			\operatorname{Re}V & \ -\operatorname{Im}V \\
			\operatorname{Im}V & \ \operatorname{Re}V \\
	   \end{array}\right)\,; \]
hence, 
       \[S_{Z_1}^{-1}S_{Z_2}-\mathcal{O}\mathcal{O}'=\left(\begin{array}{cc}
			Y_1^{-1/2}Y_2^{1/2}-\operatorname{Re}V & \ \ \ Y_1^{-1/2}(X_2-X_1)Y_2^{-1/2}+\operatorname{Im}V \\
			-\operatorname{Im}V & Y_1^{1/2}Y_2^{-1/2}-\operatorname{Re}V \\
	   \end{array}\right)\ .\] 
       %
      
      Since $\left\|S_{Z_1}^{-1}S_{Z_2}-\mathcal{O}\mathcal{O}'\right\|<\epsilon$, then 
              \begin{align}\label{68}
                &\left\|Y_1^{-1/2}Y_2^{1/2}-\operatorname{Re}V\right\|<\epsilon\ , \nonumber \\ &\left\|\operatorname{Im}V\right\| < \epsilon\ , \nonumber \\ &\left\|Y_1^{-1/2}(X_2-X_1)Y_2^{-1/2}+\operatorname{Im}V\right\|<\epsilon\ , \nonumber \\ &\left\|Y_1^{1/2}Y_2^{-1/2}-\operatorname{Re}V\right\|<\epsilon\ ; 
       \end{align}
in particular,
       $$\left\|Y_1^{-1/2}(X_2-X_1)Y_2^{-1/2}\right\| \leqslant \left\|Y_1^{-1/2}(X_2-X_1)Y_2^{-1/2}+\operatorname{Im}V\right\| + \left\|-\operatorname{Im}V\right\| < 2\epsilon\ .$$

Now, one needs to show that if $\left\|\operatorname{Im}V\right\| < \epsilon$, then $\left\|\operatorname{Re}V-I\right\| < \epsilon$. Note firstly that since $V\in U(l)$, there exist another $W\in U(l)$ and a diagonal matrix $\Lambda=\operatorname{diag}(\lambda_1,\ldots,\lambda_l)$ so that $V=W\Lambda W^{-1}$, with $\lambda_j=e^{i\theta_j}$ for some $\theta_j\in[0,2\pi)$, $j\in\{1,\ldots,l\}$.

  Then, $\operatorname{Im}(V)=W\operatorname{Im}(\Lambda)W^{-1}$ and $\operatorname{Re}(V)=W\operatorname{Re}(\Lambda)W^{-1}$, with \[\operatorname{Im}(\Lambda)=\operatorname{diag}(\sin\theta_1,\ldots,\sin\theta_l)\ \ \ \textrm{and} \ \ \ \operatorname{Re}(\Lambda)=\operatorname{diag}(\cos\theta_1,\ldots,\cos\theta_l). \]

  It follows from relation $\left\|\operatorname{Im}V\right\| < \epsilon$  that $\max_j|\sin\theta_j|<\epsilon$, and so for each $j\in\{1,\ldots,l\}$,  $|\sin\theta_j|<\epsilon$. Now, since $\left\|\operatorname{Re}V-I\right\|=\max_j(1-\cos\theta_j)$, it follows from the identity $\cos\theta_j=\sqrt{1-\sin^2\theta_j}${\footnote{Suppose that there exists $j\in\{1,\ldots,l\}$ so that $\cos\theta_j=-\sqrt{1-\sin^2\theta_j}$; then, for each $0<\epsilon<1/2$, $\cos\theta_j<-1+\epsilon^2/\sqrt{2}$. This result, combined with $\|Y_1^{1/2}Y_2^{-1/2}-\operatorname{Re}V\|<\epsilon$, implies that $Y_1^{1/2}Y_2^{-1/2}$ has at least one negative eigenvalue, which is absurd. So, for each $0<\epsilon<1/2$ and each $j\in\{1,\ldots,l\}$, $\cos\theta_j=\sqrt{1-\sin^2\theta_j}$.}}, valid for each $j\in\{1,\ldots,l\}$, that 
  \[\left\|\operatorname{Re}V-I\right\|=\max_j(1-\cos\theta_j)<1-\sqrt{1-\epsilon^2}\le \epsilon^2<\epsilon\]
  (where we have used the inequality $1-\sqrt{1-x}\le x$, valid for each $x\in[0,1]$). 

  Therefore, by combining this result with (\ref{68}), one gets       
       \begin{equation}\label{69}
           \left\|Y_1^{-1/2}Y_2^{1/2}-I\right\|<2\epsilon,\ \ \ \left\|Y_1^{1/2}Y_2^{-1/2}-I\right\|<2\epsilon,\ \ \ \left\|Y_1^{-1/2}(X_2-X_1)Y_2^{-1/2}\right\|<2\epsilon.
       \end{equation}\\[0em]
        
      \textit{Proof of item 1.} Let $S\in\mathcal{B}(\mathbb{R})$ be such that $|S|<\infty$. It follows from Hadamard Theorem for $C^1$ matrix-valued functions that{\footnote{The theorem states that if $f:Dom(f)\subset \mathcal{S}(l,\mathbb{R})\rightarrow\mathcal{S}(l,\mathbb{R})$ is a $C^1$ matrix-valued function and if $A,B\in Dom(f)$ are such that for each $t\in[0,1]$, $tA+(1-t)B\in Dom(f')$, then  $f(A)-f(B)=\int_0^1f'\big(tA+(1-t)B\big)(A-B)\ dt$. See [\ref{Bhatia}] or [\ref{Dieudonné}] for details.}}
        \begin{align}\label{EAIMPOR}
          \left\|\omega_{F(z)}^{I}(S)-\omega_{H(z)}^{I}(S)\right\|&\leqslant\frac{1}{\pi}\int_S\left\|f_{\lambda}\big((\operatorname{Im}F(z))^{1/2}+I\big)-f_{\lambda}\big((\operatorname{Im}H(z))^{1/2}+I\big)\right\|\ d\lambda \nonumber \\ &\leqslant\frac{1}{\pi}\int_S\int_0^1\left\|f_{\lambda}'\big(tA+(1-t)B\big)(A-B)\right\|\ dt\ d\lambda\,, 
      \end{align} where      $$A=[(\operatorname{Im}F(z))^{1/2}+I], \qquad B=[(\operatorname{Im}H(z))^{1/2}+I]\ $$
 are positive definite matrices such that $\Vert A\Vert,\Vert B\Vert>1$,     and for each $\lambda\in\mathbb{R}$,
 \[f_\lambda:(1,\infty)\rightarrow (0,1), \qquad f_{\lambda}(x):=\frac{x}{x^2+\lambda^2}.\]

      Therefore, one needs to estimate $\left\|f_{\lambda}'(tA+(1-t)B)(A-B)\right\|$. Firstly, note that if $X$ is a strictly positive definite matrix (and in particular, self-adjoint), then by the functional calculus of self-adjoint operators, one has       
       $$f_{\lambda}'(X)=X^{-2}(\lambda^2X^{-2}-I)(\lambda^2X^{-2}+I)^{-2}$$ 
       (note that the matrix $(\lambda^2X^{-2}+I)$ is positive definite, and in particular, it is invertible).

       Now, let $t\in[0,1]$ and set $X:=tA+(1-t)B$, with $A=Y_1^{1/2}+I$ and $B=Y_2^{1/2}+I$. Then, 
            $$X=t(Y_1^{1/2}-Y_2^{1/2})+Y_2^{1/2}+I.$$

       Note that all eigenvalues of $X$ are greater than 1. 
       Then, $X$ is positive definite and
            \begin{equation}\label{16}
          f_{\lambda}'(X)(A-B)=X^{-1}(\lambda^2X^{-2}-I)(\lambda^2X^{-2}+I)^{-2}X^{-1}(A-B).
      \end{equation}

            The next step consists in estimating, for each $\lambda\in\mathbb{R}$, 
            the norm of $X^{-1}(\lambda^2X^{-2}-I)(\lambda^2X^{-2}+I)^{-2}$. Firstly, it follows from the functional calculus that 
      $$\left\|X^{-1}(\lambda^2X^{-2}-I)(\lambda^2X^{-2}+I)^{-2}\right\|=\max_{\psi_i \in \sigma(X)}\left|\frac{(\lambda/\psi_i)^2-1}{\psi_i((\lambda/\psi_i)^2+1)^2}\right| \ ,$$
      with $\psi_i>1$ for each $i=1,\ldots,l$ and each $t\in[0,1]$. 
      Since for each $y\in\mathbb{R}$, $\dfrac{\left|y^2-1\right|}{(y^2+1)^2}<1$, 
      one concludes that

     \begin{equation}\label{70}
          \left\|X^{-1}(\lambda^2X^{-2}-I)(\lambda^2X^{-2}+I)^{-2}\right\|<1.
      \end{equation}\\[-1.5em] 
      
      \noindent For the remaining terms of $f_{\lambda}'(X)(A-B)$, note that:

      \begin{enumerate}
          \item \[A-B=Y_1^{1/2}-Y_2^{1/2}=Y_2^{1/2}(Y_2^{-1/2}Y_1^{1/2}-I)=Y_2^{1/2}\mathcal{F}=(B-I)\mathcal{F},\] 
            with  $A=Y_1^{1/2}+I$, $B=Y_2^{1/2}+I$, $\mathcal{F}=(Y_2^{-1/2}Y_1^{1/2}-I)$; moreover, $\|\mathcal{F}\|=\Vert Y_2^{-1/2}Y_1^{1/2}-I\Vert=\Vert Y_1^{1/2}Y_2^{-1/2}-I\Vert<2\epsilon$, by 
            (\ref{69});

          \item it follows from item 1. that $X=t(A-B)+B=t(B-I)\mathcal{F}+B=B(I+t\mathcal{F})-t\mathcal{F}$;

          \item it follows from items 1. and 2. that 
            \begin{align}
              X^{-1}(A-B) &= (B(I+t\mathcal{F})-t\mathcal{F})^{-1}(B-I)\mathcal{F} \nonumber \\
              &= (B(I+t\mathcal{F})-t\mathcal{F})^{-1}B\mathcal{F}-(B(I+t\mathcal{F})-t\mathcal{F})^{-1}\mathcal{F} \nonumber \\ &= (I+t\mathcal{F}-tB^{-1}\mathcal{F})^{-1}\mathcal{F}-(B(I+t\mathcal{F})-t\mathcal{F})^{-1}\mathcal{F}, \nonumber
          \end{align}
given that $(B(I+t\mathcal{F})-t\mathcal{F})^{-1}B\mathcal{F}=\big(B(I+t\mathcal{F}-tB^{-1}\mathcal{F})\big)^{-1}B\mathcal{F}=(I+t\mathcal{F}-tB^{-1}\mathcal{F})^{-1}\mathcal{F}$ ($B$ is invertible);

\item  one has, for each $t\in[0,1]$, that 
    \begin{align}
              s_l(I+t\mathcal{F}-tB^{-1}\mathcal{F}) &= s_l(I+t(I-B^{-1})\mathcal{F}) \nonumber \\ &\geqslant s_l(I)-\left\|t(I-B^{-1})\mathcal{F}\right\| \nonumber \\ &\geqslant 1-(\left\|\mathcal{F}\right\|+\left\|B^{-1}\right\|\left\|\mathcal{F}\right\|) \nonumber \\ &\geqslant 1-2\epsilon-2\epsilon=1-4\epsilon,\nonumber 
     \end{align} where we have used the facts that $\|B^{-1}\|=(s_l(B))^{-1}<1$ (since each one of the eigenvalues of $B$ is greater than 1) and that $\Vert\mathcal{F}\Vert\le 2\epsilon$;

\item one has, for each $t\in[0,1]$, that 
            \begin{align}
              s_l\big(B(I+t\mathcal{F})-t\mathcal{F}\big) &\geqslant s_l(B(I+t\mathcal{F})) - \left\|\mathcal{F}\right\| \nonumber \\ &\geqslant s_l(B)s_l(I+t\mathcal{F}) - 2\epsilon \nonumber \\ &\geqslant s_l(I+t\mathcal{F}) - 2\epsilon \nonumber \\ &\geqslant 1-\left\|\mathcal{F}\right\| - 2\epsilon=1-4\epsilon;\nonumber
          \end{align}

        \item it follows from items 4. and  5. that 
          \begin{align}\label{71}
              \left\|X^{-1}(A-B)\right\| &\leqslant \left\|(I+t\mathcal{F}-tB^{-1}\mathcal{F})^{-1}\mathcal{F}\right\|+\left\|(B(I+t\mathcal{F})-t\mathcal{F})^{-1}\mathcal{F}\right\| \nonumber \\ &\leqslant\left\|\mathcal{F}\right\| \left(\frac{1}{s_l(I+t\mathcal{F}-tB^{-1}\mathcal{F})}+\frac{1}{ s_l\big(B(I+t\mathcal{F})-t\mathcal{F}\big)}\right) \nonumber \\ &\leqslant \frac{4\epsilon}{1-4\epsilon}\ .
            \end{align} 
\end{enumerate}

          By combining relations (\ref{16}), (\ref{70}) and (\ref{71}), one concludes that for each $t\in[0,1]$, \[\left\|f_{\lambda}'(tA+(1-t)B)(A-B)\right\|\leqslant\frac{4\epsilon}{1-4\epsilon}.\] 

      Finally, by combining the previous estimate with~\eqref{EAIMPOR}, one gets
      \begin{align*}
          \left\|\omega_{F(z)}^{I}(S)-\omega_{H(z)}^{I}(S)\right\|&\leqslant\frac{1}{\pi}\int_S\left\|f_{\lambda}\big((\operatorname{Im}F(z))^{1/2}+I\big)-f_{\lambda}\big((\operatorname{Im}H(z))^{1/2}+I\big)\right\|\ d\lambda \nonumber \\ &\leqslant\frac{1}{\pi}\int_S\int_0^1\left\|f_{\lambda}'\big(tA+(1-t)B\big)(A-B)\right\|\ dt\ d\lambda \nonumber \\ &\leqslant \frac{1}{\pi}\frac{4\epsilon}{1-4\epsilon}|S|.       \end{align*}

      The result now follows by setting $u(\epsilon,|S|):=\dfrac{1}{\pi}\dfrac{4\epsilon}{1-4\epsilon}|S|$. 

\textit{Proof of item 2.} Let $C>0$ and let $S\in\mathcal{B}(\mathbb{R})$ be such that $S\subset(-C,C)$. It follows again from Hadamard Theorem that
\begin{align*}
         \left\|\omega_{F(z)}^{R}(S)-\omega_{H(z)}^{R}(S)\right\|\leqslant
         \frac{1}{\pi}\int_S\int_0^1\left\|g_{\lambda}'\big(tA+(1-t)B\big)(A-B)\right\|\ dt\ d\lambda, 
       \end{align*} 
where       
       $$A=\left[\frac{1}{\|\operatorname{Im}F(z)+I\|}\operatorname{Re}F(z)\right],\ \ \ \ \  B=\left[\frac{1}{\|\operatorname{Im}H(z)+I\|}\operatorname{Re}H(z)\right]\,, $$
       and  for each $\lambda\in\mathbb{R}$,       
       \[g_{\lambda}:\mathbb{R}\rightarrow (0,1],\qquad g_\lambda(x)=\frac{1}{(x-\lambda)^2+1}.\]
       
       Hence, one needs to estimate $\left\|g_{\lambda}'\big(tA+(1-t)B\big)(A-B)\right\|$. Firstly, note that for each $\lambda,x\in\mathbb{R}$,
       $$g_{\lambda}'(x)=\frac{-2(x-\lambda)}{\big((x-\lambda)^2+1\big)^2}\ .$$

       If $X$ is a self-adjoint matrix, then it follows from the functional calculus that
       \[g_{\lambda}'(X)=-2(X-\lambda I)\big((X-\lambda I)^2+I\big)^{-1}\big((X-\lambda I)^2+I\big)^{-1}\]      
      \noindent (note that since the matrix $((X-\lambda I)^2+I)$ is positive definite, it is invertible). Now, for each $t\in[0,1]$, let $X=tA+(1-t)B$, with $$A=\frac{1}{\left\|Y_1+I\right\|}X_1 \ \ \ \text{and} \ \ \ B=\frac{1}{\left\|Y_2+I\right\|}X_2\ .$$ Then, $$X=t\left(\frac{1}{\left\|Y_1+I\right\|}X_1-\frac{1}{\left\|Y_2+I\right\|}X_2\right)+\frac{1}{\left\|Y_2+I\right\|}X_2,\ \ \ \ t\in[0,1].$$

      Note that for each $t\in[0,1]$, $X$ is self-adjoint, given that $X_1=\operatorname{Re}(F(z))$, $X_2=\operatorname{Re}(H(z))$, $Y_1=\operatorname{Im}(F(z))$ and $Y_2=\operatorname{Im}(H(z))$ are symmetric matrices. One can write 
      \begin{align}\label{17}
          A-B &=\frac{1}{\left\|Y_1+I\right\|}(X_1-X_2)+\left(\frac{1}{\left\|Y_1+I\right\|}-\frac{1}{\left\|Y_2+I\right\|}\right)X_2 \nonumber \\ &= \mathcal{C}+\gamma B, 
      \end{align} where $$\mathcal{C}=\frac{1}{\left\|Y_1+I\right\|}(X_1-X_2) \ \ \ \text{and} \ \ \ \gamma=\left(\frac{\left\|Y_2+I\right\|}{\left\|Y_1+I\right\|}-1\right).$$ 

      Then, one has
      \begin{eqnarray}\label{PRINCI}
          \|g_{\lambda}'\big(tA+(1-t)B\big)(A-B)\|
          &\le&\left\|2(X-\lambda I)\big((X-\lambda I)^2+I\big)^{-1}\right\|\nonumber\\
          &\cdot&\left\|\big((X-\lambda I)^2+I\big)^{-1}(\mathcal{C}+\gamma B)\right\|.
\end{eqnarray}
       
 One also has, for each $\lambda\in\mathbb{R}$,         
        \begin{equation}\label{72}
            \left\|2(X-\lambda I)\big((X-\lambda I)^2+I\big)^{-1}\right\|=2\sup_{\psi_i \in \sigma(X)}\left|\frac{\psi_i-\lambda}{(\psi_i-\lambda)^2+1}\right|\leqslant 2.
        \end{equation} 
           
        It remains to estimate $\left\|\big((X-\lambda I)^2+I\big)^{-1}(\mathcal{C}+\gamma B)\right\|$.

        \begin{enumerate}
          \item Firstly, note that:
      \begin{itemize}
          \item if one sets $\mathcal{G}:=Y_1^{-1/2}Y_2^{1/2}-\operatorname{Re}V$,  then  $Y_2^{1/2}=Y_1^{1/2}(\mathcal{G}+\operatorname{Re}V)$; it follows from (\ref{68}) and $\Vert\operatorname{Re}(V)\Vert\le 1$ that $\|Y_2\|^{1/2}\leqslant\|Y_1\|^{1/2}(1+\epsilon)$. 

          \item if one sets $\Gamma=Y_1^{1/2}Y_2^{-1/2}-\operatorname{Re}V$, then $Y_1^{1/2}=(\Gamma+\operatorname{Re}V)Y_2^{1/2}$; it follows from (\ref{68}) that
            $\|Y_1\|^{1/2}\leqslant(1+\epsilon)\|Y_2\|^{1/2}$.
 \end{itemize}

      By combining these two statements, one gets
          \begin{equation}\label{18}
              \frac{\left\|Y_2\right\|}{(1+\epsilon)^2}\leqslant\left\|Y_1\right\|\leqslant\left\|Y_2\right\|(1+\epsilon)^2\ ,
          \end{equation}
         and so, 
          $$\frac{\left\|Y_2+I\right\|}{(1+\epsilon)^2}\leqslant\left\|Y_1+I\right\|\leqslant\left\|Y_2+I\right\|(1+\epsilon)^2\ ,$$
         given that $Y_1$ and $Y_2$ are self-adjoint; this, by its turn, results in 
         
         \begin{equation}\label{74}
            (1+\epsilon)^{-2}-1\leqslant\gamma=\frac{\left\|Y_2+I\right\|}{\left\|Y_1+I\right\|}-1\leqslant(1+\epsilon)^2-1\ .
          \end{equation}
          
       \item  Now, one can write
         \begin{align}
              X_2-X_1 &=Y_1^{1/2}Y_1^{-1/2}X_2Y_2^{-1/2}Y_2^{1/2}-Y_1^{1/2}Y_1^{-1/2}X_1Y_2^{-1/2}Y_2^{1/2} \nonumber \\ &=Y_1^{1/2}\left(Y_1^{-1/2}(X_2-X_1)Y_2^{-1/2}\right)Y_2^{1/2}, \nonumber
          \end{align}
          from which follows that
          \begin{align}
              \left\|X_2-X_1\right\| &\leqslant \left\|Y_1^{1/2}\right\|\left\|Y_1^{-1/2}(X_2-X_1)Y_2^{-1/2}\right\|\left\|Y_2^{1/2}\right\| \nonumber \\ &\leqslant 2\epsilon\left\|Y_1\right\|^{1/2}\left\|Y_2\right\|^{1/2}\ \ \ (\text{by}\ \eqref{69}) \nonumber \\ &\leqslant 2\epsilon(1+\epsilon)\left\|Y_1\right\|\ \ \ (\text{by}\ \eqref{18}) \nonumber \\ &\leqslant 2\epsilon(1+\epsilon)\left\|Y_1+I\right\|, \nonumber
          \end{align}
   and so, 
          \begin{equation}\label{ESTC}\left\|\mathcal{C}\right\|=\frac{1}{\left\|Y_1+I\right\|}\left\|X_1-X_2\right\|\leqslant 2\epsilon(1+\epsilon)\ .\end{equation} 

          \item It follows from (\ref{17}) that for each $t\in[0,1]$,
          \begin{align}
              X-\lambda I &= tA+(1-t)B-\lambda I \nonumber \\ &= t(A-B)+B-\lambda I \nonumber \\ &= t\mathcal{C}+t\gamma B+B-\lambda I \nonumber \\ &= (1+t\gamma)B+t\mathcal{C}-\lambda I \nonumber \\ &= (1+t\gamma)[B+\beta t\mathcal{C}-\lambda\beta I]\ , \nonumber 
          \end{align} 
          with $\beta=(1+t\gamma)^{-1}>0$. Then,
          \begin{equation}\label{Faltou}
            (X-\lambda I)^2=[tA+(1-t)B-\lambda I]^2=\beta^{-2}[B+\beta t\mathcal{C}-\lambda\beta I]^2.\end{equation}

        \item It follows from~\eqref{Faltou} that
          \begin{eqnarray}\label{77}
          &&\gamma \left[\left(X-\lambda I\right)^2+I\right]^{-1}B = \gamma B \left[\beta^{-2}(B+\beta t\mathcal{C}-\lambda\beta I)^2+I\right]^{-1} \nonumber \\ &=& \gamma\beta^2(B+\beta t\mathcal{C}-\lambda\beta I-\beta t\mathcal{C}+\lambda\beta I)\left[(B+\beta t\mathcal{C}-\lambda\beta I)^2+\beta^2\right]^{-1}.
          \end{eqnarray}

          Since $B+\beta t\mathcal{C}-\lambda\beta I$ is self-adjoint, it follows from the functional calculus that 
                    \begin{eqnarray*}
            \left\|(B+\beta t\mathcal{C}-\lambda\beta I)\left[(B+\beta t\mathcal{C}-\lambda\beta I)^2+\beta^2\right]^{-1}\right\|&\le& \sup_{x\in\mathbb{R}}\left|\frac{x}{x^2+\beta^2}\right|\\
            &=&\frac{1}{\beta}\sup_{x\in\mathbb{R}}\left|\frac{(x/\beta)}{(x/\beta)^2+1}\right|=\frac{1}{2\beta}\, .
          \end{eqnarray*}

          One also has, for each $t\in[0,1]$,
          \begin{eqnarray*}\left\|(-\beta t\mathcal{C}+\lambda\beta I)\left[(B+\beta t\mathcal{C}-\lambda\beta I)^2+\beta^2\right]^{-1}\right\|&\leqslant& \beta\left(\|\mathcal{C}\|+|\lambda|\right) \sup_{x\in\mathbb{R}}\left|\frac{1}{x^2+\beta^2}\right|\\
            &=&\frac{1}{\beta}\left(\|\mathcal{C}\|+|\lambda|\right).
            \end{eqnarray*}

          By combining the last two inequalities with relation (\ref{77}), one gets
            \begin{equation}\label{gammaBX}\left\|\left(\big(X-\lambda I\big)^2+I\right)^{-1}(\gamma B)\right\|\leqslant\gamma\beta\left(\|\mathcal{C}\|+|\lambda|\right)+\frac{\gamma\beta}{2}\ .\end{equation}

    \item  Furthermore,
          \begin{equation}\label{CX}\left\|\left(\big(X-\lambda I\big)^2+I\right)^{-1}\mathcal{C}\right\|\leqslant\left\|\mathcal{C}\right\|\sup_{x\in\mathbb{R}}\left|\frac{1}{x^2+1}\right|\leqslant\left\|\mathcal{C}\right\|\, .\end{equation}

        \item It follows from relations~\eqref{gammaBX} and~\eqref{CX} that 
          \begin{eqnarray*}\left\|\left(\big(X-\lambda I\big)^2+I\right)^{-1}(\mathcal{C}+\gamma B)\right\|\leqslant\|\mathcal{C}\|+\gamma\beta\left(\|\mathcal{C}\|+|\lambda|\right)+\frac{\gamma\beta}{2}\, .\end{eqnarray*}
        
      \item   It follows from the definition of $\beta$ and from (\ref{74}) that 
        \[\beta=(1+t\gamma)^{-1}\le(1-1+(1+\epsilon)^{-2})^{-1}=(1+\epsilon)^2\,,\] and that $\gamma\leqslant(1+\epsilon)^2-1=(2+\epsilon)\epsilon$. By combining these estimates with the previous relation and with~\eqref{ESTC}, one finally gets, for each $0<\epsilon<1/2$, 
          \begin{eqnarray}\label{73}
          \!\!\!\!\!\!\!\!\!\!\!\!\!\!\!\!    \left\|\left[\big(X-\lambda I\big)^2+I\right]^{-1}(\mathcal{C}+\gamma B)\right\| &\leqslant&(1+\gamma\beta)\|\mathcal{C}\|+\gamma\beta\left(\frac{1}{2}+|\lambda|\right) \nonumber \\ &\leqslant& 2\epsilon(1+(2+\epsilon)(1+\epsilon)^2\epsilon)(1+\epsilon)\nonumber\\&+&(2+\epsilon)(1+\epsilon)^2\epsilon\left(\frac{1}{2}+|\lambda|\right) \nonumber \\ &\leqslant &\left(12+6\left(\frac{1}{2}+|\lambda|\right)\right)\epsilon < 12(1+|\lambda|)\epsilon\,.
          \end{eqnarray}

\end{enumerate}

          At last, it follows from relations~\eqref{PRINCI},~(\ref{72}) and~(\ref{73}) that for each $\lambda\in S\subset (-C,C)$,
          $$\left\|g_{\lambda}'(tA+(1-t)B)(A-B)\right\|\leqslant 12\epsilon(1+C)\ ,$$
          and so 
        \begin{align}
          \left\|\omega_{F(z)}^{R}(S)-\omega_{H(z)}^{R}(S)\right\|
          &\leqslant\frac{1}{\pi}\int_S\int_0^1\left\|g_{\lambda}'\big(tA+(1-t)B\big)(A-B)\right\|\ dt\ d\lambda \nonumber \\ &\leqslant \frac{12}{\pi}\epsilon(1+C)|S|\le \frac{24}{\pi}\epsilon(1+C)C\,.\nonumber
        \end{align}

        The result now follows by setting $v(\epsilon, C):=\dfrac{24}{\pi}\epsilon(C+C^2)$.
        
    \end{proof1}

\begin{proof2}
  \normalfont
      One uses the same ideas presented in the proof of Lemma~\ref{21.}. We only present the main ideas.

      \textit{Proof of item 1.} Firstly,  note that if $F(z)=X_1+iY_1$, $H(z)=X_2+iY_2$ and $\|F(z)-H(z)\|<\epsilon$, then $\|X_1-X_2\|<\epsilon$ and $\|Y_1-Y_2\|<\epsilon$.  Let $S\in\mathcal{B}(\mathbb{R})$ be such that $|S|<\infty$. It follows for Hadamard Theorem that 
       \begin{align}\label{MolP}
         \left\|\omega_{F(z)}^{I}(S)-\omega_{H(z)}^{I}(S)\right\|&\leqslant\frac{1}{\pi}\int_S\left\|f_{\lambda}\big((\operatorname{Im}F(z))^{1/2}+I\big)-f_{\lambda}\big((\operatorname{Im}H(z))^{1/2}+I\big)\right\|\ d\lambda \nonumber \\ &\leqslant\frac{1}{\pi}\int_S\int_0^1\left\|f_{\lambda}'\big(tA+(1-t)B\big)(A-B)\right\|\ dt\ d\lambda\, ,
       \end{align} 

  with     
        $$A=[(\operatorname{Im}F(z))^{1/2}+I]\, , \qquad B=[(\operatorname{Im}H(z))^{1/2}+I]\ ,$$
       and for each $\lambda\in\mathbb{R}$,        
       $$f_{\lambda}:[1,\infty)\rightarrow [0,1),\qquad f_{\lambda}(x)=\frac{x}{x^2+\lambda^2}.$$
       
       So, one needs to estimate $\left\|f_{\lambda}'\big(tA+(1-t)B\big)(A-B)\right\|$. Note that for each $\lambda\in\mathbb{R}$,
       $$f_{\lambda}'(x)=\frac{\lambda^2-x^2}{(\lambda^2+x^2)^2},$$
       and since $x>1$, it follows that $|f'_\lambda(x)|\leqslant1$. 

       For each $t\in[0,1]$, set $X=t(A-B)+B=t(Y_1^{1/2}-Y_2^{1/2})+Y_2^{1/2}+I$. Since $A,B>1$, it follows from the functional calculus that the eigenvalues of $X$ are at least $1$; hence, 
        $$\left\|f'_\lambda(X)\right\|=\max_{\psi_i \in \sigma(X)}\left|\frac{\lambda^2-\psi_i^2}{(\lambda^2+\psi_i^2)^2}\right|\leqslant1\ .$$
        
        This results in 
        $$\left\|f_{\lambda}'\big(tA+(1-t)B\big)(A-B)\right\|\leqslant\left\|f_{\lambda}'\big(tA+(1-t)B\big)\right\|\left\|A-B\right\|\leqslant\left\|A-B\right\|\ .$$
        
        Now, it follows from Powers-Størmer inequality \footnote{
        Powers-Størmer inequality states that for any two positive matrices $Y_1, Y_2$, one has $\|Y_1^{1/2} - Y_2^{1/2}\|^2 \leqslant \|Y_1 - Y_2\|$.}
        that 
        $$\left\|A-B\right\|^2=\left\|Y_1^{1/2}-Y_2^{1/2}\right\|^2\leqslant\left\|Y_1-Y_2\right\|\ ,$$
        and so,
        $$\left\|A-B\right\| = \left\|Y_1^{1/2}-Y_2^{1/2}\right\|<\sqrt{\epsilon}\ .$$
        
        By combining the previous results, one gets for each $t\in[0,1]$,
        \begin{equation}\label{Moleza}
          \left\|f_{\lambda}'\big(tA+(1-t)B\big)(A-B)\right\| < \sqrt{\epsilon}\ .\end{equation}

        It follows from~\eqref{MolP} and~\eqref{Moleza} that 
        \begin{align}
          \left\|\omega_{F(z)}^{I}(S)-\omega_{H(z)}^{I}(S)\right\|&\leqslant\frac{1}{\pi}\int_S\left\|f_{\lambda}\big((\operatorname{Im}F(z))^{1/2}+I\big)-f_{\lambda}\big((\operatorname{Im}H(z))^{1/2}+I\big)\right\|\ d\lambda \nonumber \\ &\leqslant\frac{1}{\pi}\int_S\int_0^1\left\|f_{\lambda}'(X)(A-B)\right\|\ dt\ d\lambda \nonumber \\ &\leqslant \frac{1}{\pi}\sqrt{\epsilon}\ |S|\ . \nonumber
      \end{align}

The result now follows by setting $u(\epsilon,|S|):=\dfrac{1}{\pi}\sqrt{\epsilon}\ |S|$.

\textit{Proof of item 2.} Let $C>0$ and let $S\in\mathcal{B}((-C,C))$. If follows again from Hadamard Theorem that 
       \begin{align}
         \ \left\|\omega_{F(z)}^{R}(S)-\omega_{H(z)}^{R}(S)\right\|&
         \leqslant\frac{1}{\pi}\int_S\int_0^1\left\|g_{\lambda}'\big(tA+(1-t)B\big)(A-B)\right\|\ dt\ d\lambda\ , \nonumber
       \end{align} 
where      
       $$A=\left[\frac{1}{\|\operatorname{Im}F(z)+I\|}\operatorname{Re}F(z)\right],\ \ \ \ \  B=\left[\frac{1}{\|\operatorname{Im}H(z)+I\|}\operatorname{Re}H(z)\right]\ ,$$
       and for each $\lambda\in\mathbb{R}$,
        $$g_{\lambda}:\mathbb{R}\rightarrow [0,1),\qquad g_{\lambda}(x)=\frac{1}{(x-\lambda)^2+1}\ .$$
       
       Again, one needs to estimate $\left\|g_{\lambda}'\big(tA+(1-t)B\big)(A-B)\right\|$. In order to do this, one performs similar calculations to those presented in the proof of Lemma~\ref{21.}-2.

       Set
            \[A:=\frac{1}{\left\|Y_1+I\right\|}X_1\ \  \ \ \ \text{and} \ \ \ \ \ B:=\frac{1}{\left\|Y_2+I\right\|}X_2\ ;\]
then, one can write (recall that for each $i=1,2$, $\|Y_i+I\|=\|Y_i\|+1\geqslant 1$)
      \begin{align}\label{22}
          A-B&=\frac{1}{\left\|Y_1+I\right\|}(X_1-X_2)+\left(\frac{1}{\left\|Y_1+I\right\|}-\frac{1}{\left\|Y_2+I\right\|}\right)X_2 \nonumber \\ &=\frac{1}{\left\|Y_1+I\right\|}(X_1-X_2)+\frac{1}{\left\|Y_1+I\right\|}\left(\left\|Y_2\right\|-\left\|Y_1\right\|\right)\frac{1}{\left\|Y_2+I\right\|}X_2 \nonumber \\ &= \mathcal{C}+\gamma B,
      \end{align} 
      \noindent where 
          $$\mathcal{C}=\frac{1}{\left\|Y_1+I\right\|}(X_1-X_2) \ \ \ \text{and} \ \ \ \gamma=\frac{1}{\left\|Y_1+I\right\|}\left(\left\|Y_2\right\|-\left\|Y_1\right\|\right).$$ 

     Set, for each $t\in[0,1]$, $X:=tA+(1-t)B$ (which is a self-adjoint matrix). Now, given that for each $\lambda\in\mathbb{R}$, 
     $$g_{\lambda}'(x)=\frac{-2(x-\lambda)}{\big((x-\lambda)^2+1\big)^2}\ ,$$
     it follows from the functional calculus that
     $$g_{\lambda}'(X)=-2(X-\lambda I)\big((X-\lambda I)^2+I\big)^{-1}\big((X-\lambda I)^2+I\big)^{-1}$$      
     (note that $((X-\lambda I)^2+I)$ is positive definite, then it is invertible).

     Thus, one has
      \begin{eqnarray}\label{PRINCI0}
          \|g_{\lambda}'\big(tA+(1-t)B\big)(A-B)\|
          &\le&\left\|2(X-\lambda I)\big((X-\lambda I)^2+I\big)^{-1}\right\|\nonumber\\
          &\cdot&\left\|\big((X-\lambda I)^2+I\big)^{-1}(\mathcal{C}+\gamma B)\right\|.
\end{eqnarray}
       
 One also has, for each $\lambda\in\mathbb{R}$,         
        \begin{equation}\label{75}
            \left\|2(X-\lambda I)\big((X-\lambda I)^2+I\big)^{-1}\right\|=2\sup_{\psi_i \in \sigma(X)}\left|\frac{\psi_i-\lambda}{(\psi_i-\lambda)^2+1}\right|\leqslant 2.
        \end{equation} 
           
        It remains to estimate $\left\|\big((X-\lambda I)^2+I\big)^{-1}(\mathcal{C}+\gamma B)\right\|$.

        \begin{enumerate}
          \item Firstly, note that 
            \begin{equation}\label{23O}\|\mathcal{C}\|\leqslant \frac{\|X_1-X_2\|}{\Vert Y_1+ I\Vert}<\|X_1-X_2\|<\epsilon\ ,\end{equation}
            and that
      \begin{equation}\label{23}
          |\gamma|\leqslant \frac{\left|\|Y_2\|-\|Y_1\|\right|}{\Vert Y_1+ I\Vert}\leqslant \|Y_2-Y_1\|<\epsilon\ .
      \end{equation}

        \item By following items 3. to 6. in the proof of Lemma~\ref{21.}-2, one gets 
      $$\left\|\left[\big(tA+(1-t)B-\lambda I\big)^2+I\right]^{-1}.\ (A-B)\right\|\leqslant\|\mathcal{C}\|+\gamma\beta\left(\|\mathcal{C}\|+|\lambda|\right)+\frac{\gamma\beta}{2}\ ,$$
          where $\beta=(1+t\gamma)^{-1}$.
          
\item It follows from the definition of $\beta$ and from~\eqref{23} that 
  \[\beta=(1+t\gamma)^{-1}\le(1-\epsilon)^{-1}.\] 
  By combining this estimate with the previous relation and with~\eqref{23O}, one gets for each $0<\epsilon<1/2$, 
          \begin{eqnarray}\label{76}
            \!\!\!\!\!\!\!\!\!\!\!\!\!\!\!\!    \left\|\left[\big(X-\lambda I\big)^2+I\right]^{-1}(\mathcal{C}+\gamma B)\right\| &\leqslant&(1+\gamma\beta)\|\mathcal{C}\|+\gamma\beta\left(\frac{1}{2}+|\lambda|\right) \nonumber \\ &\leqslant& \epsilon(1+2\epsilon)+2\epsilon\left(\frac{1}{2}+|\lambda|\right) \nonumber \\ &\leqslant &
            4(1+|\lambda|)\epsilon\,.
          \end{eqnarray}

\end{enumerate}

                At last, it follows from relations~\eqref{PRINCI0},~(\ref{75}) and~(\ref{76}) that for each $\lambda\in S\subset (-C,C)$,
          $$\left\|g_{\lambda}'(tA+(1-t)B)(A-B)\right\|\leqslant 4\epsilon(1+C)\ ,$$
          and so 
        \begin{align}
          \left\|\omega_{F(z)}^{R}(S)-\omega_{H(z)}^{R}(S)\right\|
          &\leqslant\frac{1}{\pi}\int_S\int_0^1\left\|g_{\lambda}'\big(tA+(1-t)B\big)(A-B)\right\|\ dt\ d\lambda \nonumber \\ &\leqslant \frac{4}{\pi}\epsilon(1+C)|S|\le \frac{8}{\pi}\epsilon(1+C)C\,.\nonumber
        \end{align}

        The result now follows by setting $v(\epsilon, C):=\dfrac{8}{\pi}\epsilon(C+C^2)$.
    \end{proof2}

	\clearpage
	\phantomsection
	\addcontentsline{toc}{section}{References}

\renewcommand{\refname}{References}

\end{document}